\documentclass[12pt]{amsart}
\usepackage{extsizes}
\usepackage[applemac]{inputenc}
\usepackage[titletoc]{appendix}
\usepackage{pxfonts}

\usepackage[backend=biber, style=numeric, date=year, url=false, doi=false, isbn=false, eprint=true, firstinits=true, maxcitenames=5, maxbibnames=5]{biblatex}
\AtEveryBibitem{%
  \ifentrytype{online}{%
  }{%
    \clearfield{eprint}%
    \clearfield{urldate}%
  }%
}

\makeatletter
\def\blx@maxline{77}
\makeatother

\usepackage{blindtext}
\usepackage{amssymb}
\usepackage{amsfonts}
\usepackage{amsthm}
\usepackage{tikz}
\usepackage[margin=1.0in]{geometry}
\usepackage{caption}
\usepackage{color}
\usepackage{graphicx}
\usepackage{xcolor}
\usepackage{shadethm}
\usepackage{subfigure}
\usepackage[mathcal]{euscript}
\usepackage{epigraph}
\usepackage[T1]{fontenc}

\usepackage{tgpagella}

\usepackage{enumerate}
\usepackage{verbatim}
\usetikzlibrary{trees}
\usetikzlibrary{decorations.markings}
\usetikzlibrary{arrows}
\usepackage{mathrsfs}
\usepackage{xparse}
\usepackage{tikz-cd}

\usetikzlibrary{positioning,arrows,patterns}
\usetikzlibrary{decorations.markings}
\usetikzlibrary{calc}
\tikzset{
  photon/.style={decorate, decoration={snake}, draw=black},
  fermion/.style={draw=black, postaction={decorate},decoration={markings,mark=at position .55 with {\arrow{>}}}},
  vertex/.style={draw,shape=circle,fill=black,minimum size=5pt,inner sep=0pt},
particle/.style={thick,draw=black},
particle2/.style={thick,draw=blue},
avector/.style={thick,draw=black, postaction={decorate},
    decoration={markings,mark=at position 1 with {\arrow[black]{triangle 45}}}},
gluon/.style={decorate, draw=black,
    decoration={coil,aspect=0}}
 }
\NewDocumentCommand\semiloop{O{black}mmmO{}O{above}}
{%
\draw[#1] let \p1 = ($(#3)-(#2)$) in (#3) arc (#4:({#4+180}):({0.5*veclen(\x1,\y1)})node[midway, #6] {#5};)
}

\usepackage{url}
\usepackage{hyperref}
\hypersetup{colorlinks=true, linktocpage, linkcolor=blue, urlcolor=cyan,citecolor=red}

\newtheoremstyle{own}
  {3pt}
  {3pt}
  {\sffamily}
  {0pt}
  {\bfseries}
  {.}
  {5pt plus 1pt minus 1pt}
  {}

\theoremstyle{plain}
\newtheorem{thm}{Theorem}[subsection]

\theoremstyle{definition}
\newtheorem{defn}{Definition}[subsection]
\newtheorem{prop}{Proposition}[subsection]
\newtheorem{lem}{Lemma}[subsection]
\newtheorem{cor}{Corollary}[subsection]

\theoremstyle{definition}
\newtheorem{exe}{Exercise}[subsection]
\newtheorem{ex}{Example}[subsection]

\theoremstyle{remark}
\newtheorem{rem}{Remark}[subsection]

\newtheorem{notation}{Notation}[subsection]

\newcommand{\R}{\mathbb{R}}

\newcommand{\N}{\mathbb{N}}

\newcommand{\E}{\mathbb{E}}

\newcommand{\ndash}{\nobreakdash-\hspace{0pt}}

\newcommand{\ii}{{\mathrm{i}}}
\newcommand{\dd}{{\mathrm{d}}} 
\newcommand{\DD}{{\mathrm{D}}} 
\newcommand{\testleftlong}{\longleftarrow\!\shortmid}

\DeclareMathOperator{\Mat}{Mat}
\DeclareMathOperator{\GL}{GL}

\DeclareMathOperator{\Path}{\text{Path}}
\DeclareMathOperator{\Sym}{\text{Sym}}

\DeclareMathOperator{\ev}{ev}
\newcommand{\id}{\mathrm{id}}

\DeclareMathOperator{\gh}{gh}

\DeclareMathOperator{\Ber}{Ber}

\DeclareMathOperator{\Cob}{\bf Cob}
\DeclareMathOperator{\Vect}{\bf Vect}

\DeclareMathOperator{\Hom}{Hom}

\DeclareMathOperator{\Dens}{Dens}

\newcommand{\LL}{\mathsf{L}}

\newcommand{\bbZ}{{\mathbb{Z}}}

\newcommand{\de}{\partial}

\newcommand{\calA}{\mathcal{A}}
\newcommand{\calB}{\mathcal{B}}
\newcommand{\calH}{\mathcal{H}}
\newcommand{\calS}{\mathcal{S}}
\newcommand{\calC}{\mathcal{C}}

\newcommand{\calG}{\mathcal{G}}

\newcommand{\calO}{\mathcal{O}}
\newcommand{\calL}{\mathcal{L}}
\newcommand{\calM}{\mathcal{M}}

\newcommand{\calN}{\mathcal{N}}
\newcommand{\sfL}{\mathsf{L}}

\newcommand{\calE}{\mathcal{E}}
\newcommand{\calP}{\mathcal{P}}
\newcommand{\calF}{\mathcal{F}}

\newcommand{\btpsi}{\boldsymbol{\widetilde{\psi}}}

\newcommand{\sfe}{{\mathsf{e}}}

\newcommand{\frg}{{\mathfrak{g}}}

\newcommand{\sfx}{{\mathsf{x}}}

\def\gpd{\,\lower1pt\hbox{$\longrightarrow$}\hskip-.24in\raise2pt
               \hbox{$\longrightarrow$}\,}

\let\Tilde=\widetilde
\let\Bar=\overline

\let\Hat=\widehat

\newcommand{\hateta}{\widehat{\boldsymbol\eta}}
\newcommand{\hatX}{\widehat{\mathsf{X}}}

\DeclareMathOperator{\dr}{d}
\DeclareMathOperator{\Map}{Map}

\newcommand{\I}{\mathrm{i}}

\newcommand{\ee}{\textnormal{e}}

\newcommand{\calV}{\mathcal{V}}

\newcommand{\Z}{{\mathbb{Z}}}
\newcommand{\calY}{\mathcal{Y}}

\makeatletter

\pgfdeclareshape{genus}{
 \anchor{center}{\pgfpointorigin}
\backgroundpath{
    \begingroup
    \tikz@mode
    \iftikz@mode@fill

         \pgfpathmoveto{\pgfqpoint{-0.78cm}{-.17cm}}
     \pgfpathcurveto %
            {\pgfpoint{-0.35cm}{-.44cm}}
        {\pgfpoint{0.35cm}{-.44cm}}
        {\pgfpoint{.78cm}{-0.17cm}} 
     \pgfpathmoveto{\pgfqpoint{-0.78cm}{-0.17cm}}
     \pgfpathcurveto %
            {\pgfpoint{-0.25cm}{.25cm}}
        {\pgfpoint{.25cm}{.25cm}}
        {\pgfpoint{0.78cm}{-0.17cm}}
        \pgfusepath{fill}
        \fi

        \iftikz@mode@draw
     \pgfpathmoveto{\pgfqpoint{-1cm}{0cm}}
     \pgfpathcurveto %
            {\pgfpoint{-0.5cm}{-.5cm}}
        {\pgfpoint{0.5cm}{-.5cm}}
        {\pgfpoint{1cm}{0cm}}

     \pgfpathmoveto{\pgfqpoint{-0.75cm}{-0.15cm}}
     \pgfpathcurveto %
            {\pgfpoint{-0.25cm}{.25cm}}
        {\pgfpoint{.25cm}{.25cm}}
        {\pgfpoint{0.75cm}{-0.15cm}}
              \pgfusepath{stroke}
        \fi
        \endgroup
    }
    }

    \makeatother

\begin{document}

\title[Introduction to the BV-BFV formalism]{Introduction to the BV-BFV formalism}
\author[A. S. Cattaneo]{Alberto S. Cattaneo}
\author[N. Moshayedi]{Nima Moshayedi}
\address{Institut f\"ur Mathematik\\ Universit\"at Z\"urich\\ 
Winterthurerstrasse 190
CH-8057 Z\"urich}
\email[A.~S.~Cattaneo]{cattaneo@math.uzh.ch}
\address{Institut f\"ur Mathematik\\ Universit\"at Z\"urich\\ 
Winterthurerstrasse 190
CH-8057 Z\"urich}
\email[N.~Moshayedi]{nima.moshayedi@math.uzh.ch}

\maketitle

\emph{
\begin{verse}
\begin{flushright}
{\hypersetup{hidelinks}
\ldots then like  \href{https://youtu.be/BhPtkIMEnjk}{hedgehogs} which\\
Lie tumbling in my barefoot way, and mount\\
Their pricks at my footfall\ldots\\
\textsc{W. Shakespeare,}  \normalfont{The Tempest, Act II, Scene 2}
}
\end{flushright}
\end{verse}}

\begin{abstract}
These notes give an introduction to the mathematical framework of the Batalin--Vilkovisky and Batalin--Fradkin--Vilkovisky formalisms. Some of the presented content was given as a mini course by the first author at the 2018 QSPACE conference in Benasque.
\end{abstract}

\tableofcontents

\begin{changemargin}{1.8cm}{1.8cm}
\emph{
\\
Theoria quantica camporum iuxta Feynmanum per integralem functionalem definiri potest. Extantibus symmetriis degenerata integratio emendanda est. Cohomologica methodus a Batalino Vilkoviskyoque inventa au\-xilium fert.
\\
Profecto, regularitatis assumptionibus fictis, spatium camporum impari symplecticae varietati secumque commutanti functionali classica actio extendi possunt. Quae solutio aequationis cardinalis classicae vocatur. 
\\
Quanticae aequationi tamen mensurae compatibilis opus est, quae operatorem laplacianum batalinovilkoviskyanum iniicit. Si quanticae aequationis cardinalis in lagrangiana subvarietate solutio integratur, per deformationes invariantem theorema egregium batalinovilkoviskyanum praedicat proventum. Dissipata degeneratio non dissipata symmetria.
\\
\\
\textsc{C. Ascanius Erinaceatus}, \normalfont{De aequatione cardinali batalinovilkoviskyana}
}
\end{changemargin}

\section{Introduction}
The concept of gauge theory has become an important concept in theoretical and mathematical physics over the last decades. One can actually derive important dynamical aspects by considering symmetry arguments, comig from the concept of gauge invariance. There are several mathematical constructions for the quantization of field theories, such as the Feynman path integral approach (\cite{Feynman1942,Feynman1949,Feynman1950,FH}), which sums up the amplitudes associated to different paths connecting a given initial and final state to get the transition amplitude. The advantage of the path integral approach is that all symmetries of the given theory are explicitly (though formally) preserved. Even though the path integral gives a nice way of a field quantization, it is mathematically not rigorously defined. However, one can still generate, in an unambiguous way, a perturbative expansion around some classical field configuration in terms of so-called ``Feynman diagrams'' (see e.g. \cite{P} for a mathematical description). This type of representation can only be defined for non-gauge theories, since it crucially relies on regularity conditions that exclude the possibility of gauge invariance. Nevertheless, there exists a way out of this problem. One approach to do so was given by Faddeev and Popov in $1967$ by introducing the concept of what is called ``\emph{ghost fields}''. In this way the local symmetry of the gauge theory is replaced by a global symmetry in an extended, graded configuration space, called BRST symmetry (\cite{BRS1,BRS2,BRS3,Tyutin1976}). The resulting process is called the quantization procedure in the BRST formalism.

\vspace{0.5cm}

The focus of these notes lies in a special formalism of gauge theory developed by Batalin and Vilkovisky during the 1970's and 1980's (\cite{BV1,BV2,BV3}); see also \cite{S,Henneaux1994,Anselmi1,Catt1,Cqr,Cost,KR,Mn2} and references therein. A similar approach was developed by Batalin, Fradkin and Vilkovisky to path integrals for theories with gauge freedom in the Hamiltonian approach, whereas the formalism provided by Batalin and Vilkovisky was given in the Lagrangian setting using a Lagrangian modified by Faddeev--Popov ghost and their BRST variations, which contribute to the path integral a factor which nicely divides out the gauge redundancy. This formalism is known today as the BV formalism. The methods of gauge fixing was extended by Batalin, Fradkin, Fradkina and Vilkovisky in a series of papers \cite{BF1,BF2,FV1,FV2}. The Hamiltonian version is known today as the BFV formalism (see \cite{BF1,BF2,FV1,FV2} and also \cite{Sta,Sch}). The BV and BFV formalisms have been developed through time by the work of many different people. In \cite{CMR1} Cattaneo, Mnev and Reshetikhin studied for the first time the classical BV framework for gauge theories on spacetime manifolds with boundary and extended everything to manifolds with corners. Recently, in \cite{CMR2,CMR3}, they introduced a general perturbative quantization scheme for gauge theories on manifolds with boundary, compatible with cutting and gluing, in the cohomological symplectic formalism, which we call the BV-BFV formalism. Different types of theories are already explicitly described for this procedure, such as e.g. abelian $BF$ theory (\cite{GH,S,Mn}), different AKSZ theories (\cite{AKSZ}) (e.g. Chern--Simons theory with a splitting of the Lie algebra (\cite{AS,AS2,Witten1989}), the Poisson Sigma Model (\cite{CF1,CF2,CF4,CF5}), or non-abelian $BF$ theories (\cite{GH,Mn,CMR2})), and $2D$ Yang--Mills theory (\cite{IM}). 

\subsection*{Acknowledgements}
This research was (partly) supported by the NCCR SwissMAP, funded by the Swiss National Science Foundation, and by the
COST Action MP1405 QSPACE, supported by COST (European Cooperation in Science and Technology). We acknowledge partial support of SNF grant No. 200020\_172498/1. N. M. acknowledges partial support by the Forschungskredit of the University of Zurich, grant no. FK-18-095.

\vspace{0.2cm}

We thank Ivan Contreras, Pavel Mnev and Nicolai Reshetikhin for useful comments.

\section{Motivation}

\subsection{Classical Mechanics}

In classical mechanics we consider an action functional\footnote{In the physics literature, it is common to denote the time-derivatives by ``dots'', i.e. $\frac{\dd}{\dd t}q(t)=\dot{q}(t)$.} $$S(q)=\int_{t_0}^{t_1} L(q(t),\dot{q}(t))\dd t,$$ 

\begin{center}
\begin{figure}[h!]

\begingroup%
  \makeatletter%
  \providecommand\color[2][]{%
    \errmessage{(Inkscape) Color is used for the text in Inkscape, but the package 'color.sty' is not loaded}%
    \renewcommand\color[2][]{}%
  }%
  \providecommand\transparent[1]{%
    \errmessage{(Inkscape) Transparency is used (non-zero) for the text in Inkscape, but the package 'transparent.sty' is not loaded}%
    \renewcommand\transparent[1]{}%
  }%
  \providecommand\rotatebox[2]{#2}%
  \ifx\svgwidth\undefined%
    \setlength{\unitlength}{114.30579976bp}%
    \ifx\svgscale\undefined%
      \relax%
    \else%
      \setlength{\unitlength}{\unitlength * \real{\svgscale}}%
    \fi%
  \else%
    \setlength{\unitlength}{\svgwidth}%
  \fi%
  \global\let\svgwidth\undefined%
  \global\let\svgscale\undefined%
  \makeatother%
  \begin{picture}(1,0.67732551)%
    \put(0,0){\includegraphics[width=\unitlength]{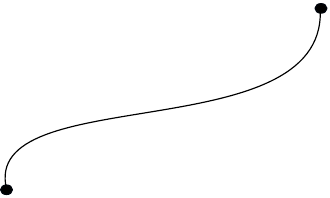}}%
    \put(0.56948513,0.21842524){\color[rgb]{0,0,0}\makebox(0,0)[lb]{\smash{$q$}}}%
    \put(0.09700054,0.01566911){\color[rgb]{0,0,0}\makebox(0,0)[lb]{\smash{$q(t_0)$}}}%
    \put(0.73540662,0.62451362){\color[rgb]{0,0,0}\makebox(0,0)[lb]{\smash{$q(t_1)$}}}%
  \end{picture}%
\endgroup%

\caption{The path of least action, i.e. the solution to $\delta S=0$, between two points $x=q(t_0)$ and $y=q(t_1)$ in space-time.}
\label{least_action_path}
\end{figure}
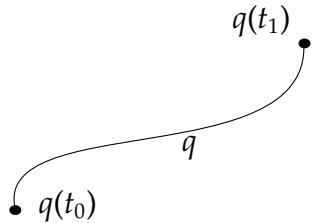
\end{center}

where $L(q,\dot{q})=\frac{1}{2}m\|\dot{q}\|^2-V(q)$ is called the \emph{Lagrangian} function of the paths $q\colon [t_0,t_1]\to\R^n$ with some function $V\in C^\infty(\R^n)$ depending on $q$, called the \emph{potential energy}. We denote by $\Path_{(x,y)}^{[t_0,t_1]}(\R^n)$ the space of all such paths with $q(t_0)=x$ and $q(t_1)=y$. By considering the methods of variational calculus, one can show that the solutions of the equation $\delta S=0$ for fixed endpoints (i.e. the extremal points of $S$) give us the classical trajectory of the particle with mass $m\in\R^+$. The equations following from $\delta S=0$ are called the Euler--Lagrange equations (EL), and they are exactly the equations of motion obtained from Newtonian mechanics. Netwon's equations of motion appear from the law $F=ma(t)=m\ddot{q}(t)$ (read it ``force equals mass times acceleration''). To see this, we recall that the \emph{momentum} in physics is given by $p=mv$, where $v$ denotes the velocity of the particle with mass $m$. Then, by the fact that $v=\dot{q}$, one considers the coordinates $\dot{q}=\frac{p}{m}$ and $\dot{p}=-\nabla V$, where $\nabla$ denotes the gradient operator. The \emph{Hamiltonian} approach considers the space with these coordinates to be the classical phase space (classical space of \emph{states}) given by $T^*\R^n\ni(q,p)$ endowed with a \emph{symplectic form}\footnote{we will not always write $\land$ between forms but secretly always mean the exterior product between them, i.e. for two differential forms $\alpha,\beta$, we have $\alpha\beta=\alpha\land \beta$.} (see Subsection \ref{sympl_mnf}) given by $$\omega=\sum_{i=1}^n\dd q^{i}\dd p_i.$$ Moreover, one considers a \emph{total energy function} (or a \emph{Hamiltonian function}) $H(q,p)=\frac{\|p\|^2}{2m}+V$, where $V$ is again a potential energy function. In the physics literature, the first term of $H$ is called the \emph{kinetic energy}. This function is said to be \emph{Hamiltonian} if there is a vector field $X_H$ such that $$\iota_{X_H}\omega=-\dd H,$$ where $\iota$ denotes the \emph{contraction} map (also called \emph{interior derivative}). The vector field $X_H$ is called the \emph{Hamiltonian vector field} of $H$. In the case at hand, since $\omega$ is nondegenerate, every function is Hamiltonian and its Hamiltonian vector field is uniquely determined. For $H$ being the total energy function and the canonical symplectic form on the cotangent space, we get the following Hamiltonian vector field: A vector field on $T^*\R^n$ has the form general form $X=X^{i}\partial_{q^{i}}+X_i\partial_{p_i}$. Thus, applying the equation for being the Hamiltonian vector field of $H$ we get $-\dd H=X_i\dd q^{i}+X^{i}\dd p_i=\iota_X\omega$. Now since $\dd H=\partial_i V\dd q^{i}+\frac{p_i}{m}$, we get the coefficients of the vector field to be $X_i=-\partial_i V$ and $X^{i}=\frac{p_i}{m}$. Hence, we get the Hamiltonian vector field 
$$X_H=-\partial_i V\partial_{q_i}+\frac{p_i}{m}\partial_{p_i}.$$
Naturally, $X_H$ induces a \emph{Hamiltonian flow} $T^*\R^n\to T^*\R^n$.

\subsection{Quantization}
An approach of quantization of the above is to associate to $T^*\R^n$ the space of square integrable functions $L^2(\R^n)$ on $\R^n$. The Hamiltonian flow can then be replaced by a linear map $$\ee^{\frac{\I}{\hbar}\widehat{H}}\colon L^2(\R^n)\to L^2(\R^n),$$ where $\widehat{H}:=-\frac{\hbar^2}{2m}\Delta+V$ denotes the \emph{Hamilton operator}, which is the canonical quantization of the classical Hamiltonian function, where $\Delta=\sum_{1\leq j\leq n}(\partial_{x^j})^2$ denotes the Laplacian. Note that the space of states is now given by a Hilbert space $\calH_0$ and the observables as operators on $\calH_0$. One can show that the action of this operator can be expressed as an integral of the form 
\[
\left(\ee^{\frac{\I}{\hbar}\widehat{H}}\psi\right)(x)=\int K(x,y)\psi(y)\dd y,
\]
for $\psi\in \calH_0$, where $K$ denotes the \emph{integral kernel} for the operator.  Feynman showed in \cite{Feynman1942} that this kernel (quantum mechanical \emph{propagator}) can be seen as a \emph{path integral}, which is given by 
\[
K(x,y)=\int_{\Path^{[t_0,t_1]}_{(x,y)}(\R^n)}\ee^{\frac{\I}{\hbar}S(q)}\mathscr{D}q.
\]
where $S$ denotes the action of the classical system and $\mathscr{D}$ a measure on the path space. 

Since $\mathscr{D}$ is suppose to be a ``measure'' on an infinite-dimensional space, it is mathematically ill-defined. However, one can still make sense of such an integral in several ways; one of them is by considering its \emph{perturbative expansion} in formal power series with \emph{Feynman diagrams} as coefficients. This procedure is mathematically well-defined.

\begin{center}
\begin{figure}[h!]

\begingroup%
  \makeatletter%
  \providecommand\color[2][]{%
    \errmessage{(Inkscape) Color is used for the text in Inkscape, but the package 'color.sty' is not loaded}%
    \renewcommand\color[2][]{}%
  }%
  \providecommand\transparent[1]{%
    \errmessage{(Inkscape) Transparency is used (non-zero) for the text in Inkscape, but the package 'transparent.sty' is not loaded}%
    \renewcommand\transparent[1]{}%
  }%
  \providecommand\rotatebox[2]{#2}%
  \ifx\svgwidth\undefined%
    \setlength{\unitlength}{131.29432958bp}%
    \ifx\svgscale\undefined%
      \relax%
    \else%
      \setlength{\unitlength}{\unitlength * \real{\svgscale}}%
    \fi%
  \else%
    \setlength{\unitlength}{\svgwidth}%
  \fi%
  \global\let\svgwidth\undefined%
  \global\let\svgscale\undefined%
  \makeatother%
  \begin{picture}(1,1.01600495)%
    \put(0,0){\includegraphics[width=\unitlength]{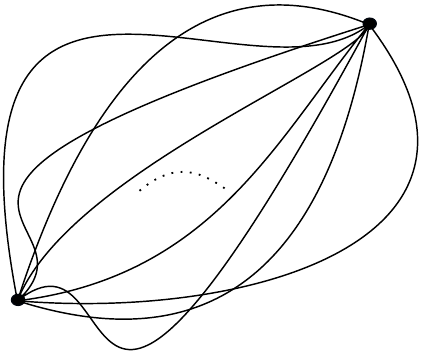}}%
    \put(-0.00369337,0.01329828){\color[rgb]{0,0,0}\makebox(0,0)[lb]{\smash{$q(t_0)$}}}%
    \put(0.81556436,0.97862504){\color[rgb]{0,0,0}\makebox(0,0)[lb]{\smash{$q(t_1)$}}}%
  \end{picture}%
\endgroup%

\caption{Illustration of the fact that all the paths between $x=q(t_0)$ and $y=q(t_1)$ are taken into account.}
\label{QM_paths}
\end{figure}
\end{center}

\subsubsection{Idea of the BV-BFV formalism}
\begin{enumerate}
\item{The quantum mechanical formalism is $1$-dimensional, whereas we would like to consider higher dimensional theories.
}
\item{We would like to take \emph{gauge symmetries} into account. This can be done by perturbation theory for a given local action $S(\phi)=\int_ML(\phi,\partial\phi,...)$, where now $L$ denotes a \emph{Lagrangian density} depending on the field $\phi$ (depending on elements in $M$) and higher derivatives. The idea is to look at the path integral $\int_{\calL\subset \calF_M} \ee^{\frac{\I}{\hbar}S(\phi)}\mathscr{D}\phi$, where $\calL\subset \calF_M\ni \phi$ denotes a subspace of the space of fields, perturb around a classical solution (formal power series in terms of Feynman graphs), and take gauge symmetries into account ($\infty$-dimensional kernel).  
}
\end{enumerate}

This procedure is called \emph{gauge fixing} (we will see that this corresponds to choosing a Lagrangian submanifold $\calL$). Another (older) approach was by the method of Faddeev--Popov ghosts, or the BRST formalism. The approach of Batalin--Vilkovisky considers the machinery of path integrals, which gives some advantages in dealing with gauge theories. The construction for manifolds with boundary uses in the quantization the approach of complexes whose differential comes from the BV construction, i.e. the quantum Hilbert space $\calH=\bigoplus_{j}\calH_j$ is given as a complex and its cohomology in degree zero is given by the original Hilbert space $\calH_0$, i.e. $H^0(\calH)=\calH_0$. 


\section{Elements of Symplectic Geometry}

In this section we want to recall some standard definitions and properties of symplectic geometry. More on symplectic geometry can be found e.g. in \cite{daSilva01}.  We will start with the local picture and continue with the global structure.

\subsection{Symplectic vector spaces} Let $V$ be a finite-dimensional vector space over $\mathbb{K}=\R$ or $\mathbb{C}$. Denote by $V^*$ the dual of $V$. An element of $V^*$ is a $\mathbb{K}$-linear map $f\colon V\to \mathbb{K}$. Let $0\leq m\leq \dim V$. Define 
\begin{multline*}
\bigwedge^mV^*:=\Big\{\phi\colon \overbrace{V\times\dotsm \times V}^{m}\to \mathbb{K}\mid \text{ $\phi$ is linear in each argument and $\phi$ is alternating,}\\
\text{ i.e. $\phi(v_1,...,v_j,v_{j+1},...,v_m)=-\phi(v_1,...,v_{j+1},v_j,...,v_m)$ for all $j=1,2,...,m-1$}\Big\}
\end{multline*}

\begin{ex}
Let $f,g\in V^*$. Then we can define $(f\land g)\in \bigwedge^2 V^*$ by 
\[
(f\land g)(v_1,v_2)=f(v_1)g(v_2)-f(v_2)g(v_1).
\]
In fact, it can be shown that all the elements of $\bigwedge^2V^*$ are finite linear combinations of such elements. Given $\omega\in \bigwedge^2 V^*$, we can define a map
\begin{align*}
\omega^\flat\colon V&\to V^*\\
v&\mapsto \omega^\flat(v),
\end{align*}
where $\omega^\flat(v)(w):=\omega(v,w)$.
\end{ex}

\begin{defn}[Symplectic vector space]
A \emph{symplectic vector space} is a pair $(V,\omega)$, where $V$ is a (finite-dimensional) vector space and $\omega\in\bigwedge^2 V^*$ such that $\omega^\flat$ is a vector space isomorphism.
\end{defn}

\begin{rem}
Since we are in the finite-dimensional setting, $\omega^\flat$ is a vector space isomorphism if and only $\omega^\flat$ is injective.
\end{rem}

\begin{exe}
\label{ex1}
Show that $\omega^\flat$ is injective if and only if $\omega(v,w)=0$ for all $w\in V$ implies $v=0$.
\end{exe}

\begin{ex}
Let $(W,\langle\enspace,\enspace\rangle)$ be an inner product space. Consider $V=W\oplus W$, with $\omega((w_1,w_2),(w_1',w_2'))=\langle w_2',w_1\rangle-\langle w_2,w_1'\rangle$. Then $(V,\omega)$ is a real symplectic vector space. More generally, $V=W\oplus W^*$ and $\omega_{can}((w,\alpha),(w',\alpha'))=\alpha'(w)-\alpha(w')$, then $(V,\omega_{can})$ is a symplectic vector space.
\end{ex}

\begin{exe}
Show that $(V,\omega_{can})$ is a symplectic vector space. Use exercise \ref{ex1}.
\end{exe}

\begin{defn}[Isotropic/Coisotropic/Lagrangian]
\label{lag_subsp}
Let $(V,\omega)$ be a symplectic vector space. Let $Y$ be a subspace of $V$. Define the symplectic orthogonal complement of $Y$ by $Y^{\perp}:=\{v\in V\mid \omega(v,y)=0,\forall y\in Y\}$. Then 
\begin{itemize}
\item{$Y$ is \emph{isotropic} if $Y\subseteq Y^\perp$,
}
\item{$Y$ is \emph{coisotropic} if $Y^\perp\subseteq Y$,}
\item{$Y$ is \emph{Lagrangian} if $Y$ is isotropic and $Y$ is symplectic if $\omega\vert_{Y\times Y}$ is nondegenerate, i.e. $Y\cap Y^\perp=\{0\}$.}
\end{itemize}
\end{defn}

\begin{ex}
If $\dim Y=1$, then $Y$ is isotropic. If $Y$ is isotropic, then $Y^\perp$ is coisotropic. If $Y$ is symplectic, then so is $Y^\perp$. Moreover, $Y^{\perp\perp}=(Y^\perp)^\perp=Y$.
\end{ex}

\begin{prop}
Let $(V,\omega)$ be a symplectic vector space. Then there is a basis $\{e_1,...,e_n,f_1,...,f_n\}$ of $V$ such that 
\begin{align*}
\omega(e_i,e_j)&=0,\\
\omega(f_i,f_j)&=0,\\
\omega(e_i,f_j)&=\delta_{ij}.
\end{align*}
for all $i,j\in\{1,2,...,n\}$. Hence, we can write $\omega=\sum_{j=1}^ne_j^*\land f_j^*$.
\end{prop}

\begin{exe}
Show that $\dim V=\dim Y+\dim Y^\perp$. Moreover, show $Y$ is Lagrangian if and only if $V\cong Y\oplus Y^\perp$.
\end{exe}

\begin{exe}
\label{lag_ex}
Show that $Y$ is a Lagrangian subspace if and only if $Y$ is isotropic and $\dim Y=\frac{1}{2}\dim V$. Moreover, show that $Y$ is Lagrangian if and only if $Y$ is a maximal isotropic subspace.
\end{exe}

\subsection{Symplectic manifolds}
\label{sympl_mnf}
To understand the mathematical structure of classical mechanics, it is necessary to understand the notion of a symplectic manifold. 

\begin{defn}[Closed/exact]
We call a $k$-form $\omega$ \emph{closed}, if $\dd\omega=0$. It is called \emph{exact} if there is a $(k-1)$-form $\alpha$ such that $\dd\alpha=\omega$.
\end{defn}

\begin{ex}
If $\omega$ is exact, then $\dd\omega=0$, i.e. exact forms are closed as well. Let $M=\R^n$, then $\omega$ is closed if and only if $\omega$ is exact (this is given by the ``Poincar\'e lemma'') for $k>0$.
\end{ex}

\begin{defn}[Symplectic manifold]
A \emph{symplectic manifold} is a pair $(M,\omega)$, where $M$ is a smooth manifold and $\omega$ is a $2$-form on $M$ such that
\begin{enumerate}
\item{$\omega$ is closed, i.e. $\dd\omega=0$,}
\item{$\omega$ is nondegenerate, i.e. for all $q\in M$, $\omega^\flat\colon T_qM\to T_q^*M$ is injective.}
\end{enumerate}
\end{defn}

\begin{defn}[Tautological $1$-form]
Let $M=T^*N$. Define a $1$-form $\alpha$ on $M$ as 
$$\alpha_{x,p}(X_{x,p}):=p(\dd \pi_{x,p}X_{x,p}),$$
where $p\colon T^*N\to N$, $\pi\colon TM\to M$ and $X_{x,p}\in T_{x,p}M$. The form $\alpha$ is called the \emph{tautological $1$-form} on $T^*N$.
\end{defn}

Let $X$ be an $n$-manifold, with $M=T^*X$ its cotangent bundle. If $x_1,...,x_n$ are coordinates on $U\subseteq X$, with associated cotangent coordinates $x_1,...,x_n,\xi_1,...,\xi_n$ on$T^*U$, then the tautological $1$-form on $T^*X$ is $\alpha=\sum_i\xi_i\dd x^{i}$ and the canonical $2$-form is
$$\omega=-\dd\alpha=\sum_i\dd x^{i}\land \dd\xi_i.$$

\begin{ex}
Let $M=T^*\R\cong \R\times \R\ni(x,p)$. Let $\alpha=f\dd x+g\dd p$. Then $\alpha(\partial_x)=f$ and $\alpha(\partial_p)=0$, thus $\alpha=f\dd x$. on the other hand $\alpha_{x,p}(\partial_x)=p$ and hence $\alpha=p\dd x$. More generally, if $M=T^*\R^n\ni(x^{1},...,x^n,p_1,...,p_n)$, then $\alpha=\sum_{1\leq j\leq n}p_j\dd x^j$
\end{ex}

\begin{exe}
Let $(U,\phi)$ be a local coordinate system on $M=T^*N$ given by $$\phi(q)=(x^1(q),...,x^n(q),p_1(q),...,p_n(q)),$$ Show that $\alpha=\sum_{1\leq j\leq n}p_j\dd x^j$. Moreover, show that $(T^*N,\omega=\dd\alpha)$ is a symplectic manifold.
\end{exe}

\subsection{The Lie derivative}
The concept of a derivative can be generalized in a tensorial sense by the following definiton:

\label{Lie_der}
\begin{defn}[Lie derivative]
Let $f\in C^\infty(M)$ and $X$ be a vector field. The \emph{Lie derivative} of $f$ along $X$ is defined as $\mathsf{L}_Xf=X(f)$. Let $X$ and $Y$ be two vector fields. Then we define $\mathsf{L}_XY=[X,Y]$. Moreover, let $X$ be a vector field and $\alpha$ a $1$-form. Then $\mathsf{L}_X\alpha$ is a $1$-form defined by the equation $$\mathsf{L}_X(\alpha(Y))=(\mathsf{L}_X\alpha)(Y)+\alpha(\mathsf{L}_XY).$$ More generally, if $\alpha$ is a $k$-form  then $\mathsf{L}_X\alpha$ is again a $k$-form defined by $$(\mathsf{L}_X\alpha)(Y_1,...,Y_k)=\mathsf{L}_X(\alpha(Y_1,...,Y_k))-\sum_{j=1}^k\alpha(Y_1,...,Y_{j-1},[X,Y_j],...,Y_k).$$
\end{defn}

\begin{rem}
Given a $k$-form $\alpha$, and a vector field $X$, $(\mathsf{L}_X\alpha)(p)$ is the rate of change of $\alpha$ in the direction of the ``flow'' of $x$ at $p$.
\end{rem}


\subsection{Lagrangian submanifolds and Conormal Bundles}

\subsubsection{Lagrangian submanifolds}
Let $(M,\omega)$ be a $2n$-dimensional symplectic manifold. 

\begin{defn}[Lagrangian submanifold]
A submanifold $Y$ of $M$ is called a \emph{Lagrangian} if, at each $p\in Y$, $T_pY$ is Lagrangian subspace of $T_pM$ (see definiton \ref{lag_subsp}). Equivalently, if $i\colon Y\hookrightarrow M$ is the inclusion, then $Y$ is Lagrangian if and only if $i^*\omega=0$ and $\dim Y=\frac{1}{2}\dim M$ (see also exercise \ref{lag_ex}).
\end{defn}

\begin{defn}[Zero section]
The \emph{zero section} of $T^*X$, defined by 
$$X_0=\{(x,\xi)\in T^*X\mid \xi=0\}$$
is an $n$-dimensional submanifold of $T^*X$.
\end{defn}

\begin{prop}
$X_0$ is a Lagrangian submanifold of $T^*X$.
\end{prop}

\begin{proof}
Clearly $\alpha$ vanishes on $X_0\cap T^*U$. In particular, if $i_0\colon X_0\hookrightarrow T^*X$ is the inclusion, we have $i_0^*\alpha=0$. Hence, $i^*_0\omega=i_0^*\dd\alpha=0$, and $X_0$ is Lagrangian.
\end{proof}

\begin{rem}
\label{iso_coiso}
Similar to Lagrangian submanifolds, we can define \emph{isotropic} and \emph{coisotropic} submanifolds, by requiring that at each point of the submanifold the tangent space is isotropic or coisotropic  (see definiton \ref{lag_subsp}) respectively.
\end{rem}

\subsubsection{Conormal bundles}
Let $S$ be any $k$-dimensional submanifold of an $n$-manifold $X$. 
\begin{defn}[Conormal space]
The \emph{conormal space} at $x\in S$ is given by 
\begin{equation}
\label{conormal_sp}
N^*_xS=\{\xi\in T_x^*X\mid \xi(v)=0, \text{ for all } v\in T_xS\}.
\end{equation}
\end{defn}

\begin{defn}[Conormal bundle]
\label{conormal_bun1}
The \emph{conormal bundle} of $S$ is given by 
\begin{equation}
\label{conormal_bun}
N^*S=\{(x,\xi)\in T^*X\mid x\in S, \xi\in N_x^*S\}.
\end{equation}
\end{defn}

\begin{rem}
The conormal bundle $N^*S$ is an $n$-dimensional submanifold of $T^*X$ (one can check that by using coordinates on $X$ adapted to $S$).
\end{rem}

\begin{prop}
Let $i\colon N^*S\hookrightarrow T^*X$ be the inclusion, and let $\alpha$ be the tautological $1$-form on $T^*X$. Then $i^*\alpha=0$.
\end{prop}

\begin{proof}
Let $(U,x_1,...,x_n)$ be a coordinate system on $X$ centered at $x\in S$ and adapted to $S$, so that $U\cap S$ is described by $x_{k+1}=\dotsm =x_n=0$. Let $(T^*U,x_1,...,x_n,\xi_1,...,\xi_n)$ be the associated cotangent coordinate system. The submanifold $N^*S\cap T^*U$ is then described by 
\begin{align*}
x_{k+1}&=\dotsm =x_n=0,\\
\xi_1&=\dotsm =\xi_k=0.
\end{align*}
Since $\alpha=\sum_i\xi_i\dd x^{i}$ on $T^*U$, we conclude that, at $p\in N^*S$, 
$$(i^*\alpha)_p=\alpha_p\vert_{T_p(N^*S)}=\sum_{i>k}\xi_i\dd x^{i}\Big\vert_{span\left(\partial_{x_i}, i\leq k\right)}=0$$ 
\end{proof}

\begin{cor}
A conormal bundle is a Lagrangian submanifold.
\end{cor}

\subsection{Hamiltonian actions}
The mathematical structure of classical mechanics, in the sense of Hamiltonian mechanics, is characterized by the symplectic structure of the phase space. The notion of a Hamiltonian action gives more insights for the symplectic structure.

\subsubsection{Momentum and Comomentum Maps}
Let $(M,\omega)$ be a symplectic manifold, and denote by $\text{Sympl}(M,\omega)\subset \text{Diff}(M)$ the set of symplectomorphisms on $M$. Moreover, ler $G$ be a Lie group, and $\Psi\colon G\to \text{Sympl}(M,\omega)$ a smooth symplectic action, i.e. a group homomorphism such that the evaluation map $\ev_\Psi(g,p):=\Psi_g(p)$ is smooth. Let us consider the case where $G=\R$. We have the following bijective correspondence:
\begin{align*}
\{\text{symplectic actions of $\R$ on $M$}\}&\longleftrightarrow \{\text{complete symplectic vector fields on $M$}\}\\
\Psi&\longmapsto X_p=\frac{\dd}{\dd t}\Psi_t(p)\\
\Psi=\exp(tX)&\testleftlong X\\
\text{``flow of $X$''}&\qquad \text{``vector field generated by $\Psi$''}
\end{align*}

\begin{defn}[Hamiltonian action (for $\R$-action)]
The action $\Psi$ is said to be \emph{Hamiltonian} if there is a function $H\colon M\to \R$ such that $-\dd H=\iota_X\omega$ where $X$ is the vector field on $M$ generated by $\Psi$, and $\iota_X$ denotes the contraction with the vector field $X$.
\end{defn}

Consider now the case where $G=S^1$. An action of $S^1$ is an action of $\R$ which is $2\pi$-periodic, i.e. $\Psi_{2\pi}=\Psi_0$. The $S^1$-action is called Hamiltonian if the underlying $\R$-action is Hamiltonian. Let us now consider the general case. Let $\mathfrak{g}$ be the Lie algebra of $G$ and denote its dual by $\mathfrak{g}^*$. 

\begin{defn}[Hamiltonian action (general)]
\label{Hamiltonian_action}
The action $\Psi$ is called \emph{Hamiltonian} if there is a map
\begin{equation} 
\label{momentmap}
\mu\colon M\to \mathfrak{g}^*,
\end{equation}
satisfying the following conditions:
\begin{enumerate}
\item{For each $X\in\mathfrak{g}$, let
\begin{itemize}
\item{
$\mu^X\colon M\to \R$, $\mu^X(p):=\langle \mu(p),X\rangle$, be the component of $\mu$ along $X$,
}
\item{$X^\sharp$ be the vector field on $M$ generated by the $1$-parameter subgroup $\{\exp(tX)\mid t\in\R\}\subseteq G$.
}
\end{itemize}
Then $\dd\mu^X=\iota_{X^\sharp}\omega$, i.e. $\mu^X$ is a Hamiltonian function for the vector field $X^\sharp$.
}
\item{$\mu$ is equivariant with respect to the given action $\Psi$ of $G$ on $M$ and the coadjoint action $\text{Ad}^*$ of $G$ on $\mathfrak{g}^*$, given by 
\begin{equation}
\mu\circ \Psi_g=\text{Ad}^*_g\circ \mu
\end{equation}
for all $g\in G$.
}
\end{enumerate}
\end{defn}

\begin{defn}[Momentum map]
The map $\mu$ given as in Definition \ref{Hamiltonian_action} is called \emph{momentum map}.
\end{defn}

\begin{defn}[Hamiltonian $G$-space]
The Quadruple $(M,\omega,G,\mu)$ given as in Definition \ref{Hamiltonian_action} is called \emph{Hamiltonian $G$-space}.
\end{defn}

\begin{defn}[Comomentum map]
We call a map 
\begin{equation}
\label{comomentmap}
\mu^*\colon \mathfrak{g}\to C^\infty(M)
\end{equation}
a \emph{comomentum map} if 
\begin{enumerate}
\item{$\mu^*(X):=\mu^X$ is a Hamiltonian function for the vector field $X^\sharp$,
}
\item{$\mu^*$ is a Lie algebra homomorphism, i.e. 
$$\mu^*([X,Y])=\{\mu^*(x),\mu^*(Y)\},$$
for all $X,Y\in\mathfrak{g}$, where $[\enspace,\enspace]$ is the Lie bracket on $\mathfrak{g}$ and $\{\enspace,\enspace\}$ is the Poisson bracket on $C^\infty(M)$.
}
\end{enumerate}
\end{defn}

\subsection{Symplectic reduction}
The concept of a symplectic reduction is important for the treatment of gauge theories on manifolds with boundary. To be able to perform \emph{geometric quantization} on the boundary, we have to make sure that the space of fields can be reduced to a symplectic space of boundary fields with certain conditions. The following theorem guarantees such a structure for the needed purposes in many examples.

\begin{thm}[Marsden-Weinstein]
\label{MW}
Let $(M,\omega,G,\mu)$ be a Hamiltonian $G$-space for a compact Lie group $G$. Let $i\colon \mu^{-1}(0)\hookrightarrow M$ be the inclusion map. Assume that $G$ acts freely on $\mu^{-1}(0)$. Then 
\begin{itemize}
\item{the orbit space $M_{red}=\mu^{-1}(0)/G$ is a manifold,
}
\item{$\pi\colon\mu^{-1}(0)\to M_{red}$ is a principal $G$-bundle, and
}
\item{there is a symplectic form $\omega_{red}$ on $M_{red}$ satisfying $i^*\omega=\pi^*\omega_{red}$.
}
\end{itemize}
\end{thm}

\begin{proof}
See e.g. \cite{daSilva01}.
\end{proof}

\begin{rem}
In field theory we work in the infinite dimensional setting, so Theorem  \ref{MW} does not apply as it is. Also, for field theory, we normally consider weak symplectic vector spaces, i.e. the corresponding linear map from $V$ to $V^*$ is injective. Morover, there is a notable example, 4D gravity, where one has to deal with general coisotropic reduction and not just Marsden--Weinstein reduction.
\end{rem}

\begin{defn}[Reduction]
The pair $(M_{red},\omega_{red})$ is called the \emph{reduction} of $(M,\omega)$ with respect to $(G,\mu)$, or the \emph{reduced space}, or the \emph{symplectic quotient}, or the \emph{Marsden-Weinstein quotient}.
\end{defn}

\section{Elements of Supergeometry}
Supergeometry is an essential tool to understand reduction in cohomological terms\footnote{The concept of Supergeometry is also used for SUSY field theories, but this is not what we discuss here.}. In particular the Faddeev--Popov, BRST, or Batalin--Vilkovisky gauge formalism require the structure of a supermanifold and the notion of \emph{odd} and \emph{even} coordinates (in physics language they correspond to \emph{fermionic}\footnote{Here \emph{fermionic} just means anticommuting (which corresponds to Pauli's exlclusion principle) but not necessarily spinorial.} (anticommuting) and \emph{bosonic} (commuting) particles respectively). 
The mathematical theory of supergeometry goes back to the work of Berezin and Leites between 1975 and 1980 in \cite{BerLei75,Lei80}. For our purposes we need a refinement, called graded geometry, where we give the variables an additional integer degree. More on graded geometry can be found e.g. in \cite{CattSchaetz11,Fiorenza2008,Cu}.

\begin{rem}
We will denote the \emph{exterior algebra} of a vector space $V$ by $\bigwedge V=\bigoplus_k\bigwedge^k V$ and the algebra of \emph{differential forms} on a manifold $M$ by $\Omega(M)=\bigoplus_k\Omega^k(M)$. We use $\widehat{\otimes}$ to indicate the topological tensor product, i.e. the unique tensor product such that $C^\infty(\R^n)\widehat{\otimes} C^\infty(\R^m)\cong C^\infty(\R^{n+m})$.
\end{rem}

\subsection{Graded spaces}
To understand the concept of a ``supermanifold'' we need to look first at the linear case, which is the concept of a graded space (or also graded vector space). 

\begin{defn}[Graded vector space]
A $\mathbb{Z}$-\emph{graded vector space} (or often just \emph{graded vector space}) is a collection\footnote{Working with the ``total space'', i.e. the direct sum of the $V_k$s, create all sort of troubles when one start talking of morphisms, so it should be avoided, thus we just talk of the components $V_k$. However, one should mention that the ``total space'' picture is fine in the the finite dimensional case (or just when only finitely many $V_k$ are nontrivial.)} of vector spaces $V=(V_k)_{k\in\mathbb{Z}}$.
\end{defn}

\begin{defn}[Shift]
Let $V=(V_n)_{n\mathbb{Z}}$ be a graded vector space. For any integer $k$ we can define the \emph{$k$-shift} of $V$ to be given by 
$$(V[k])_n:=(V_{k+n})_{n\in\mathbb{Z}}.$$
\end{defn}

\begin{defn}[Graded linear map]
A \emph{graded linear map} $f\colon V\to W$ between to graded vector spaces $V$ and $W$ is given by a collection of linear maps $(f_k\colon V_k\to W_k)_{k\in\mathbb{Z}}$.
\end{defn}

\begin{rem}
A graded linear map of \emph{degree $k$} between two graded vector spaces $V$ and $W$ is a graded linear map between $V$ and $W[k]$.
\end{rem}

\begin{defn}[Dual of a graded vector space]
Let $V=(V_k)_{k\in\mathbb{Z}}$ be a graded vector space. The dual space $V^*$ of $V$ is defined as $V^*=(V^*_{-k})_{k\in\mathbb{Z}}$.
\end{defn}

\begin{defn}[Superspace]
A \emph{superspace} is a $\mathbb{Z}_2$-graded vector space $V=V_{\bar 0}\oplus V_{\bar 1}$. We denote by $V_0$ the vector space of \emph{even} vectors and by $V_1$ the space of \emph{odd} vectors. 
\end{defn}

\begin{defn}[Parity/degree]
The \emph{parity} (or \emph{degree}) is the $\mathbb{Z}_2$-grading of the superspace. The degree of a homogeneous element $x\in V$ is denoted by $\vert x\vert$ and is defined by 
$$\vert x\vert:=\begin{cases} even,&x\in V_{\bar 0}\\ odd,& x\in V_{\bar 1}\end{cases}$$
\end{defn}

\begin{rem}
In addition to the $\mathbb{Z}$-grading we introduce independently a parity in $\mathbb{Z}_2$ denoted by $\vert\enspace\vert$. Moreover, in this note we will only consider the case when the parity is the degree modulo 2.
\end{rem}

\begin{rem}
If we consider the superspace $V$ as an ordinary vector space $V$ together with an automorphism $P\colon V\to V$, such that $P^2=\id$, we get that $V_{\bar 0}$ is the $1$-eigenspace and $V_{\bar 1}$ is the $(-1)$-eigenspace. Moreover, $P(x)=(-1)^{\vert x\vert}x$.
\end{rem}

\begin{defn}[Dimension]
We define the \emph{dimension of a superspace} to be given by 
$$\dim V=(\dim V_{\bar 0},\dim V_{\bar 1}),$$
and we also say that $V$ is a $\dim V_{\bar 0}\mid \dim V_{\bar 1}$-dimensional superspace. 
\end{defn}

\begin{rem}
We can always think of a vector space $V$ to be a superspace with $P=\id$. In this case, one usually considers $V$ to be the even part of the superspace $(V=V\oplus 0,\id)$. We can also consider $V$ to be given as the odd part of the superspace by the pair $(V=0\oplus V[1],-\id)$, where we denote by $V[1]$ the odd degree which is simply given by $(V,-\id)$. 
\end{rem}

\begin{rem}
Note that, considering a vector space $V$ as a superspace $(V,P_V)$, we can consider its dual by $(V^*,P_{V^*}=(P_V)^*)$.
\end{rem}

Consider a $\mathbb{Z}$-graded vector space $V=\bigoplus_{k\in \mathbb{Z}}V_k$, where only finitely many $V_k$s are nontrivial. We can consider $V$ as a superspace by setting 
\begin{align}
\label{superspace}
\begin{split}
V_{\bar 0}&=\bigoplus_{k\in2\mathbb{Z}}V_k,\\
V_{\bar 1}&=\bigoplus_{k\in2\mathbb{Z}+1}V_k.
\end{split}
\end{align}

\begin{defn}[Morphism of superspaces]
A \emph{morphism} between two superspaces $V$ and $W$ is a linear map $\phi\colon V\to W$ which preserves the grading. Considering $V$ and $W$ with their correpsonding automorphisms $P_V$ and $P_W$ we can equivalently say that $\phi$ is a morphisms between $V$ and $W$ if the following diagram commutes:
\[
\begin{tikzcd}
V\arrow[d,swap,"\phi"]\arrow[r,"P_V"]&\arrow[d,"\phi"]V\\
W\arrow[r,"P_W"]&W
\end{tikzcd}
\]
\end{defn}

\subsection{Category of superspaces}
One can observe that superspaces form a category (over some fixed field $\mathbb{K}$), which we denote by ${\bf SuperVect}_\mathbb{K}$. One can define an \emph{endofunctor} on ${\bf SuperVect}_\mathbb{K}$
$$\Pi\colon (V,P)\mapsto (V,-P).$$
This functor is called \emph{change of parity}. Thus we get 
\begin{align}
\begin{split}
(\Pi V)_{\bar 0}&=V_{\bar 1},\\
(\Pi V)_{\bar 1}&=V_{\bar 0}.
\end{split}
\end{align}

\begin{rem}
The superspace $V[1]$, coming from a vector space $V$, can also be written as $\Pi V$.
\end{rem}

\begin{rem}
One can observe that ${\bf SuperVect}_\mathbb{K}$ is actually a symmetric monoidal category by 
\begin{align}
\begin{split}
(V\otimes W)_{\bar 0}&=(V_{\bar 0}\otimes W_{\bar 0})\oplus (V_{\bar 1}\otimes W_{\bar 1}),\\
(V\otimes W)_{\bar 1}&=(V_{\bar 0}\otimes W_{\bar 1})\oplus (V_{\bar 1}\otimes W_{\bar 0}).
\end{split}
\end{align}
Moreover, $P_{V\otimes W}=P_V\otimes P_W$. One can also define an operation 
\begin{align}
\begin{split}
V\otimes W&\to V\otimes W\\
x\otimes y&\mapsto (-1)^{\vert x\vert \vert y\vert}y\otimes x,
\end{split}
\end{align}
which is called the \emph{braiding}.
\end{rem}

\begin{rem}
We have a natural embedding of symmetric monoidal categories:
\begin{align}
\begin{split}
{\bf Vect}_\mathbb{K}&\to {\bf SuperVect}_\mathbb{K}\\
V&\mapsto (V,\id),
\end{split}
\end{align}
where ${\bf Vect}_\mathbb{K}$ denotes the category of $\mathbb{K}$-vector spaces.
\end{rem}

\begin{defn}[Symmetric superspace]
For any superspace $V$ we can define for any positive integer $n$ its $n$th symmetric power as
\begin{equation}
\Sym^nV=V^{\otimes n}/( x_1\otimes\dotsm \otimes x_n-\sigma(x_1\otimes\dotsm \otimes x_n),\sigma\in \Sigma_n),
\end{equation}
where $\Sigma_n$ denotes the symmetric group of order $n$ and $(\enspace\enspace)$ the ideal generated by some relation.
\end{defn}

\begin{rem}
We get $\Sym^n(V_{\bar 0}\oplus 0)=\Sym^nV_{\bar 0}$ and $\Sym^n(0\oplus V_{\bar 1})=\bigwedge^n V_{\bar 1}$. Moreover
\begin{equation}
\label{symm}
\Sym^n(V_{\bar 0}\oplus V_{\bar 1})=\bigoplus_{0\leq k\leq n}\left(\Sym^kV_{\bar 0}\otimes \bigwedge^{n-k}V_{\bar 1}\right).
\end{equation}
\end{rem}

\subsection{Short description for the local picture}
Locally, we consider coordinates $(x^{i})$ on an open $U\subset \R^n$ and the algebra of smooth maps $C^\infty(U)$, which are algebraically described by the commutativity of the coordinates, i.e.  we have the equivalence relation $x^{i}x^j=x^jx^{i}$. Considering this relation, we can add coordinates $\theta^\mu$ such that algebraically we have 
\begin{align*}
\theta^\mu x^{i}&=x^{i}\theta^\mu\\ 
\theta^\mu\theta^\nu&=-\theta^\nu\theta^\mu,
\end{align*}
i.e. an anticommuting relation for the $(\theta^\mu)$. So, we can describe the algebra, generated by these coordinates, as $\calA=\widehat{\R[x,\theta]}/\sim$, where $\sim$ is given by the commutative relation of the $(x^{i})$-coordinates and the anticommutative relation of the $(\theta^\mu)$-coordinates. Equivalently we can write $$\calA=C^\infty(U)\otimes \bigwedge V^*=:C^\infty(U\times \Pi V),$$ for some vector space $V$. 


\subsection{Supermanifolds}
We want to explain the globalization of the local picture given above, i.e. we want to describe a manifold structure such that locally we get the structure as we have seen before.
Consider a diffeomorphism between patches
\begin{equation}
\label{diffeo}
\varphi\colon U\times \Pi V\to \widetilde{U}\times \Pi \widetilde{V} 
\end{equation}
such that 
$$\varphi^*\colon C^\infty(\widetilde{U}\times \Pi\widetilde{V})\to C^\infty(U\times \Pi V)$$ or equivalently 
$$\varphi^*\colon C^\infty(\widetilde{U})\otimes \bigwedge \widetilde{V}^*\to C^\infty(U)\otimes \bigwedge V^*$$
is a superalgebra morphism.
Let $x$ and $\theta$ be coordinates on $U$ and $\Pi V$ respectively and $\widetilde{x}$ and $\widetilde{\theta}$ coordinates of $\widetilde{U}$ and $\Pi\widetilde{V}$ respectively. Then we can write 
\begin{align*}
\widetilde{x}^{\widetilde{i}}&=\varphi^{\widetilde{i}}(x,\theta)\hspace{0.5cm} \text{even on \hspace{0.5cm}$C^\infty(U)\otimes \bigwedge^{even}V^*$},\\
\widetilde{\theta}^{\widetilde{\mu}}&=\varphi^{\widetilde{\mu}}(x,\theta)\hspace{0.5cm}\text{odd on \hspace{0.5cm}$C^\infty(U)\otimes \bigwedge^{odd}V^*$}.
\end{align*}


A supermanifold is then given by patching together $U\times \Pi V$. Mathematically more clearly, we have the following definition:

\begin{defn}[Supermanifold]
\label{supermnf}
A \emph{supermanifold} $\calM$ is a locally ringed space $(M,\calO_M)$, which is locally isomorphic to $$\left(U,C^\infty(U)\otimes \bigwedge V^*\right),$$ where $U\subset \R^n$ is open and $V$ is some finite-dimensional real vector space. We call $M$ the \emph{body} of $\calM$ and $\calO_M$ the \emph{structure sheaf} of $M$.
\end{defn}

\begin{rem}
The isomorphism mentioned in Definition \ref{supermnf} is in the category of $\mathbb{Z}_2$-graded algebras, which is the parity:
\begin{align*}
\vert\enspace\enspace\vert\colon \bigoplus_{k\geq 0}C^\infty(U)\otimes \bigwedge^kV^*&\to \Z_2\\
f\otimes \theta&\mapsto \vert f\otimes \theta\vert:=\vert \theta\vert=k\mod 2.
\end{align*}
This induces that globally $C^\infty(\calM)$ is a graded commutative algebra. In particular, for two homogeneous elements $f,g\in C^\infty(\calM)$, we have $fg=(-1)^{\vert f\vert \vert g\vert}gf$.
\end{rem}

\begin{rem}
The supermanifold denoted by $\R^{n\vert m}$ is the supermanifold with body $\R^n$ and sheaf of functions $C^\infty(\R^n)\otimes \bigwedge \R^m$.
\end{rem}

\subsection{Morphisms of supermanifolds}
Let $\calM$ and $\calN$ be two supermanifolds. Then (in the smooth setting) we can define a morphism between $\calM$ and $\calN$ to be given by a morphism of \emph{superalgebras} from $C^\infty(\calN)$ to $C^\infty(\calM)$. Note that a \emph{superalgebra} is a $\mathbb{Z}_2$-graded algebra, i.e. an algebra over a commutative ring, where the multiplication preserves the grading. Hence the notion of a morphism is simillar as for superspaces. This idea is clear by the following definition:

\begin{defn}[Morphism of locally ringed spaces]
Let $(X,\calO_X)$ and $(Y,\calO_Y)$ be two ringed spaces. A \emph{morphism} $(X,\calO_X)\to (Y,\calO_Y)$ is a pair $(f,f^\sharp)$, where $f\colon X\to Y$ is a continuous map and $f^\sharp\colon \calO_Y\to f_*\calO_X$ is a morphism of local rings (i.e. it respects the maximal ideal). Here $f_*$ denotes the direct image functor for sheaves.
\end{defn}

\begin{rem}
Assume $\mathcal{N}$ to be given a local patch $\bar V\oplus \Pi \bar W$. Then one would get 
$$\Hom(\calM,\calN)\cong \Hom(C^\infty(\calN),C^\infty(\calM))\cong \left((\bar V\oplus \Pi\bar W)\otimes C^\infty(\calM)\right)_0.$$
\end{rem}
The last space represents the even part of the corresponding superspace.

We can also consider the algebra of \emph{polynomial functions} on $V\oplus \Pi W$, where $V$ and $W$ are finite-dimensional vector spaces, which is given by $\Sym(V^*)\otimes \bigwedge W^*$. We denote by $\Sym$ the symmetric algebra. Globally, the algebra of smooth functions on a supermanifold $C^\infty(\calM)$ is defined to be the algebra of global sections of the sheaf associated to $\calM$. The parity can be extended to this algebra and thus it is a super commutative algebra with respect to this parity, i.e. for two homogeneous elements $f$ and $g$ of degree $\vert f\vert$ and $\vert g\vert$ respectively, one has $f\cdot g=(-1)^{\vert f\vert \vert g\vert}g\cdot f$.

\begin{ex}[Differential forms]
Let $M$ be a manifold. Then the algebra of differential forms $\Omega(M)$ is locally isomorphic to $C^\infty(U)\otimes \bigwedge T_x^*M$, where $x$ is some point on $U$. This means that the sheaf of differential forms on a manifold corresponds to a supermanifold. 
\end{ex}

\begin{ex}
Let $V$ be a vector space. If we regard $V$ as a superspace $(V,\id)$, then one can consider its associated shifted (or odd) superspace $\Pi V=V[1]=(V,-\id)$ such that $C^\infty(\Pi V)=C^\infty(V[1])=\bigwedge V^*$. 
\end{ex}

\begin{ex}[Odd vector bundle]
\label{ex_oddVB}
Let $E\to M$ be a vector bundle. Then we can associate to it the odd vector bundle $\Pi E$, which is a supermanifold. Moreover, the smooth functions are given by 
\begin{equation}
\label{oddVB}
C^\infty(\Pi E)=\Gamma\left(\bigwedge E^*\right),
\end{equation}
where $\Gamma$ denotes the space of smooth sections.
\end{ex}

\begin{ex}[Chevalley--Eilenberg complex]
\label{CE_complex}
Consider a real, finite-dimensional Lie algebra $\mathfrak{g}$. The cochains of the \emph{Chevalley--Eilenberg} complex of $\mathfrak{g}$ are the elements of $\bigwedge\mathfrak{g}^*$, i.e. the vector spaces $\bigwedge^{k}\mathfrak{g}^*$, which is the same as the algebra of smooth functions on the supermanifold $\Pi \mathfrak{g}$. In particular, this is a subexample of the more general example \ref{ex_oddVB}.
\end{ex}

\begin{ex}[Odd (co)tangent bundle]
\label{odd_cotangent}
This is again a subexample of example \ref{ex_oddVB}. Consider a smooth manifold $M$. Then we can consider the supermanifolds $\Pi TM$ and $\Pi T^*M$, which we call the \emph{odd tangent bundle} and \emph{odd cotangent bundle} of $M$ respectively. According to example \ref{ex_oddVB}, the smooth functions are then given by 
\begin{align}
\label{oddTBfct}
C^\infty(\Pi TM)&=\Gamma\left(\bigwedge T^*M\right)=\Omega(M),\\
\label{oddCTBfct}
C^\infty(\Pi T^*M)&=\Gamma\left(\bigwedge TM\right)=\mathfrak{X}_{mult}(M),
\end{align}
where $\mathfrak{X}_{mult}(M)$ denotes the algebra of \emph{multivector fields} on $M$ endowed with the \emph{Schouten--Nijenhuis bracket}. 
\end{ex}

\begin{rem}
Since we mostly consider the the supermanifolds to be odd (co)tangent bundles, equations \eqref{oddTBfct} and \eqref{oddCTBfct} are indeed helpful, since the fields are then given by differential forms or multivector fields.
\end{rem}


\subsection{Graded Manifolds} We want to introduce another important manifold structure, which can be combined with the structure of supermanifolds. 

\begin{defn}[Graded manifold]
\label{grmanifold}
A \emph{graded manifold} $\calM$ is a manifold $M$, which locally looks like $(U,C^\infty(U)\otimes \Sym(V^*))$, where $U\subset \R^n$ is open and $V$ is a graded vector space. 
\end{defn}

\begin{rem}
One can construct an isomorphism between the the structure sheaf of a supermanifold and the local model of a graded manifold, which will be in the category of $\mathbb{Z}$-graded algebras. Thus we consider a graded manifold as a tuple $(M,\calO_M)$. We denote the $\mathbb{Z}$-grading by $\gh$.
\end{rem}


\begin{rem}
In application to field theory, the space of fields arise with a graded manifold structure. Moreover, the $\mathbb{Z}$-grading in bosonic field theory, as defined before, is called the \emph{ghost number} (see Subsection \ref{ghostnr}). Hence, we can combine the structure of a graded manifold and a supermanifold by considering a superspace of the form \eqref{superspace} and consider $\Sym(V^*)$ in Definition \ref{grmanifold} to be given as in \eqref{symm}.
\end{rem}

\begin{defn}[Graded vector bundle]
A \emph{graded vector bundle} $E$ over a manifold $M$ is a collection of ordinary vector bundles $\bigoplus_{k\in\mathbb{Z}}E_k$ over $M$. 
\end{defn}

\begin{rem}
One can consider a sheaf $U\mapsto \Gamma(U,\Sym(E\vert_U^*))$, which maps an open subset to a section with values in the graded symmetric algebra of the dual bundle restricted to $U$. Examples \ref{CE_complex} and \ref{odd_cotangent} have natural structures of graded manifolds, with ghost number corresponding to the usual degrees of the algebras associated to these examples.
\end{rem}

\begin{prop}
Any smooth graded manifold is isomorphic to a graded manifold associated to a graded vector bundle.
\end{prop}

\subsection{Vector fields}
We can now define the notion of a vector field on a supermanifold.
\begin{defn}[Supervector field]
A \emph{supervector field} is a vector field $X$ on a supermanifold $\calM$, with local coordinates $(x^{i},\theta^\mu)$, which is given by 
$$X=\sum_{i,\mu} X^{i}\partial_{x^{i}}+Y^\mu\partial_{\theta^\mu},$$
such that $\partial_{\theta^\mu}x^{i}=0$ and $\partial_{\theta^\nu}\theta^\mu=\delta_\nu^\mu$.
\end{defn}
Equivalently, vector fields on $U\times \Pi V$ are global derivations on $C^\infty(U)\otimes \bigwedge V^*$.
\begin{rem}
The $\partial_{x^{i}}$ are the usual position derivatives.
\end{rem}
\begin{defn}[Graded supervector field]
A \emph{graded vector field} is a graded linear map $X\colon C^\infty(\calN)\to C^\infty(\calN)[k]$, where $\calN$ is a graded manifold, which satisfies the \emph{graded Leibniz rule}: for any two homogeneous functions $f,g\in C^\infty(\calN)$ we have 
\begin{equation}
\label{grLeibniz}
X(fg)=X(f)g+(-1)^{k\vert f\vert}fX(g).
\end{equation}
\end{defn}
Denote the space of graded vector fields on a given graded manifold (or supermanifold) $\calM$ by $\mathfrak{X}(\calM)$.
We can define a \emph{graded} Lie bracket $[\enspace,\enspace]_{\mathfrak{X}(\calM)}\colon \mathfrak{X}(\calM)\otimes \mathfrak{X}(\calM)\to \mathfrak{X}(\calM)$ by 
$$[X,Y]_{\mathfrak{X}(\calM)}:=X\circ Y-(-1)^{\vert X\vert\vert Y\vert} Y\circ X.$$ 
The space $\mathfrak{X}(\calM)$ endowed with $[\enspace,\enspace]_{\mathfrak{X}(\calM)}$ becomes then a graded Lie algebra (since the bracket is graded).

\begin{rem}
We will not always write the ``super'' in front of ``vector field'' and assume it is understood from the context.
\end{rem}

\subsection{Differential forms}
Consider a graded manifold $\calM$ with homogeneous local coordinates $(x^{i})$. Then we can locally form an algebra by adding coordinates $(\dd x^{i})$ to the coordinates we had before. Note that if $x^{i}$ is a coordinate of odd degree, then $\dd x^{i}$ is even and $(\dd x^{i})^2\not=0$. On an ordinary smooth manifold, differential forms have two important properties: they can be differentiated – hence the name differential forms – and they also provide the right objects for an integration theory on submanifolds. It turns out that on graded manifolds, the latter is no longer true, since the differential forms we introduced do not come along with a nice integration theory. 
Differential forms on a supermanifold $\calM$  with local coordinates $(x^{i},\theta^\mu)$ are locally generated by the $1$-forms $\dd x^{i}$ and $\dd \theta^\mu$. We can take the algebra, which is generated by $x^{i},\theta^\mu, \dd x^{i}, \dd \theta^\mu$, whereas now the $\dd x^{i}$ are odd and the $\dd\theta^\mu$ are even. In particular, for any local coordinate $y$ of $\calM$, we have $\vert\dd y^j\vert=\vert y^j\vert+1$. Thus, we have a de Rham differential on the algebra of differential forms 
$$\dd\colon \hspace{0.2cm}x\mapsto \dd x,\hspace{0.2cm} \theta\mapsto \dd\theta,$$
with $\dd^2=0$, which is a \emph{graded differential}. Considering a supermanifold $\calM$ with local coordinates $(x^{i},\theta^\mu)$ we can look at the \emph{shifted (odd) tangent bundle} $\Pi T\calM$ with local coordinates $(x^{i},\theta^\mu,\dd x^{i},\dd \theta^\mu)$. Moreover, the differential is then given by 
\begin{equation}
\label{diff}
\dd =\sum_{i,\mu} \dd x^{i}\partial_{x^{i}}+\dd\theta^\mu\partial_{\theta^\mu},
\end{equation}
with the property 
\begin{equation}
\label{cohom}
[\dd,\dd]=0. 
\end{equation}
Note that the differential is given as a vector field on $\Pi T\calM$. A vector field, such as in \eqref{diff}, satisfying \eqref{cohom} is called a \emph{cohomological vector field}. Consider the \emph{de Rham complex} $(\Omega(\calM),\dd)$ for some given supermanifold $\calM$. It is given by $C^\infty(T[1]\calM)$ equipped with a graded vector field $Q$ of degree $+1$, i.e. $Q\colon C^\infty(\calM)\to C^\infty(\calM)[1]$, which satisfies the graded Leibniz rule \eqref{grLeibniz}. This becomes more clear by the following lemma:
\begin{lem}
Every \emph{cohomological vector field} on a supermanifold $\calM$ corresponds to a \emph{differential} on the graded algebra of smooth functions $C^\infty(\calM)$
\end{lem}
\begin{proof}
Since a cohomological vector field $Q$ is a vector field of degree $+1$, which commutes with itself we have $[Q,Q]=2Q\circ Q=0$, and since it raises the degree of a function by $1$, it corrsponds to a differential.
\end{proof}

\begin{ex}[Cohomological vector field on odd tangent bundles]
Consider the odd tangent bundle $\Pi TM$ for some smooth manifold $M$. Since $C^\infty(\Pi TM)$ is given by $\Omega(M)$, we get that the de Rham differential $\dd_M$ on $M$ is a cohomological vector field. In local coordinates $(x^{i},\dd x^{i})$, we get 
$$Q=\sum_{i}\dd x^{i}\partial_{x^{i}}.$$
\end{ex}

\begin{ex}[Chevalley--Eilenberg differential]
Let $(\mathfrak{g},[\enspace,\enspace])$ be a finite-dimensional Lie algebra. The graded manifold $\Pi\mathfrak{g}$ carries a cohomological vector field $Q$, which corresponds to the Chevalley--Eilenberg differential on $\bigwedge\mathfrak{g}=C^\infty(\Pi\mathfrak{g})$. Let $(e_i)$ be a basis of $\mathfrak{g}$, and let $(f_{ij}^k)$ be the corresponding structure constants given by $[e_i,e_j]=\sum_k f_{ij}^ke_k$. Then, we get the cohomological vector field 
$$Q=\frac{1}{2}\sum_{i,j,k}x^{i}x^jf_{ij}^k\partial_{x^k},$$
where $(x^{i})$ are the coordinates on $\Pi\mathfrak{g}$, which correspond to the basis dual to $(e_i)$. In particular, one can check that $[Q,Q]=0$ is equivalent to the fact that the bracket $[\enspace,\enspace]\colon\mathfrak{g}\otimes\mathfrak{g}\to\mathfrak{g}$ defined on the generators as above, satisfies the Jacobi identity.
\end{ex}

\subsection{Graded symplectic forms}
A graded \emph{symplectic form} of degree $k$ on a graded manifold $\calM$ is a closed (w.r.t. the de Rham differential), nondegenerate $2$-form $$\omega\colon T\calM\to T^*[k]\calM,$$ which is, in local coordinates, given by 
\begin{equation}
\label{sympl}
\omega=\sum_{i,j}\dd z^{i}\omega_{ij}\dd z^j,
\end{equation} 
where $z^\alpha\in\{x^{i},\theta^\mu\}$. 
\begin{rem}
In application to the \emph{BV formalism}, one considers $\omega$ to be \emph{odd} and of degree $-1$, and  for the \emph{BFV formalism} to be \emph{even} and of degree 0.
\end{rem}
\begin{ex}
Let $V$ be a real vector space. The contraction between $V$ and $V^*$ defines a nondegenerate pairing on $V\oplus V^*$. This includes a constant symplectic form of degree $k+\ell$ on $V[k]\oplus V^*[\ell]$.
\end{ex}

\begin{ex}
Consider $\R[1]$ with the $2$-form $\omega=\dd x\dd x$. This is a symplectic form of degree $2$.
\end{ex}

Next, we want to introduce some notation. Let $V$ be a superspace and consider 
for a homogeneous element $v\in V$, and a function $f$ on $V$, the so-called \emph{left} and \emph{right derivatives}, which are defined as 
\begin{align} 
\overrightarrow{\partial}_{v}f&:=\partial_vf\\
f\overleftarrow{\partial}_v&:=(-1)^{\vert v\vert (\vert f\vert+1)}\partial_vf
\end{align}

\begin{rem}
One can check that the right derivative satisfies the Leibniz rule from the right.
\end{rem}

One can observe that a symplectic form as in \eqref{sympl}, induces a \emph{graded Poisson bracket} $$\{\enspace,\enspace\}\colon C^\infty(\calM)\otimes C^\infty(\calM)\to C^\infty(\calM),$$ which is given by 
$$\{f,g\}= f\overleftarrow{\partial}_i(\omega^{-1})^{ij}\overrightarrow{\partial}_j g,$$
which is a graded Poisson bracket. Similarly to ordinary manifolds, the Hamiltonian vector field $X_H$ for a Hamiltonian function $H\in C^\infty(\calM)$ is given by the equation $\iota_{X_H}\omega=-\dd H$, which leads to the equation $$\{H,G\}=(-1)^{\vert H\vert+1}X_H(G)$$ for some function $G\in C^\infty(\calM)$. Note that $\iota_X$, for some vector field $X$, is a vector field on $\Pi T\calM$.  Moreover, if $\{\enspace,\enspace\}$ is of degree $k$, we have
$$\{f,g\}=-(-1)^{(\vert f\vert+k)(\vert g\vert+k)}\{g,f\}$$

\subsection{Lie derivative}
We can also extend the definition of a \emph{Lie derivative} (see Subsection \ref{Lie_der}) $\sfL$ to supermanifolds, by noticing that with respect to $X$, we have $$\sfL_X=[\iota_X,\dd]\in \mathfrak{X}(\Pi T\calM),$$ which is obtained by the \emph{Cartan calculus}. Moreover, we have $[\sfL_X,\sfL_Y]=\sfL_{[X,Y]}$ and $[\iota_X,\sfL_Y]=\pm \iota_{[X,Y]}$. In particular, $\sfL_X$ can be obtained by differentiating the \emph{flow} of the vector field $X$.


\subsection{Integration on superspaces}
After defining the most important concepts around supermanifolds, we want to be able to perform integrals. Clearly, we need to extend the usual integration theory on manifolds to a more general picture. We want to start locally by considering integration over superspaces. Naturally, we define the integration for the even coordinates as usual. To get the correct integration theory, we need some notions on supermatrices and so-called \emph{Berezinians}.
\begin{defn}[Supermatrix]
A \emph{supermatrix structure} is a matrix structure with parity attached to each row and column. 
\end{defn}
We usually arange a supermatrix structure in such a way that all the even rows and columns come first, and the odd ones second, so that it can be conveniently written in block form, e.g. 
$$X=\begin{pmatrix}A&B\\ C&D\end{pmatrix},$$
where $A,B,C,D$ are the matrices corresponding to the partition into even and odd rows and columns. 
\begin{defn}[Order]
If a supermatrix structure has $p$ even rows and $q$ odd rows, and $r$ even columns and $s$ odd columns, we call it a matrix structure of size $(p,q)\times (r,s)$. A $(p,q)\times (p,q)$ structure is said to have \emph{order} $(p,q)$.
\end{defn}
Denote by $\Mat_{p,q}(\calA)$ the space of matrices of order $(p,q)$ on a commutative superalgebra $\calA=\calA_{even}\oplus\calA_{odd}$ and by $\GL_{p,q}(\calA)\subset \Mat_{p,q}(\calA)$ we denote the subset of invertible elements. We want to consider a homomorphism $\GL_{p,q}(\calA)\to \GL_{1,0}(\calA)=\calA^\times_{even}$, where $\calA^\times_{even}$ denotes the group of invertible elements of $\calA_{even}$. Such a map would be an analogue to the usual determinant.
\begin{lem}
\label{inv}
Consider 
$$X=\begin{pmatrix}A&B\\ C&D\end{pmatrix}\in(\Mat_{p,q}(\calA))_{even}.$$
Then $X$ is invertible if and only if $A$ and $D$ are invertible.
\end{lem}
Lemma \ref{inv} is an immediate consequence of the following proposition.
\begin{prop}
Let $\calA$ be a commutative superalgebra, and $$\pi\colon \calA\to \widetilde{\calA}=\calA/(\calA_{odd})$$ the natural homomorphism, and 
$$\pi\colon \Mat_n(\calA)\to \Mat_n(\widetilde{\calA})$$
the corresponding homomorphism of matrix algebras (where the superstructure is ignored). Then $X\in\Mat_n(\calA)$ is invertible if and only if $\pi(X)$ is invertible.
\end{prop}

\begin{defn}[Berezinian]
For $X=\begin{pmatrix}A&B\\ C&D\end{pmatrix}\in \GL_{p,q}(\calA)$ we define the \emph{Berezinian} of $X$ by 
$$\Ber(X):=\frac{\det(A-BD^{-1}C)}{\det(D)}$$
\end{defn}

\begin{rem}
Note that $D$ is invertible by Lemma \ref{inv}. The entries in $D$ and $A-BD^{-1}C$ lie in the commutative algebra $\calA_{even}$ so that the determinants are well-defined and $\Ber(X)\in \calA_{even}$.
\end{rem}

\begin{thm}
If $X,Y\in\GL_{p,q}(\calA)$, then 
$$\Ber(XY)=\Ber(X)\Ber(Y).$$
\end{thm}

Consider now a supermanifold $\calM$ with local coordinates $(x^{i},\theta^\mu)$ such that $1\leq i\leq n$ and $1\leq\mu\leq m$. For the odd coordinates $\theta$, we integrate the Berezinian $\dd \theta^\mu$ according to the following rules:

\begin{itemize}
\item{$\int \dd\theta^\mu=0$,}
\item{$\int\theta^\mu\dd\theta^\nu=\delta_{\mu\nu}$,}
\item{Fubini's theorem holds.}
\end{itemize}

Moreover, consider a function $f\in C^\infty(\calM)$ given locally as 
\begin{equation} 
f=\sum_{1\leq\mu_1,...,\mu_{m}\leq m}f_0+f_{\mu_1}\theta^{\mu_1}+f_{\mu_1\mu_2}\theta^{\mu_1}\theta^{\mu_2}+\dotsm + f_{\mu_1,...,\mu_m}\theta^{\mu_1}\dotsm\theta^{\mu_{m}},
\end{equation}
where $f_0$ and each $f_{\mu_1,...,\mu_k}$ are elements in $C^\infty(U)$ for $U$ an open subset of $\R^n$ (for all $1\leq k\leq m$). Then, according to the rules above, we get 
$$\int f\dd\theta^1\dotsm \dd\theta^n=f_{\mu_1,...,\mu_m}=:f_{top}.$$
Thus, for a supermanifold of the form $\Pi V$, integration in general will give a map  
\begin{equation}
\label{int_1}
\int_{\Pi V}\colon\enspace\bigwedge V^*\rightarrow\R,
\end{equation}
such that for any $g\in \bigwedge^{<top}V^*$ we get $\int_{\Pi V}g=0$, and hence the map \eqref{int_1} is given by 
\begin{equation}
\int_{\Pi V}=\int_{\Pi V}\!\Bigg\vert_{\bigwedge^{top}V^*}\colon\enspace\bigwedge^{top}V^*\rightarrow \R,
\end{equation}
for each choice of frame of $V$. If we take a coordinate patch $U\times \Pi V$ of $\calM$, we get 
$$\int_{U\times \Pi V}f=\int_U\int_{\Pi V}f=\int_U f_{top}\dd\mu.$$
Consider now a linear map $A\colon V\to W$ between two vector spaces $V$ and $W$. We can look at the map $\wedge A^*\colon \bigwedge W^*\to \bigwedge V^*$, which can be regarded as a map $\Pi V\to \Pi W$. Then, for any $f\in C^\infty(\Pi W)=\bigwedge W^*$, we have 
$$\int_{\Pi V}\wedge A^* f=\det(A)\int_{\Pi W}f.$$
Moreover, consider a linear isomorphism $$\varphi:=D\otimes\wedge A^*\colon U\times \Pi V\to\widetilde{U}\times \Pi W,$$ corresponding to linear maps $D\colon U\to \widetilde{U}$ and $A\colon V\to W$. Then 
$$\frac{ \det(D)}{\det(A)}\int_{U\times \Pi V}\varphi^*f=\Ber\begin{pmatrix}D&0\\ 0&A\end{pmatrix}\int_{U\times \Pi V}\varphi^*f=\int_{\widetilde{U}\times \Pi W}f.$$

The corresponding measure is given by 
$$\frac{\det(D)}{\det(A)}=\dd x^{1}\dotsm \dd x^n\dd\theta^1\dotsm\dd\theta^m,$$

We can apply this construction also to \emph{Gaussian integrals} and observe 
\[
\int_{\mathbb{C}} \ee^{\frac{\I}{2}z^tMz}\dd z=\frac{\text{const.}}{\Ber(M)}, 
\]
where $M$ is the matrix for some nondegenerate pairing. Using the diffeomorphism \eqref{diffeo}, we can obtain. 
$$\int_{U\times \Pi V}\varphi^*f \Ber(\varphi\overleftarrow{\dd})=\int_{\widetilde{U}\times \Pi \widetilde{V}}f.$$

\subsection{Integration on supermanifolds}
To perform integration on a manifold, we need the notion of a \emph{density}, which is fairly standard for ordinary manifolds. A \emph{density} for our purposes is a section of the \emph{Berezinian bundle} tensor the orientation bundle of the underlying manifold $\calM$, which locally means that they are functions transforming like $\Ber(\varphi\overleftarrow{\dd})$. Everything else is constructed in the same way as for ordinary manifolds.

\section{The BV formalism}

The \emph{Batalin--Vilkovisky (BV) formalism} is based on the idea of having an \emph{odd} symplectic form. Let $\calM$ be a supermanifold with local coordinates $(q^{i},p_i)$ for $1\leq i\leq n$. Then we have a symplectic form $\omega=\sum_{1\leq i\leq n}\dd q^{i}\dd p_{i}$, where the parity of the $p_i$ is the opposite of the parity of the $q^{i}$. Now we can define the \emph{BV Laplacian}, which is given by 
\begin{equation}
\label{BVLaplacian}
\Delta=\sum_{1\leq i\leq n}(-1)^{\vert q^i\vert}\frac{\partial^2}{\partial q^{i}\partial p_i}.
\end{equation}
One can then show that $\Delta f=-\frac{1}{2}\text{div} X_f$, where $\text{div} X:=\sum_{1\leq i\leq n} \pm\partial_i X^{i}$, denotes the divergence of the vector field $X$ with respect to the standard Berezinian associated to this choice of coordinates. We have the following properties:
\begin{itemize}
\item{$\Delta^2=0$.
}
\item{$\Delta(fg)=\Delta fg\pm f\Delta g\pm (f,g)$,
where $(\enspace,\enspace)$ denotes an odd Poisson bracket coming from the symplectic form $\omega$.}
\end{itemize}

\begin{defn}[BV integration]
Suppose $f$ is a function of the $p$ and $q$ variables, and consider an odd function $\Psi$ depending only on the $q$-variables. Then we can write $p_i:=\frac{\partial\Psi}{\partial q^{i}}$. Let us choose a projectable\footnote{Geometrically we think of this assignment as the definition of a Lagrangian submanifold that is projectable onto the space of the $q$-variables.} Lagrangian submanifold $\calL_\Psi$ of $\Pi T^*\R^n$. Then we can define the \emph{BV integration} $\int_{\calL_{\Psi}}$ as a map by
\begin{equation}
\label{BVpushforward}
\int_{\calL_\Psi}f:=\int f\Big\vert_{p_i=\frac{\partial\Psi}{\partial q^{i}}}\dd^nq.
\end{equation}
\end{defn}

\begin{rem}
In the physics literature, the odd function $\Psi$ is called the \emph{gauge fixing fermion} (see Remark \ref{gauge_independence}). 
\end{rem}
The following theorem is one of the main results of Batalin and Vilkovisky, which gives a powerful tool to deal with supersymmetric gauge theories.
\begin{thm}[Batalin--Vilkovisky]
\label{BVthm}
The following hold:
\begin{enumerate}
\item{If $f=\Delta g$, then $\int_{\calL_\Psi}f=0$,\text{ for every $\Psi$ such that the integral is defined.}}
\item{If $\Delta f=0$, then $\frac{\dd}{\dd t}\int_{\calL_{\Psi_t}}f=0$, where $\{\Psi_t\}$ is a continuous family of gauge fixing fermions, such that the integral is defined for every $t$.}
\end{enumerate}
\end{thm}

\begin{proof}
We denote $\partial^i:=\frac{\partial}{\partial q^{i}}$ and $\partial_i:=\frac{\partial}{\partial p_i}$.
For $(1)$ compute 
$$\sum_{1\leq i,j\leq n}(-1)^{\vert q^i\vert}
\partial_i(\partial^jg)\Big\vert_{p=\dd\Psi}=\sum_{1\leq i,j\leq n}\left[(-1)^{\vert q^i\vert}\partial_i\partial^j g\right]\Big\vert_{p=\dd\Psi}+\underbrace{\sum_{1\leq i,j\leq n}\left[(-1)^{\vert q^{i}\vert}\partial_i\partial_j\Psi\partial^j\partial^{i}g\right]\Big\vert_{p=\dd\Psi}}_{=0}.$$
Hence 
$$\Delta g\Big\vert_{p=\dd\Psi}=\sum_{1\leq i,j\leq n}(-1)^{\vert q^{i}\vert}\partial_i(\partial^j g)\Big\vert_{p=\dd\Psi}.$$
Taking the integral, we get
$$\int \Delta g\Big\vert_{p=\dd\Psi}=\sum_{1\leq i,j\leq n}(-1)^{\vert q^{i}\vert}\int\partial_i(\partial^jg)\Big\vert_{p=\dd\Psi}\dd q^{i}=0.$$ 
For $(2)$ compute
$$\frac{\dd}{\dd t}\int_{\mathcal{L}_{\Psi_t}}f=\int (\partial_i\dot{\Psi}_t\partial^{i}f)\Big\vert_{p=\dd\Psi}\dd q^{1}\dotsm \dd q^n=\int \pm \Delta (\dot{\Psi}_tf)\Big\vert_{p=\dd\Psi}.$$
By the properties of the BV Laplacian together with the asumption we get $\Delta f=0$ and $\Delta\dot{\Psi}_t=0$, and thus by $(1)$ we get
$$\int \pm \Delta(\dot{\Psi}_tf)=0.$$
\end{proof}

\begin{rem}[Gauge independence]
\label{gauge_independence}
Consider a continuous family $\{\Psi_t\}$ of gauge fixing fermions. Moreover, let $\calL_{\Psi_0}=\mathcal{L}_0=\R^{n\vert m}\subset \Pi T^*\R^{n\vert m}$ and assume that $\int_{\calL_{\Psi_0}}f$ does not converge. Then if $\Delta f=0$ and if $\int_{\calL_{\Psi_t}}f$ converges for all $t\not=0$ in a neighborhood of $0$, then we can define the integral as 
\[
\int_{\calL_{\Psi_0}}f:=\int_{\calL_{\Psi_t}}f
\]
for some $t$ in this pointed neighborhood (by Theorem \ref{BVthm} it does not matter which one).
Since the choice of a Lagrangian submanifold $\calL_\Psi$ is equivalent to \emph{choosing a gauge}, the above statement tells us exactly that it is independent of gauge fixing.
\end{rem}
In application to quantum field theory we are interested in the function $f=\ee^{\frac{\I}{\hbar}\calS}$ for some even function $\calS$. Then we can observe 
\begin{equation}
\label{BV_acting}
\Delta f=\frac{\I}{\hbar}\Delta \calS \ee^{\frac{\I}{\hbar}\calS}+\left(\frac{\I}{\hbar}\right)^2\frac{1}{2}(\calS,\calS)\ee^{\frac{\I}{\hbar}\calS},
\end{equation}
and thus $\Delta f=0$ if and only if 
\begin{equation}
\label{QME}
\frac{1}{2}(\calS,\calS)-\I\hbar\Delta \calS=0.
\end{equation}
Equation \eqref{QME} is called the \emph{Quantum Master Equation (QME)}. The lowest order term of the QME in the expansion of $\calS_\hbar=\calS_0+\hbar \calS_1+\hbar^2\calS_2+\dotsm $ is then given by $(\calS_0,\calS_0)=0$, which is called the \emph{Classical Master Equation (CME)}. Moreover, we want that the lowest order term of the function $\calS_\hbar$ is given by the original function, i.e. $\calS_0=\calS$. 

\begin{rem}
Note that the CME makes sense also in the infinite dimensional setting of field theory, where $\mathcal{S}$ is a local functional, which roughly means that we want to be able to express it as an integral of a differential form, and we understand $\omega$ as a local symplectic structure.
\end{rem}

\subsection{The BV pushforward}
Consider a split of the coordinates as $(p,q)=(p',q',p'',q'')$ and the corresponding symplectic forms $\omega'=\sum\dd q'\dd p'$ and $\omega''=\sum \dd q''\dd p''$, and the Laplacian $\Delta=\Delta'+\Delta''$. Moreover, let $\Psi$ be a gauge fixing fermion as an odd function of the $q''$. Then, considering a Lagrangian submanifold $\calL_{\Psi}$ of the double prime factor, we get 
\begin{equation}
\label{fiber_BV}
\int_{\calL_{\Psi}}f:=\left(\int f\Big\vert_{p''=\frac{\partial\Psi}{\partial q''}}\dd^\bullet q''\right)(p',q').
\end{equation}

\begin{prop}
\label{fiber_int}
Assuming that the integrals are defined as in Theorem \ref{BVthm}, the following hold:
\begin{enumerate}
\item{$\Delta'\int_{\calL_\Psi}f=\int_{\calL_\Psi}\Delta f$}
\item{If $\Delta f=0$, then $\frac{\dd}{\dd t}\int_{\calL_{\Psi_t}}f=\Delta'(\dotsm)$, where $\{\Psi_t\}$ is a continuous family of gauge fixing fermions.
}
\end{enumerate}
\end{prop}

\begin{rem} 
\label{cohom_class}
From Proposition \ref{fiber_int} we can obtain that the $\Delta'$-cohomology of $\int_{\calL_{\Psi}}f$ is well-defined, i.e. for an $f$ with $\Delta f=0$, we get 
$$\Delta'\int_{\calL_\Psi}f=0.$$
In other words, the BV pushforward is a chain map.
\end{rem}

\subsection{Global version}
We want to globalize the picture above on any odd symplectic manifold $\calM$. One can show that for any such $\calM$, there is a supermanifold $\calN$, such that $\calM\cong \Pi T^*\calN$. Then we get $C^\infty(\calM)\cong C^\infty(\Pi T^*\calN)=\Gamma\left(\bigwedge T\calN\right)$. The \emph{Berezinian bundle} on $\calM$ is given by 
$$\Ber(\calM)\cong\bigwedge^{top}T^*\calN\otimes \bigwedge^{top}T^*\calN=\left(\bigwedge^{top}T^*\calN\right)^{\otimes 2}.$$
We define the half-densities on $\calM$ to be given as 
$$\Dens^{\frac{1}{2}}(\calM):=\Gamma(\Ber(\calM)^{\frac{1}{2}})\cong\Gamma\left(\bigwedge T\calN\right)\otimes \Gamma\left(\bigwedge^{top}T^*\calN\right).$$
Moreover, one can show that there exists a canonical operator $\Delta_\calM^{\frac{1}{2}}$ on $\Dens^{\frac{1}{2}}(\calM)$ such that $(\Delta_\calM^{\frac{1}{2}})^2=0$. Indeed, if we consider the map 
\begin{align*}
\Phi\colon\Gamma\left(\bigwedge T\calN\right)\otimes_{C^\infty(\calN)}\Gamma\left(\bigwedge^{top} T^*\calN\right)&\to\Omega(\calN)\\
X\otimes v&\mapsto \iota_Xv, 
\end{align*}
we get such a Laplacian by setting $\Delta^{\frac{1}{2}}_\calM=\Phi^{-1}\circ \dd\circ \Phi$, where $\dd $ is the de Rahm differential on $\Omega(\calN)$. Moreover, consider a Lagrangian submanifold $\calL$ of $\calM$ and a half-density $\rho\in \Dens^{\frac{1}{2}}(\calL)$. Then $i^*\rho\in \Ber(\calL)$, where $i$ is the inclusion $\calL\hookrightarrow\calM$. Now we can see that the conditions $\int_\calL \Delta^{\frac{1}{2}}_\calL\rho=0$ and $\Delta^{\frac{1}{2}}_\calL\rho=0$ imply $\frac{\dd}{\dd t}\int_{\calL_t}\rho=0$. Consider now a submanifold $\calC$ of $\calN$. Then the odd conormal bundle (see Definition \ref{conormal_bun1}) of $\calC$ as a submanifold $\Pi N^*\calC\subset \calM$ is Lagrangian, and in particular we get
$$\int_{\Pi N^*\calC}\rho=\int_\calC \Phi(\rho).$$
For $\mathcal{L}$ of this form, it follows from Stokes' theorem. For $\mathcal{L}$ not of this form, one has to show that $\mathcal{L}$ may be reduced to one of this form by a Hamiltonian transformation and then apply the local computation.

Fix some nonvanishing half-density $\rho\in \Dens^{\frac{1}{2}}(\calM)$ satisfying $\Delta^{\frac{1}{2}}_\calM\rho=0$, and consider a function $f\in C^\infty(\calM)$. Then we can define a Laplacian $\Delta_\rho$ by  $\Delta_\rho f:=\frac{1}{\rho}\Delta^{\frac{1}{2}}_\calM(f\rho)$, and automatically we get $\Delta_\rho^2=0$. For two functions $f,g\in C^\infty(\calM)$, we can easily observe the generalized Leibniz rule $$\Delta_\rho(fg)=\Delta_\rho f g\pm f\Delta_\rho g\pm (f,g).$$ 

\subsection{Classical data}
 We can now define the mathematical data for the Batalin--Vilkovisky formalism in the classical setting. 
 \begin{defn}[BV manifold]
 \label{BV_manifold}
 A \emph{BV manifold} is a triple $(\calF,\omega,\calS)$ where $\calF$ is a supermanifold, $\omega$ an odd symplectic form and
$\calS$ an even function on $\calF$ satisfying the \emph{Classical Master Equation (CME)}
\[
(\calS,\calS)=0
\]
where $(\enspace ,\enspace)$ denotes the odd Poisson bracket (a.k.a.\ Gerstenhaber bracket) associated to $\omega$.
Usually, $\calS$ is called the \emph{BV action} and $(\enspace ,\enspace )$ is called the \emph{BV bracket}.
\end{defn}
Equivalently, we may introduce the Hamiltonian vector field $Q$ of $\calS$, defined by the equation
\[
\iota_Q\omega=-\dd \calS,
\]
which is odd, and require it to satisfy
\[
[Q,Q]=0,
\]
where $[\enspace,\enspace ]$ denotes the Lie bracket of vector fields.\footnote{The condition on $Q$ may seem weaker at first sight as
it seems just to induce that $(\calS,S)$ is a (locally) constant function. However, this constant must be odd, and the only odd constant is $0$.}$^,$\footnote{In applications to field theory the supermanifold $\calF$ is usually infinite-dimensional (typically modeled on a mapping space with Banach or Fr\'echet structure). The closed $2$\ndash form $\omega$ is then often required to be nondegenerate just in the weak sense, i.e., it is assumed that the induced map from vector fields to $1$\ndash forms is injective. This ensures the uniqueness of Hamiltonian vector fields but not their existence. The condition above then says that $\calS$ admits a Hamiltonian vector field.}
We can then equivalently define a BV manifold to be a quadruple $(\calF,\omega,\calS,Q)$ where $\calF$ is a supermanifold, $\omega$ an odd symplectic form,
$\calS$ an even function and $Q$ an odd vector field on $\calF$ satisfying the equations
\[
\iota_Q\omega=\dd \calS,
\qquad
[Q,Q]=0.
\]
\begin{ex}[BRST formalism]
A typical situation (especially in field theory) is when one starts with a manifold
$F$ (the space of classical fields), a Lie algebra $\frg$ acting on it (the symmetries) and a $\frg$\ndash invariant function
$S$ (the classical action). Assume moreover that $\frg$ has a nondegenerate invariant pairing.
The Chevalley--Eilenberg differential for the $\frg$\ndash module $C^\infty(F)$ may then be
reinterpreted as an odd vector field $\delta$ (the BRST operator)
on $F_\frg:=F\times\frg[1]$ which satisfies  $[\delta,\delta]=0$.\footnote{The space
of CE cochains $\bigwedge\frg\to C^\infty(F)$ may be reintepreted as the space of functions
on $F_\frg$. If $(e_i)$ is a basis for $\frg$, 
$f_{ij}^k$ the corresponding structure constants ($[e_i,e_j]=\sum_k f_{ij}^ke_k$),
and
$X_i$ the vector field on $F$ corresponding to $e_i$,
then the vector field on $F_\frg$ corresponding to the CE differential is
\[
\delta = \sum_{ijk} \frac12 c^ic^j f_{ij}^k \frac\de{\de c^k} + 
\sum_i c^iX_i,
\]
where the $c^i$s are the (odd) linear coordinates on $\Pi\frg$ corresponding to the chosen basis.}
The BV setting is then obtained by considering $\calF:=T^*[1]F_\frg$ with its canonical symplectic 
structure.\footnote{If we introduce local coordinates $x^\alpha$ on $F$ in addition to the coordinates
$c^i$ on $\frg$ introduced in the previous footnote, then on $\calF$ we have local coordinates
$(x^\alpha,c^i,x^+_\alpha,c^+_i)$, where $c^+_i$ is the (even) momentum of $c^i$ and $x^+_\alpha$ is the (odd) momentum of 
$x^\alpha$. Locally the canonical symplectic structure will read
$
\omega = \sum_\alpha \dd x^+_\alpha\,\dd x^\alpha +\sum_i \dd c^+_i\,\dd c^i.$}
We denote by $\Hat\delta$
the fiber linear function\footnote{With the notations of the previous footnotes,
we have
\[
\Hat\delta = \sum_{ijk} \frac12 c^ic^j f_{ij}^k  c^+_k+ 
\sum_i c^iX_i^\alpha x^+_\alpha,
\]
where we have written locally $X_i=\sum X_i^\alpha\frac\de{\de x^\alpha}$.} 
on $\calM$ canonically associated to $\delta$ and set
$\calS:=p^*S+\Hat\delta$, where $p$ is the composition of projections $\Pi T^*F_\frg\to F_\frg$ and
$F_\frg\to F$. 
It then follows that $(\calF,\omega,\calS)$, with $\omega$ the canonical symplectic form, is
a BV manifold. 
\end{ex}

%

\begin{ex}[Yang--Mills]
\label{YM}
Let $G$ be a finite-dimensional compact\footnote{Compactness of $G$ is rather important in physics, which also yields the Killing form as an invariant nondegenerate bilinear form on $\mathfrak{g}$} Lie group with Lie algebra $\mathfrak{g}$. The fields in the first order formulation of Yang--Mills theory are connections $A$ on a principal $G$-bundle $P$ over a $n$-manifold $N$, and $(n-2)$-forms $B$ with coefficients in the associated adjoint bundle. For simplicity, we assume that the principal bundle is trivial and consider connections as $1$-forms with coefficients in $\mathfrak{g}$ and $B$ fields as $(n-2)$-forms with coefficients in $\mathfrak{g}$. The \emph{ghost fields} $c$ are $0$-forms with coefficients in $\mathfrak{g}$. The BV extension includes also the \emph{antifields} $A^+,B^+$, and $c^+$. Note that for a field $\phi$ and its antifield $\phi^+$, we have the following relation of form degree and ghost number:
\begin{align}
\begin{split}
\label{field_antifield}
\deg(\phi^+)&=n-\deg(\phi),\\
\gh(\phi^+)+\gh(\phi)&=-1.
\end{split}
\end{align}
The total space of BV extended Yang--Mills theory\footnote{Here, as in the case of classical electrodynamics we only discuss the minimal BV extension. When $n=4$, the BV extension of Yang--Mills theory can also be presented in a different way using the decomposition of $2$-forms into self-dual and anti-self-dual parts.} is 
\begin{multline}
\calF_N=\underbrace{\mathfrak{g}\otimes\Omega^1(N)}_{\ni A}\oplus \underbrace{\mathfrak{g}\otimes\Omega^{n-2}(N)}_{\ni B}\oplus \underbrace{\mathfrak{g}\otimes\Omega^0(N)[1]}_{\ni c}\oplus \\ \oplus\underbrace{\mathfrak{g}\otimes\Omega^{n-1}(N)[-1]}_{\ni A^+}\oplus \underbrace{\mathfrak{g}\otimes \Omega^{2}(N)[1]}_{\ni B^+}\oplus\underbrace{\mathfrak{g}\otimes \Omega^n(N)[-2]}_{\ni c^+}.
\end{multline}
The symplectic form on $\calF_N$ is given by 
\begin{equation}
\label{symplect_YM}
\omega_N=\int_N\textnormal{tr}(\delta A\delta A^++\delta B\delta B^++\delta c\delta c^+),
\end{equation}
and the BV action is 
\begin{equation}
\calS_N(A,A^+,B,B^+,c,c^+)=\int_N\textnormal{tr}\left( BF_A+\frac{1}{2}B*B+A^+\dd_Ac+B^+[B,c]+\frac{1}{2}c^+[c,c]\right),
\end{equation}
where $*$ denotes the Hodge star operator, and $F_A=\dd_A+\frac{1}{2}[A,A]$ is the curvature of $A$. The cohomological vector field is then
\begin{multline}
Q_N=\int_N\textnormal{tr}\Bigg(\dd_Ac\frac{\delta}{\delta A}+[B,c]\frac{\delta}{\delta B}+\frac{1}{2}[c,c]\frac{\delta}{\delta c}+(\dd_AB+[A^+,c])\frac{\delta}{\delta A^+}\\+(F_A+*B+[B^+,c])\frac{\delta}{\delta B^+}+(\dd_AA^++[B,B^+]+[c,c^+])\frac{\delta}{\delta c^+}\Bigg).
\end{multline}
One can then check that 
$$(\calF_N,\omega_N,\calS_N)$$
defines a BV manifold together with the cohomological vector field $Q_N$ as the Hamiltonian vector field $\calS_N$.
%
\end{ex}

\begin{ex}[Abelian Chern--Simons]
\label{CS}
The classical abelian Chern--Simons theory (see \cite{AS,AS2} for the abelian and non-abelian case) is given by the action functional 
\begin{equation}
\label{cl_action_CS}
S(A)=\frac{1}{2}\int_MA\dd A,
\end{equation}
where $A$ is a connection on a trivial principal $U(1)$-bundle $P$ seen as a $1$-form, and $M$ a $3$-manifold.

The superspace of fields on the $3$-dimensional space-time manifold $M$ is given by 
$$\calF_M=\Omega(M)[1].$$
The fields corresponding to forms of degree $0,1,2,3$ will be denoted by $c,A,A^+,c^+$ respectively. The ghost numbers are $\gh(c)=1$, $\gh(A)=0$, $\gh(A^+)=-1$, and $\gh(c^+)=-2$. We will write $\mathsf{A}=c+A+A^++c^+$ for the \emph{BV superfield}. The symplectic form on $\calF_M$ is then given by 
\begin{equation}
\label{symplectic_CS}
\omega_M=\frac{1}{2}\int_M\delta\mathsf{A}\delta\mathsf{A}=\int_M(\delta c\delta c^++\delta A\delta A^+).
\end{equation}
The BV action is given by 
\begin{equation}
\label{action_CS}
\calS_M=\frac{1}{2}\int_M\mathsf{A}\dd\mathsf{A}=\frac{1}{2}\int_M(A\dd A+A^+\dd c+c\dd A^+),
\end{equation}
and the corresponding cohomological vector field is given by 
\begin{equation}
\label{cohom_CS}
Q_M=\int_M\dd\mathsf{A}\frac{\delta}{\delta \mathsf{A}}=\int_M\left(\dd c\frac{\delta}{\delta A}+\dd A\frac{\delta}{\delta A^+}+\dd A^+\frac{\delta}{\delta c^+}\right)
\end{equation}

\end{ex}

\begin{ex}[Poisson Sigma Model]
\label{PSM}
The following Sigma Model plays an important role for deformation quantization as it was shown in \cite{CF1}. One can show that Kontsevich's star product (see \cite{K}) can be written as the perturbative expansion of a path integral given in terms of the Poisson Sigma Model (\cite{CF1}). The Poisson Sigma Model (PSM, \cite{I,SS1}) consists of the following data: 
\begin{itemize}
\item{A connected, compact, oriented $2$-manifold $\Sigma$ (possibly with boundary),} 
\item{A finite-dimensional Poisson Manifold $(M,\pi)$, where $\pi$ denotes the Poisson bivector field,
}
\item{A space of fields, given as vector bundle maps $F_\Sigma=\Map(T\Sigma,T^*M)$. 
The fields are then of the form $(X,\eta)\in F_\Sigma$, where $X\colon \Sigma\to M$ is a map, and $\eta\in\Gamma(\Sigma,T^*\Sigma\otimes X^*T^*M)$.}
\item{An action functional 
\begin{equation}
\label{PSM_action}
S(X,\eta)=\int_{\Sigma}\left(\langle \eta,\dd X\rangle+\frac{1}{2}\langle \pi(X),\eta\eta\rangle\right),
\end{equation}
where $\langle\enspace,\enspace\rangle$ represents the pairing of the tangent- and cotangent space of $M$.
}
\end{itemize}
Let us now consider the setting with \emph{superfields}. For this we let 
$$\mathcal{F}_\Sigma=\Map(T[1]\Sigma,T^*[1]M),$$
and thus consider the superfields $\mathsf{X}\colon T[1]\Sigma\to M $ and $\boldsymbol{\eta}\in\Gamma(\mathsf{X}^*T^*[1]M)$. We can write the superfields in terms of fields and antifields by $\mathsf{X}=X+\eta^++\beta^+$, and $\boldsymbol{\eta}=\beta+\eta+X^+$. In particular, we can decompose each coordinate of the superfields into
\begin{align}
\mathsf{X}^{i}&=X^{i}+\theta^{\mu}\eta_\mu^{+ i}-\frac{1}{2}\theta^{\mu}\theta^\nu\beta_{\mu\nu}^{+ i},\\
\boldsymbol{\eta}_i&=\beta_i+\theta^\mu\eta_{i,\mu}+\frac{1}{2}\theta^\mu\theta^\nu X^+_{i,\mu\nu},
\end{align}
where $\beta$ denotes the ghost field, and again $\phi^+$ denotes the antifield of a field $\phi$.
Note that $X^{i},\beta_i\in\Omega^0(\Sigma)$, $\eta_\mu^{+i},\eta_{i,\mu}\in\Omega^1(\Sigma)$, and $\beta^{+i}_{\mu\nu},X^{+}_{i,\mu\nu}\in\Omega^2(\Sigma)$. Moreover, we have $\gh(\beta)=1$, $\gh(\beta^+)=-2$, $\gh(X)=0$, $\gh(X^+)=-1$, $\gh(\eta)=0$, and $\gh(\eta^+)=-1$ (again, one can compute this by using \eqref{field_antifield}, note also that here $d=2$). 
One can check that the BV action is of the same form as \eqref{PSM_action}, i.e. we have
\begin{equation}
\label{PSM_BV}
\calS_{\Sigma}(\mathsf{X},\boldsymbol{\eta})=\int_\Sigma\left(\langle \boldsymbol{\eta}, D\mathsf{X}\rangle+\frac{1}{2}\langle \pi(\mathsf{X}),\boldsymbol{\eta}\boldsymbol{\eta}\rangle\right),
\end{equation}
where $D=\theta^\mu\frac{\partial}{\partial \theta^\mu}$ is the differential on $T[1]\Sigma$. The symplectic form is given by 
\begin{equation}
\label{symplectic_PSM}
\omega_\Sigma=\int_\Sigma \delta\mathsf{X}\delta\boldsymbol{\eta}
\end{equation}
where $\delta$ is the de Rham differential on the space of fields $\calF_\Sigma$. One can check that the cohomological vector field $Q_\Sigma$ is then given by
\begin{equation}
\label{cohom_VF_PSM}
Q_\Sigma=\int_\Sigma\left(\delta\mathsf{X}^{i}\frac{\delta}{\delta\mathsf{X}^{i}}+\delta\boldsymbol{\eta}_i\frac{\delta}{\delta\boldsymbol{\eta}_i}\right)=\int_\Sigma \left\{\left(\dd \mathsf{X}^{i}+\pi^{ij}(\mathsf{X})\boldsymbol{\eta}_j\right)\frac{\delta}{\delta \mathsf{X}^{i}}+\left(\dd\boldsymbol{\eta}_i+\frac{1}{2}\partial_i\pi^{jk}(\mathsf{X})\boldsymbol{\eta}_j\boldsymbol{\eta}_i\right)\frac{\delta}{\delta \boldsymbol{\eta}_i}\right\},
\end{equation}
where we use the \emph{Einstein summation convention}. One can then check that 
$$(\calF_\Sigma,\omega_\Sigma,\calS_\Sigma)$$
is a BV manifold together with the Hamiltonian vector field $Q_\Sigma$ of $\calS_\Sigma$. 
\end{ex}

\subsection{Relaxed BV manifolds}\label{s:rBV}
In applications (e.g., to field theories on manifolds with boundary) it is too much to require $\iota_Q\omega=\dd \calS$ as often $\calS$ does not even possess a Hamiltonian vector field. At this level we can even drop the condition that $\omega$ is nondegenerate (which is also useful for applications to discretized field theories with boundary). We then come to the following relaxed definition.
\begin{defn}[Relaxed BV manifold]
A \emph{relaxed BV manifold} is a quadruple $(\calF,\omega,\calS,Q)$ where $\calF$ is a supermanifold, $\omega$ a closed odd $2$\ndash form of degree $-1$,
$\calS$ an even function of degree $0$ and $Q$ an odd vector field on $\calF$ of degree $1$ satisfying the equation
\[
[Q,Q]=0.
\]
\end{defn}
It may seem that we have dropped to much, but the theory is still quite interesting. First, we introduce
the ``check'' even $1$\ndash form
\[
\Check\alpha :=\iota_Q\omega-\dd \calS,
\]
which checks the failure of the first equation for a BV manifold. Then we denote by $\Check\omega$ its differential,
\[
\Check\omega:=\dd\Check\alpha=-\LL_Q\omega.
\]
Here $\LL_Q$ denotes the Lie derivative w.r.t.\ $Q$. Note that, since $Q$ is odd, Cartan's formula reads
\[
\LL_Q=\iota_Q\,\dd-\dd\,\iota_Q.
\]
The equation $[Q,Q]=0$ then immediately implies
\[
\LL_Q\Check\omega = 0.
\]
For the following it is also useful to remark the identies
\[
\LL_Q\,\dd=-\dd\,\LL_Q,
\qquad
\LL_Q\,\iota_Q=\iota_Q\,\LL_Q.
\]

\section{The BFV construction and symplectic reduction}
 Next we want to develop this formalism for gauge fixing for the case of manifolds with boundary. This construction, called the \emph{BV-BFV formalism}, was introduced in \cite{CMR1} in the classical setting and in \cite{CMR2} in the quantum setting.

\subsection{Induced BFV manifolds}
For applications it is now useful to move to the reduction of $(\calF,\Check\omega)$. Notice that
$\Check\omega$ is closed. As a consequence its kernel, i.e., the vector fields $X$ such that $\iota_X\Check\omega=0$,
form a Lie algebra. Moreover, $[Q,X]$ is in the kernel for every $X$ in the kernel. We leave these two statements
as a simple exercise in Cartan calculus. If the span of the vector fields in the kernel has (locally) constant dimension,
we then have an involutive distribution,\footnote{Recall that a regular distribution on a manifold $M$ is
a subbundle $D$ of its vector bundle $TM$. A section of the distribution $D$ is then a vector field X on $M$ such
that $X(x)\in D_x$ for all $x\in M$. 
A regular distribution $D$ is called involutive if for every two sections $X$ and $Y$ of $D$
their Lie bracket $[X,Y]$ is also a section of $D$. A regular distribution is called integrable if 
there is a chart around each point in which it is given as the span of the first coordinate vectors.
}
which by Frobenius theorem is then integrable.\footnote{In the infinite-dimensional case we have to be more careful. The correct condition is that the span of the vector fields in the kernel defines a subbundle of the tangent bundle. Moreover, Frobenius theorem holds in the Banach case, but not necessarily in the Fr\'echet case.} 
If the leaf space, which we denote by $\calF^\de$, is smooth we then have a uniquely defined even
$2$\ndash form $\omega^\de$ and a uniquely defined odd vector field $Q^\de$ such that
\[
\Check\omega=\pi^*\omega^\de\quad\text{and}\quad
\dd\pi\, Q = Q^\de,
\]
where $\pi$ is the natural projection $\calF\to\calF^\de$. Moreover, $Q^\de$ automatically satisfies $[Q^\de,Q^\de]=0$ and $\LL_{Q^\de}\omega^\de = 0$.
In particular, $\omega^\de$ is symplectic and $Q^\de$ is a symplectic vector field.

If we make the assumption that $Q$ is not only symplectic but also Hamiltonian, i.e., that there
is an odd function $\calS^\de$ such that
\[
\iota_{Q^\de}\omega^\de=\dd \calS^\de,
\]
then the triple $(\calF^\de,\omega^\de,\calS^\de)$ is an example of a \emph{BFV manifold} as in the following definition:

\begin{defn} [BFV manifold]
\label{BFV_manifold}
A \emph{BFV manifold} is the same as a BV manifold (defintion \ref{BV_manifold}) with
shifted parities on the symplectic form and the function. 
\end{defn}

Notice that in principle $\calS^\de$ is defined up to a constant; however, since the only odd constant is zero, $S^\de$ is uniquely defined (if it exists). Also notice that, defining 
\[
\Check \calS:=\pi^*\calS^\de
\]
yields
\[
\iota_{\Check Q}\,\Check\omega = \Check \calS.
\]
Vice versa, if we assume the last equation, then we see that $X(\Check \calS)=0$ for every vector field $X$
in the kernel of $\Check\omega$. Hence, $\Check \calS$ is basic and there is a uniquely defined function $\calS^\de$,
which is the Hamiltonian for $Q^\de$.

Finally, note that in the case when $\iota_Q\omega=\dd \calS$, we have that $\Check\omega=0$, so
$\calF^\de$ is a point (or a collection of points if $\calF$ is not connected). 

\subsection{Ghost number}
\label{ghostnr}
As already mentioned, in some applications to field theories one has an additional $\bbZ$\ndash grading that helps bookbooking. It is called
ghost number, with the idea that physical fields have ghost number zero. We make a short digression on it (even though it is not needed for the following and is not always available in field theory).

The function $\calS$ is assigned ghost number zero (with the idea that it is a physical function). The vector field $Q$ is assigned ghost number $+1$. This has the consequence that $Q$ defines a differential on the algebra of functions
(and for this reason it is usually called a cohomological vector field). As a consequence $\omega$ must be assigned
ghost number $-1$. If a BV manifold is defined in the presence of a $\bbZ$\ndash grading, these assignements are always 
understood.\footnote{An example of this setting is the BRST formalism. In this case one assigns ghost number zero to the functions on $F$
and ghost number $+1$
to the linear coordinates on $\Pi\frg$, so that the BRST operator $\delta$ acquires ghost number $+1$. The odd cotangent bundle is then shifted by $-1$,
meaning that the momenta corresponding to coordinates on $F$ receive ghost number $-1$, whereas the momenta corresponding to the coordinates on $\frg$ receive
ghost number $-2$ ($=-1-1$), so that the symplectic form acquires ghost number $-1$. Observe that $\calS$ then automatically has ghost number $0$.}

It follows that $\Check\alpha$, $\Check\omega$ and $\omega^\de$ have ghost number zero (so $\omega^\de$ may be interpreted as a physical symplectic form); $Q^\de$ has ghost number $+1$ (so it is still cohomological). Finally,
the function $\calS^\de$ has ghost number $+1$. Note that, if a BFV manifold is defined in the presence of a $\bbZ$\ndash grading, these assignements are always understood.

Typically the ghost number is a local notion. This means that there is an even vector field $E$ of ghost number zero,
called the Euler vector field,
 such that, on homogenous differential forms or vector fields, we have
\[
\LL_E\, \xi = x\,\xi,
\]
where $x$ is the ghost number of $\xi$. A consequence is that $Q$ is always 
$\Check\omega$\ndash Hamiltonian (even if $\Check\omega$ is degenerate), with
Hamiltonian function
\[
\Check \calS = \iota_{E}\,\iota_{Q}\,\Check\omega,
\]
as a simple application of Cartan's calculus shows. 

As we have already remarked this implies that, if reduction is possible, also $Q^\de$ is Hamiltonian.
All this is actually compatible with the $\bbZ$\ndash grading. In fact, one can immediately see that
$[E,X]$ is in the kernel of $\Check\omega$\
for every $X$ in the kernel; hence, $E$ is projectable. One also easily realizes that its projection is the Euler
vector field on $\calF^\de$, which we will denote by $E^\de$ to avoid confusion.
Finally, we have that the Hamiltonian for $Q^\de$ is 
\[
\calS^\de = \iota_{E^\de}\,\iota_{Q^\de}\,\omega^\de.
\]

\begin{ex}
\label{exBFV_formalism}
Let $(M,\omega)$ be a symplectic manifold. Consider a map $\phi\in C^\infty(M)$, and the submanifold $C\subset M$ given by $C:=\phi^{-1}(0)$. In particular, the submanifold $C$ is coisotropic (see Remark \ref{iso_coiso}). Here $\phi$ is actually a moment map as in Definition \ref{momentmap}. Denote by $X_\phi$ the Hamiltonian vector field of $\phi$, and observe that
$$\iota_{X_\phi}\omega\Big\vert_{C}=\dd\phi\Big\vert_C=0.$$
Moreover, the kernel of $\omega$ is generated (as a $C^\infty(M)$-module) by the vector field $X_\phi$. Thus we can define the reduction $\underline{C}:=C/\langle X_\phi\rangle=C/\ker\omega$ and hence 
$$C^\infty(C)=C^\infty(M)/\langle\phi\rangle.$$
Then we get $C^\infty(\underline{C})=(C^\infty(C))^X$. If we denote $I:=\langle \phi\rangle\subset C^\infty(M)$, we can define $$N(I):=\{f\in C^\infty(M)\mid \{f,I\}\subset I\}\subset C^\infty(M),$$ 
where $\{\enspace,\enspace\}$ denotes the Poisson bracket on $C^\infty(M)$, 
which gives us then $$C^\infty(\underline{C})=N(I)/I.$$

\begin{rem}
Above we defined this algebra as first quotienting by the ideal and then taking the invariant part, now we take, equivalently, functions that are invariant modulo the ideal and then quotient by the ideal.
\end{rem}

To describe this in the BFV formalism, we add new variables $c$ and $b$ with $\gh(c)=1$ and $\gh(b)=-1$ respectively, and consider the supermanifold $\calM=M\times T^*[1]\R$ endowed with the symplectic form $\widetilde{\omega}=\omega+\dd b\dd c$ together with the action $$\calS=c\phi,$$ which clearly satisfies $\{\calS,\calS\}=0$. 

Moreover, we consider the Hamiltonian vector field $Q$ of $\calS$ given by $Q=\{\enspace, \calS\}$ and hence we have $Q(b)=\{b,c\phi\}=\phi$ and $Q(f)=c\{f,\phi\}$ for $f\in C^\infty(M)$. This shows that the $Q$-cohomology in degree zero is given by the quotient $N(I)/I$, i.e 
\begin{equation}
\label{cohom_Q}
H^0_Q(\calM)=N(I)/I.
\end{equation}
In fact, consider a function $F=f+bcg$ of ghost number zero, where $f,g\in C^\infty(M)$. Then $$Q(F)=c\{f,\phi\}+\phi cg$$ and if $Q(F)=0$, we get $\{f,\phi\}=-\phi g$ and thus $f\in N(I)$. Consider now a function $\Psi:=b h$ of ghost number minus one, for some $h\in C^\infty(M)$. Then $$Q(\Psi)=\underbrace{\phi h}_{\in I}-bc\{h,\phi\}.$$ 
One can generalize this situation by considering $\phi$ to be given as several maps $\phi_1,...,\phi_n\in C^\infty(M)$, such that $\phi=(\phi_1,...,\phi_n)\colon M\to \R^n$. Assume that the zero in $\R^n$ is a regular value. Then we get $\{\phi_i,\phi_j\}\Big\vert_C=0$, which is the coisotropic condition of $C$. Note that now we have $$\ker\omega\Big\vert_C=\langle X_1,...,X_n\rangle,$$ where $X_j$ denotes the Hamiltonian vector field of $\phi_j$ as before, and $$\underline{C}:=C/\langle X_1,...,X_n\rangle.$$ We can observe that $\underline{C}$ is now symplectic and moreover $C^\infty(\underline{C})=N(I)/I$. We can add again extra variables $c^{i}$ and $b_i$ with ghost numbers as before and consider the supermanifold $\calM=M\times T^*[1]\R^n$ endowed with the symplectic form $$\widetilde{\omega}=\omega+\sum_i\dd b_i\dd c^{i}.$$
For the next step one has to be careful, since $\calS=\sum_i c_i\phi^{i}$ in general will not satisfy $\{\calS,\calS\}=0$. 
However, there is a lemma, which tells us that there is always a correction term $O(b)$ such that 
\begin{equation}
\{\calS+O(b),\calS+O(b)\}=0.
\end{equation}
This shows the generalized version of \eqref{cohom_Q}. Moreover, there are functions $f_{ij}^k\in C^\infty(M)$, such that $$\{\phi_i,\phi_j\}=\sum_k f_{ij}^k\phi_k.$$ For the special case where the $f_{ij}^k$ are constant, we have $$\calS=\sum_i c_i\phi^{i}+\frac{1}{2}\sum_{ijk}f_{ij}^kb_kc^{i}c^{j}.$$ In this case, the $f_{ij}^k$ are the \emph{structure constants} of some Lie algebra $\mathfrak{g}$, which is the setting of BRST.
\end{ex}

\subsection{The relaxed master equation}
If $Q$ is Hamiltonian, w.r.t.\ $\Check\omega$ (which happens, e.g., in the $\bbZ$\ndash graded case), i.e., if
\[
\iota_{\Check Q}\,\Check\omega = \dd\Check \calS,
\]
we can express the failure of the CME in terms of the uniquely defined odd function $\Check \calS$.

The derivation is quite simple. If we apply $\LL_Q$ to the equation defining $\Check\alpha$ and use the rules mentioned
at the end of Section~\ref{s:rBV}, we get
\[
\LL_Q\Check\alpha = \iota_Q\LL_Q\omega + \dd\LL_Q\calS.
\]
We have already seen that $\LL_Q\omega=-\dd\Check\alpha$. Using the Cartan formula
$\LL_Q\Check\alpha=\iota_Q\dd\Check\alpha-\dd\iota_Q\Check\alpha$, we get
\[
\dd\LL_Q\calS =2\iota_Q\dd\Check\alpha-\dd\iota_Q\Check\alpha.
\]
However, from $\iota_Q\dd\Check\alpha=\dd\Check \calS$ we obtain
\[
\dd\LL_Q\calS =2\dd\Check \calS-\dd\iota_Q\Check\alpha.
\]
Again, using the fact that the only odd constant function is zero, we finally get
\[
\LL_Q\calS =2\Check \calS-\iota_Q\Check\alpha.
\]
Notice that when $Q$ is the Hamiltonian vector field of $\calS$, we have $(\calS,\calS)=\LL_Q\calS$, so the right hand side
of this equation may be seen as the error in the CME. 

However, there is another way of expressing the CME. Namely, we also have $(\calS,\calS)=\iota_Q\iota_Q\omega$
if $Q$ is the Hamiltonian vector field of $\calS$. Otherwise we can compute this by applying $\iota_Q$ to the definition
of $\Check\alpha$; namely,
\[
\iota_Q\Check\alpha = \iota_Q\iota_Q\omega-\LL_Q\calS.
\]
Inserting the previous expression for $\LL_Q\calS$, we finally get
\[
\iota_Q\iota_Q\omega = 2\Check \calS.
\]
Notice that in this equation the $1$\ndash form $\Check\alpha$ no longer appears.
If the reduction $\calF^\de$ is smooth, then we can rewrite this equation as
\begin{equation}
\label{mCME}
\iota_Q\iota_Q\omega = 2\pi^*\calS^\de.
\end{equation}

\begin{rem}
Equation \eqref{mCME} is called the \emph{modified CME (mCME)}.
\end{rem}

\section{The BV-BFV formalism}
We get an improvement of the picture above if we in addition require $\Check\alpha$ to descend to the
leaf space $\calF^\de$. In this case there is a unique even $1$\ndash form $\alpha^\de$ on $\calF^\de$
such that $\Check\alpha=\pi^*\alpha^\de$. We also have
$\omega^\de=\dd\alpha^\de$ and $\pi^*\alpha^\de=\iota_Q\omega-\dd \calS$.
We can summarize this with the notion of an exact BV-BFV pair:
\begin{defn}[Exact BV-BFV pair]
\label{exactBVBFV}
An \emph{exact BV-BFV pair} consists of a nonuple $$(\calF,\omega,Q,\calS,\calF^\de,\alpha^\de,Q^\de,\calS^\de,\pi)$$ where
$\calF$ and $\calF^\de$ are supermanifolds, 
$\pi\colon\calF\to\calF^\de$ is a surjective submersion,
$Q$ and $Q^\de$ are $\pi$\ndash related odd vector fields (on $\calF$ and on $\calF^\de$, 
respectively) of degree $1$, $\omega$ is an odd symplectic form on $\calF$ of degree $-1$, $\omega^\de:=\dd\alpha^\de$ 
is an even symplectic form on $\calF^\de$ of degree $0$, $\calS$ is an even function on $\calF$ of degree $0$, $\calS^\de$ is the Hamiltonian function
for $Q^\de$ of degree $1$, and we have $[Q,Q]=0$ and
\begin{equation}
\iota_Q\omega=\dd \calS+\pi^*\alpha^\de.
\end{equation}
\end{defn}
\begin{rem}
Notice that in particular $(\calF^\de,\omega^\de,Q^\de,\calS^\de)$ is a BFV manifold. 
\end{rem}

For applications one can actually drop the condition that $\omega$ is nondegenerate. What is really needed in the quantum version is the existence of an integrable distribution $\calP$ on $\calF^\de$ with smooth leaf space $\mathcal{B}$ such that the restriction of 
$\alpha^\de$ 
to each leaf vanishes and such that the restriction of $\omega$ to each fiber of $\calF\to\calB$, is nondegenerate. Note that in particular the leaves of $\calP$ are $\check{\omega}$-isotropic. For the purposes
of quantization one requires them to be actually Lagrangian.

\begin{ex}[Yang--Mills]
We describe the BFV structure of Example \ref{YM}. Let us denote the pullback to the boundary of forms $A,B,A^+,c$ by the same letters. The space of boundary fields is the quotient space of the pullback of $\calF_N$ to the boundary over the kernel of the form $\dd \check{\alpha}_{\partial N}$ 
\begin{equation}
\label{boundary_YM}
\calF^\partial_{\partial N}=\underbrace{\mathfrak{g}\otimes\Omega^1(\partial N)[1]}_{\ni A}\oplus \underbrace{\mathfrak{g}\otimes\Omega^{n-2}(\partial N)[n-2]}_{\ni B}\oplus \underbrace{\mathfrak{g}\otimes\Omega^0(\partial N)[1]}_{\ni c}\oplus\underbrace{\mathfrak{g}\otimes\Omega^{n-2}(N)[n-2]}_{\ni A^+} ,
\end{equation}
where $\gh(A)=0$, $\gh(B)=0$, $\gh(c)=1$, and $\gh(A^+)=-1$. The structure of an exact symplectic manifold on $\calF_{\partial N}$ is given by 
\begin{align*}
\alpha^\partial_{\partial N}&=\int_{\partial N}\textnormal{tr}(B\delta A+A^+\delta c),\\
\omega^\partial_{\partial N}&=\int_{\partial N}\textnormal{tr}(\delta B\delta A+\delta A^+\delta c).
\end{align*}
The boundary action and its corresponding Hamiltonian vector field with respect to $\omega^\partial_{\partial N}$ are given by 
\begin{align*}
\calS^\partial_{\partial N}&=\int_{\partial N}\textnormal{tr}\left(B\dd_A c+\frac{1}{2}A^+[c,c]\right)\\
Q^\partial_{\partial N}&=\int_{\partial N}\textnormal{tr}\left(\dd_Ac\frac{\delta}{\delta A}+[B,c]\frac{\delta}{\delta B}+(\dd_AB+[A^+,c])\frac{\delta}{\delta A^+}+\frac{1}{2}[c,c]\frac{\delta}{\delta c}\right)
\end{align*}
One can then check that 
$$(\calF_N, \omega_N,Q_N, \calS_N, \calF^\partial_{\partial N}, \alpha^\partial_{\partial N}, Q^\partial_{\partial N}, \calS^{\partial}_{\partial N}, \pi)$$
is an exact BV-BFV pair, and $(\calF^\partial_{\partial N},\omega^\partial_{\partial N}, Q^\partial_{\partial N},\calS^\partial_{\partial N})$ is a BFV manifold.
\end{ex}

\begin{ex}[Abelian Chern--Simons]
We describe the BFV structure of Example \ref{CS}. Note that the boundary fields are pullbakcs of the bulk fields to the boundary and thus the space of boundary fields is given by 
\begin{equation}
\label{boundary_fields_CS}
\calF^\partial_{\partial M}=\Omega(\partial M)[1].
\end{equation}
Let us denote the pullback to the boundary of the $0,1,2$ forms $c,A,A^+$ respectively by the same letters, and assign $\gh(c)=1$, $\gh(A)=0$, and $\gh(A^+)=-1$. Moreover, we denote the \emph{superfield} $\mathsf{A}=c++A+A^++c^+$ on the boundary by $\mathbb{A}$.
The $1$-form $\alpha^\partial_{\partial M}$, the symplectic structure, the cohomological vector field and the boundary action are then given by 
\begin{align}
\alpha_{\partial M}^\partial&=\frac{1}{2}\int_{\partial N}\mathbb{A}\delta\mathbb{A}=\frac{1}{2}\int_{\partial M}(A\delta A+c\delta A^++A^+\delta c)\\
\omega^\partial_{\partial M}&=\frac{1}{2}\int_{\partial N}\delta\mathbb{A}\delta\mathbb{A}=\int_{\partial M}\left(\frac{1}{2}\delta A\delta A+\delta c\delta A^+\right)\\
Q^\partial_{\partial M}&=\int_{\partial M}\dd\mathbb{A}\frac{\delta}{\delta \mathbb{A}}=\int_{\partial M}\left(\dd c\frac{\delta}{\delta A}+\dd A\frac{\delta}{\delta A^+}\right)\\
\calS^\partial_{\partial M}&=\frac{1}{2}\int_{\partial M}\mathbb{A}\delta\mathbb{A}=\int_{\partial M}c\dd A.
\end{align}
One can then check that 
$$(\calF_M, \omega_M,Q_M, \calS_M, \calF^\partial_{\partial M}, \alpha^\partial_{\partial M}, Q^\partial_{\partial M}, \calS^{\partial}_{\partial M}, \pi)$$
is an exact BV-BFV pair, and $(\calF^\partial_{\partial M},\omega^\partial_{\partial M}, Q^\partial_{\partial M},\calS^\partial_{\partial M})$ is a BFV manifold.
\end{ex}

\begin{ex}[Poisson Sigma Model]
We describe the BFV structure of Example \ref{PSM}. The space of boundary superfields is given by 
\begin{equation}
\label{boundary_fields_PSM}
\calF^\partial_{\partial \Sigma}=\Map(T[1]\partial\Sigma,T[1]^*M).
\end{equation}
We denote the boundary fields by $\mathbb{X}$ and $\mathbb{E}$. Then the BFV structure for the Poisson Sigma Model is given by 
\begin{align}
\alpha^\partial_{\partial \Sigma}&=\int_{\partial\Sigma}\mathbb{E}_i\delta\mathbb{X}^i,\\
\omega_{\partial \Sigma}^\partial&=\int_{\partial\Sigma}\delta\mathbb{E}_i\delta\mathbb{X}^i,\\
Q^\partial_{\partial \Sigma}&=\int_{\partial\Sigma}\left\{\left(\dd\mathbb{X}^{i}+\pi^{ij}(\mathbb{X})\mathbb{E}_i\mathbb{E}_j\right)\frac{\delta}{\delta\mathbb{X}^{i}}+\left(\dd\mathbb{E}_i+\frac{1}{2}\partial_i\pi^{jk}(\mathbb{X})\mathbb{E}_j\mathbb{E}_k\right)\frac{\delta}{\delta\mathbb{E}_i}\right\},\\
\calS^\partial_{\partial\Sigma}&=\int_{\partial \Sigma}\left(\mathbb{E}_i\dd\mathbb{X}^{i}+\frac{1}{2}\pi^{ij}(\mathbb{X})\mathbb{E}_i\mathbb{E}_j\right).
\end{align}
Then one can check that 
$$(\calF_\Sigma, \omega_\Sigma,Q_\Sigma, \calS_\Sigma \calF^\partial_{\partial \Sigma}, \alpha^\partial_{\partial \Sigma}, Q^\partial_{\partial \Sigma}, \calS^{\partial}_{\partial \Sigma}, \pi)$$
is an exact BV-BFV pair and $(\calF^\partial_{\partial \Sigma},\omega^\partial_{\partial \Sigma}, Q^\partial_{\partial \Sigma},\calS^\partial_{\partial \Sigma})$ is a BFV manifold.
\end{ex}


\subsection{Gauge transformations}
Suppose we have an exact BV-BFV pair as in Definiton \ref{exactBVBFV}.
If $f$ is an even function on $\calF^\de$, we may change $\alpha^\de$ to
\[
(\alpha^\de)^f=\alpha^\de+\dd f
\]
without changing the symplectic form $\omega^\de$. The last equation defining the BV-BFV structure is preserved
if we change $\calS$ to 
\[
\calS^f = \calS-\pi^*f.
\]
We then get another exact BV-BFV pair by replacing $\alpha^\de$ and $\calS$ with $(\alpha^\de)^f$ and $\calS^f$.
It makes sense to consider the two BV-BFV pairs as equivalent. We then regard this as a gauge transformation.
In applications to field theory, this corresponds to changing the action by a boundary term. In the quantum version, it will correspond to chainging the wave function by a phase.

\begin{rem}
A first consequence of this is that in the quantum version mentioned before,
we should just require that we have one representative of $\alpha^\de$ in the gauge class that vanishes
on the leaves of the foliation. 

A second consequence is that we have a natural generalization to more general (i.e., non exact) BV-BFV pairs
where we require $\omega^\de$ to be $\dd\alpha^\de$ only locally.
\end{rem}

\begin{defn}[BV-BFV$_\hbar$ pair]
For a fixed $\hbar\not=0$, we define a \emph{BV-BFV$_\hbar$ pair} as
a hendecuple $(\calF,\calE,\omega,Q,\sigma,\calF^\de,\calE^\de,\theta^\de,Q^\de,\calS^\de,\pi)$ where 
$\calF$ and $\calF^\de$ are supermanifolds, 
$\pi\colon\calF\to\calF^\de$ is a surjective submersion,
$Q$ and $Q^\de$ are 
$\pi$\ndash related 
odd vector fields (on $\calF$ and on $\calF^\de$, 
respectively), $\omega$ is an odd symplectic form on $\calF$, 
$S^\de$ is the Hamiltonian function
for $Q^\de$, $\calE^\de$ is a $U(1)$\ndash bundle
over $\calF^\de$,  $\calE$ is the pullback of $\calE^\de$ by 
$\pi$,\footnote{So we could have avoided putting $\calE$ in the list of data but we could not resist using the word hendecuple.}  $\sigma$ is an even section of $\calE$, $\theta^\de$ is a connection on $\calE^\de$
such that its curvature $R^\de$ is a symplectic form, and we have
\[
[Q,Q]=0\quad\text{and}\quad
\frac{\I}{\hbar}\iota_Q\omega\,\sigma= \DD\,\sigma,
\]
where $\DD$ is the covariant derivative with respect to the pullback connection $\pi^*\theta^\de$.
\end{defn}

An exact BV-BFV pair is also a BV-BFV$_\hbar$ pair, for all $\hbar\not=0$, by choosing $\calE^\de$ to be the
trivial $U(1)$\ndash bundle and setting $\sigma=\ee^{\frac\ii\hbar \calS}$ and $\theta^\de=\frac\ii\hbar\alpha^\de+\dd u$,
with $u$ the coordinate on $U(1)$. We then get $R^\de=\frac\ii\hbar\omega^\de$.

Notice that again we may allow some degeneracy for $\omega$. Namely, we assume that we have
an $\check{\omega}$-Lagrangian integrable distribution $\calP$ on $\calF^\de$, such that the restriction of $\theta^\de$ to each leaf is trivial, and require that
the restriction of $\omega$ to each fiber of $\calF\to\calB$, where
$\calB$ is the leaf space, be nondegenerate.

Why did we make this digression? The point is that this may be needed. Recall that we started with a relaxed
BV manifold $(\calF,\omega,\calS,Q)$ and assumed that the leaf space $\calF^\de$ of the distribution given
by the kernel of $\Check\omega$ was smooth. To get an exact BV-BFV pair we also had to assume that
$\Check\alpha$ descended to a $1$\ndash form $\alpha^\de$ on the leaf space. This last assumption is a bit strong. A weaker version is the following. We take the trivial $U(1)$\ndash bundle $\calE$ over $\calF$ and regard
$\Check\theta:=\frac\ii\hbar\Check\alpha+\dd u$ as a connection. The weaker condition now is that
$\calE$ reduces to a smooth $U(1)$\ndash bundle $\calE^\de$ over $\calF^\de$ and that $\Check\theta^\de$
descends to a connection $\theta^\de$. This produces a BV-BFV$_\hbar$ pair.

Note that there is an even further generalization as we may start with relaxed BV$_\hbar$ data, namely
a sextuple $(\calF,\calE,\omega,\sigma,Q,\Check\theta)$ where $\calE$ is a $U(1)$\ndash bundle
over $\calF$, $\Check\theta$ a connection $1$\ndash form  and the relaxed BV equation is replaced by
$\iota_Q\omega\,\sigma= \Check\DD\,\sigma$ where $\Check\DD$ is the covariant derivative
with respect to $\Check\theta$.\footnote{This setting actually occurs in field theory, e.g.,
in the case of a charged particle in an external electromagnetic field or in the WZWN model.}

%

\subsection{Parametrizations}
It is often the case that the BV action depends on parameters but in a BV irrelevant way (this is, e.g., the case in field theory when one expands around background fields). More precisely, we assume that we have a
BV manifold $(\calF,\omega,\calS_0)$ with $\calS_0$ depending on some parameter in a parameter space $P$ (i.e.,
$\calS_0$ is a function on $\calF\times P$). We also assume that the variations of $\calS_0$ on $P$ are trivial in the
$(\calS_0,\enspace )$\ndash cohomology. More precisely, we assume
\[
D\calS_0+(\calS_0,\calS_1)=0,
\]
where $D$ denotes the de~Rham differential on $P$ and $\calS_1$ is an odd function on $\calF$ which is also a
$1$\ndash form on $P$. By setting $\calS=\calS_0+\calS_1$, we can summarize the above equation and the CME in the single equation
\begin{equation}
\label{DCME}
D\calS+\frac12(\calS,\calS)=0
\end{equation}
which we call the \emph{differential Classical Master Equation (dCME)}. 
\begin{rem}
Equation \eqref{DCME} also contains the ``integrability condition'' $D\calS_1+\frac{1}{2}(\calS_1,\calS_1)=0$.
\end{rem}

The above setting may be generalized in two ways. The first is by replacing the de~Rham complex
$(\Omega(P),D)$ by a differential graded commutative algebra $(\calA,D)$. The second is by allowing
$\calS$ to be a sum $\sum_{i=0}^\infty \calS_i$, where $\calS_i$ is a function of the same parity as $i$ on $\calF$
and of degree $i$ in $\calA$.

For notational simplicity it is useful to think of $D$ as a vector field, so we assume $\calA$ to be the algebra
of function on a graded manifold $\calP$ (e.g., $T^*[1]P$ in the starting case). In other words,
we now work on $\calF_\calP=\calF\times\calP$. The BV form $\omega$ is now pulled back to $\calF_\calP$
and the vector field $D$ is extended to it: note that $\iota_D\omega=0$. The BV action $\calS$ is an even function
on $\calF_\calP$. If we have a compatible $\bbZ$\ndash grading on $\calF$, then $\calS$ has total degree zero
(i.e., $\calS_i$ has degree $i$ on $\calP$ and $-i$ on $\calF$). 

\begin{defn}[Differential BV manifold]
The quadruple
$(\calF,\omega,\calP,D,\calS)$ is called a \emph{differential BV manifold}.
\end{defn}

We then introduce the Hamiltonian vector field $Q$ of $\calS$
\[
\iota_Q\omega = \dd_\calF \calS,
\]
where $\dd_\calF$ denotes the de~Rham differential on $\calF$. Note that $Q$ is a vector field
on $\calF$ parametrized by $\calP$. As usual it follows that $[Q,Q]$ is the Hamiltonian vector field
of $(\calS,\calS)$:
\[
\iota_{[Q,Q]}\omega = \dd_\calF(\calS,\calS).
\]
On the other hand, since $\calS$ depends also on $\calP$ we have
\[
\LL_Q\omega = \dd_\calP\dd_\calF \calS,
\]
where $\dd_\calP$ is the de~Rham differential on $\calP$ and $\LL$ denotes the Lie derivative of
$\calF_\calP$. We then have
\[
\iota_{[Q,D]}\omega=[\LL_Q,\iota_\calS]\omega=-\iota_D\LL_Q\omega=
-\dd_\calF\iota_D\dd_\calP \calS = -\dd_\calF D\calS.
\]
Thus, the dCME is equivalent to
\[
[Q,Q]-2[Q,D]=0
\]
which is also equivalent to
\[
[Q-D,Q-D]=0.
\]
Thus we come to the equivalent definition of  a {differential BV manifold}
as a
$(\calF,\omega,\calP,D,\calS,Q)$ where $\omega$ is an odd symplectic form $\calF$,
$D$ is a cohomological vector field on $\calP$, $\calS$ is an even function $\calF\times\calP$ and
and $Q$ is an odd vector field on $\calF$ parametrized by $Q$ satisfying
\[
\iota_Q\omega=\dd_\calF \calS,
\qquad
[Q-D,Q-D]=0.
\]

We then define a \text{relaxed differential BV manifold} by the same data by possibly dropping
the nondegeneracy condition of $\omega$ and requiring only $[Q-D,Q-D]=0$. We introduce

\[
\Check\alpha :=\iota_Q\omega-\dd_\calF \calS
\]
and
\[
\Check\omega:=\dd_\calF\Check\alpha=-\LL_Q^\calF\omega,
\]
where $\LL^\calF$ denotes the Lie derivative on $\calF$. Note that  in general $\alpha^\de$ and
$\omega^\de$ might depend parametrically on $\calP$. However, we now wish to reduce $\calF$ with respect to the kernel of $\omega$ on $\calF$. We denote
by $\calF^\de$ the reduced space and by $\omega^\de$ the reduced symplectic form. Note that in general
$\omega^\de$ might depend parametrically on $\calP$. However, we have
\[
\iota_D\Check\alpha=\iota_D\Check\omega=0.
\]
We now want to check that $Q$ is projectable. Let $Y$ be a vector field on $\calF$ in the kernel of $\omega$.
Note that $[Y,Q]=[Y,Q]_\calF$ as neither $Q$ not $Y$ have component along $\calP$ .

\section{Local field theory}
\label{local_field_theory}
The main application of the BV formalism is in the context of (perturbative) quantum field theory.

\begin{defn}[BV local field theory]
A \emph{BV local field theory} in $d$ dimensions is the assignment of a BV manifold
$(\calF_M,\omega_M,\calS_M)$ to each $d$\ndash manifold $M$ (possibly with some required structure)
where $\calF_M$ is a space of local fields (functions, maps, sections of bundles, connections) and both
$\omega_M$ and $\calS_M$ are local, in the sense that they are integrals over $M$ of densities depending
at each point on finitely many jets of the fields at that point. It then follows that the Hamiltonian vector field $Q_M$ of $\calS_M$,
$\iota_{Q_M}\omega_M=\dd \calS_M$, is also local. The local functional $\calS_M$ is called the \emph{BV action} (as in Definition \ref{BV_manifold}).
\end{defn}

If $M$ is a compact $d$\ndash manifold with boundary we can extend to it the definition of $\calF_M$,
$\omega_M$, $\calS_M$ and $Q_M$. We thus associate a relaxed BV manifold
$(\calF_M,\omega_M,\calS_M,Q_M)$ to each compact $d$\ndash manifold with boundary $M$.
We call this the naive choice. Other choices are obtained by allowing in $\calF_M$ only fields that respect certain fixed
boundary conditions. 

We can apply the general procedure to the relaxed BV manifold
$$(\calF_M,\omega_M,\calS_M,Q_M).$$ The main remark here is that, by locality, the kernel of
$\Check\omega_M$ contains all bulk fields. This has two consequences.
The first is that the induced BFV manifold depends only on boundary data; we will therefore
denote it by $(\calF^\de_{\de M},\omega^\de_{\de M},\calS^\de_{\de M})$. The second is that
its construction factors through an intermediate step. We denote by $\Tilde\calF_{\de M}$ the space of transversal
jets of fields at the boundary and by $\Tilde\pi_M$ the projection $\calF_M\to \Tilde\calF_{\de M}$. Again by locality,
we have that $Q_M$ is $\Tilde\pi_M$\ndash projectable and that $\Check\alpha_M$ is
$\Tilde\pi_M$\ndash basic. We then have a uniquely defined odd vector field $\Tilde Q_{\de M}$
and a uniquely
defined even $1$\ndash form $\Tilde\alpha_{\de M}$ on $\Tilde\calF_{\de M}$
such that $Q_M$ and $\Tilde Q_{\de M}$ are $\Tilde\pi_M$\ndash related and 
$\Check\alpha_M=\Tilde\pi_M^*\Tilde\alpha_{\de M}$. Moreover, we have
\[
[\Tilde Q_{\de M},\Tilde Q_{\de M}]=0,\qquad
\iota_{Q_M}\omega_M=\dd \calS_M + \Tilde\pi_M^*\Tilde\alpha_{\de M}.
\]
We then get $\calF^\de_{\de M}$, if smooth, as the reduction of $\Tilde\calF_{\de M}$ by the kernel
of $\Tilde\omega_{\de M}:=\dd\Tilde\alpha_{\de M}$.

Two remarks are now in order. The first is that, if the boundary of $M$ is the disjoint union of
two manifolds $\de_1M$ and $\de_2M$, we then have
\[
\Tilde\calF_{\de M}=\Tilde\calF_{\de_1 M}\times \Tilde\calF_{\de_2 M},\qquad
\calF^\de_{\de M}=\calF^\de_{\de_1 M}\times\calF^\de_{\de_2 M}.
\]
In particular, if $\Sigma$ is an oriented compact $(d-1)$\ndash manifold, we then have, for every interval $I$,
\[
\Tilde\calF_{\de(\Sigma\times I)}=\Tilde\calF_{\Sigma}\times \Tilde\calF_{\Sigma^\text{op}},\qquad
\calF^\de_{\de (\Sigma\times I)}=\calF^\de_{\Sigma}\times\calF^\de_{\Sigma^\text{op}},
\]
where $\Sigma^\text{op}$ denotes $\Sigma$ with the opposite orientation. This way we can associate
``boundary data'' $\Tilde\calF_{\Sigma}$ and $\calF^\de_{\Sigma}$ intrinsically to every
oriented compact $(d-1)$\ndash manifold $\Sigma$.

The second remark is that changing the action functional by a boundary term leads to physically equivalent theories.
This means that the BV theory
is equipped with gauge transformations
\[
\Tilde\alpha_{\de M}\mapsto (\Tilde\alpha_{\de M})^f=\Tilde\alpha_{\de M}+\dd f,\qquad
\calS_M\mapsto \calS^f_{M} = \calS_M-\Tilde\pi_M^*f,
\]
for every $f$ in $\Tilde\calF_{\de M}$. In some cases, one may use gauge transformations to make
$\Tilde\alpha_{\de M}$ basic as a $1$\ndash form (not just as a connection).

%
%

\subsection{The BFV formalism}
BFV manifolds, i.e., triples $(\calF,\omega,\calS)$ where $\calF$ is a supermanifold, $\omega$ an even symplectic form
and $\calS$ an odd function satisfying $\{\calS,\calS\}=0$, with $\{\ ,\ \}$ the Poisson bracket associated to $\omega$,
have an important application in symplectic geometry.

Namely, let $(\calF^0,\omega^0)$ be a symplectic manifold (here $\calF^0$ might also be a supermanifold, but
$\omega^0$ is anyway assumed to be even). Let $\calC\subset \calF^0$ be a coistropic submanifold; i.e., the Hamiltonian
vector field of every function vanishing on $\calC$ is tangent to $\calC$ (see Definition \ref{lag_subsp}). The span of these Hamiltonian vector fields
is called the characteristic distribution of $\calC$ and is an involutive distribution.

\begin{defn}[BFV resolution]
A BFV resolution of $(\calF^0,\omega^0,\calC)$ is a BFV manifold $(\calF,\omega,\calS)$  together
with an inclusion $\iota\colon\calF^0\hookrightarrow\calF$ 
such that
\begin{enumerate}
\item $\omega^0=\iota^*\omega$, and
\item the zero locus of the restriction of the Hamiltonian vector field $Q$ of $\calS$ to $\calF^0$ is $\calC$. 
\end{enumerate}
\end{defn}

One can always find a BFV resolution such
that $\calF^0$ is the ghost number zero component of $\calF$, $\omega$ has ghost number zero, and
$\calS$ and $Q$ have ghost number $+1$.

Finally, if the reduction $\underline{\calC}$ of $\calC$ (i.e., the leaf space of its characteristic distribution) is smooth then
$C^\infty(\underline{\calC})$ and the ghost-number-zero cohomology of $(C^\infty(\calF),Q)$ are isomorphic as Poisson algebras.

\begin{rem}
It very often happens that the (induced) BFV theory is also local. In this case, as observed in \cite{CMR1}, one may iterate the procedure of this section verbatim: namely, one considers the relaxed BFV theory associated to a manifold with boundary and induce an dg  symplectic manifold with symplectic structure of degree $1$ on its boundary. This procedure may be iterated (if the reductions are smooth and produce local structures). If one starts with a manifold with corner with a local BV structure, one then ends up having a local dg symplectic structure, with symplectic form of degree $-1+k$, on fields on the codimension $k$ corners (we call this a local BF$^k$V structure). 
\end{rem}

\section{Quantization on manifolds with boundary}
We want to construct an equivalent condition to the Quantum Master Equation for manifolds with boundary, i.e. we have to \emph{modify} the QME with an additional term taking care of the boundary structure on the quantum level.
For this purpose, quantization is performed in the guise of \emph{geometric quantization} (see e.g. \cite{Kir85,Wood97}) as it is described in \cite{CMR2}. 

\subsection{General idea}
Let us recall the functorial theory of topological (quantum) field theories. I.e. we consider the partition function $Z$ as a map $\mathbb{C}\to \mathbb{C}$ constructed with the Atiyah TQFT-axioms for closed manifolds. More generally, we consider a symmetric monoidal category $\Cob_n$, whose objects are $(n-1)$-dimensional manifolds (boundary components), and whose morphisms are diffeomorphism classes of bordisms, i.e. the objects are given as the components of the boundary and the morphism is, roughly speaking, given by the bulk of the manifold, and composition is given by gluing. Moreover, the objects of $\Cob_n$ are endowed with a certain symbol $\{in,out\}$ representing an orientation of the manifold (\emph{incoming} or \emph{outgoing}), and the mononidal structure on $\Cob_n$ is given by disjoint union. Note that the orientation is important for composition (gluing). In particular, we can only glue together boundaries with opposite orientation (see Figure \ref{gluing}).

\begin{defn}[Topological field theory]
\label{TFT_def}
Let $\mathcal{C}$ be a symmetric monoidal category. An $n$-dimensional \emph{topological field theory} (short \emph{TFT}) is a symmetric monoidal functor
\begin{equation}
\label{TFT}
\mathscr{F}\colon \Cob_n\to \mathcal{C},
\end{equation} 
\end{defn}

\begin{defn}[Topological quantum field theory]
A \emph{topological quantum field theory} (short \emph{TQFT}) is a functor as in Definition \ref{TFT_def}, where the target category is given by the category of vector spaces over the complex numbers, i.e. $\mathcal{C}=\Vect_\mathbb{C}$.
\end{defn}

A TQFT $\mathscr{F}$ will send the empty set $\varnothing$ to the trivial vector space $\mathbb{C}$. The partition function $Z$ is then defined on closed cobordisms, i.e. we have a diagram 
\[
\begin{tikzcd}
\Cob_n\ni\varnothing\arrow[d,swap,"\mathscr{F}"]\arrow[r,"{\includegraphics[scale=0.1]{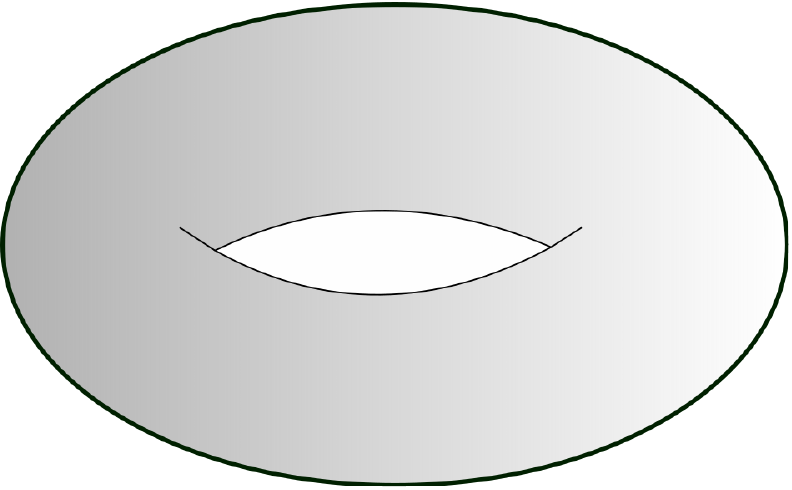}}"]&\arrow[d,"\mathscr{F}"]\varnothing\in \Cob_n\\
\Vect_\mathbb{C}\ni\mathbb{C}\arrow[r,"Z"]&\mathbb{C}\in \Vect_\mathbb{C}
\end{tikzcd}
\]
For cobordisms with nonempty boundary, a TQFT  $\mathscr{F}$ associates to each boundary component $\partial_k$ a Hilbert space $\mathcal{H}_k$. By the orientation on each boundary component of the cobordism, we assign an orientation for the state as an evolution mapping (e.g. see Figure \ref{pair_of_pants})

\begin{center}
\begin{figure}[h!]

\begingroup%
  \makeatletter%
  \providecommand\color[2][]{%
    \errmessage{(Inkscape) Color is used for the text in Inkscape, but the package 'color.sty' is not loaded}%
    \renewcommand\color[2][]{}%
  }%
  \providecommand\transparent[1]{%
    \errmessage{(Inkscape) Transparency is used (non-zero) for the text in Inkscape, but the package 'transparent.sty' is not loaded}%
    \renewcommand\transparent[1]{}%
  }%
  \providecommand\rotatebox[2]{#2}%
  \ifx\svgwidth\undefined%
    \setlength{\unitlength}{120.76020961bp}%
    \ifx\svgscale\undefined%
      \relax%
    \else%
      \setlength{\unitlength}{\unitlength * \real{\svgscale}}%
    \fi%
  \else%
    \setlength{\unitlength}{\svgwidth}%
  \fi%
  \global\let\svgwidth\undefined%
  \global\let\svgscale\undefined%
  \makeatother%
  \begin{picture}(1,0.69044631)%
    \put(0,0){\includegraphics[width=\unitlength]{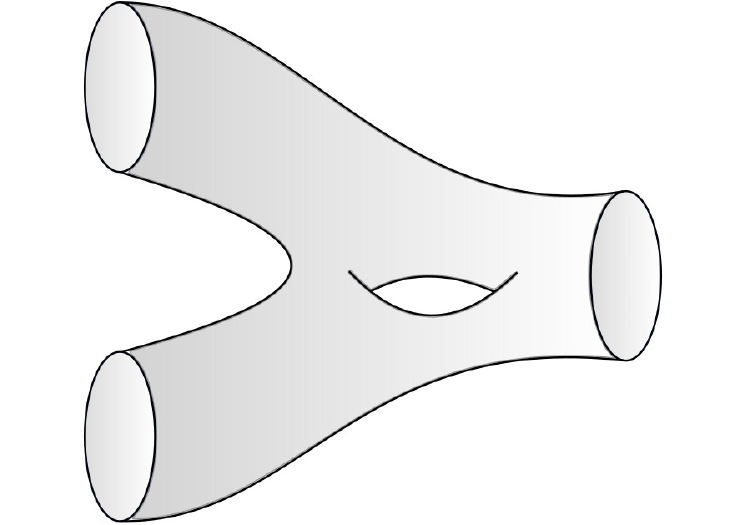}}%
    \put(-0.00071671,0.60247302){\color[rgb]{0,0,0}\makebox(0,0)[lb]{\smash{$\partial_1$}}}%
    \put(-0.00082684,0.11658998){\color[rgb]{0,0,0}\makebox(0,0)[lb]{\smash{$\partial_2$}}}%
    \put(0.94137699,0.32800496){\color[rgb]{0,0,0}\makebox(0,0)[lb]{\smash{$\partial_3$}}}%
  \end{picture}%
\endgroup%

\caption{An example where the cobordism is some pair of pants with genus $1$. A TQFT functor $\mathscr{F}$ assigns to each boundary component a Hilbert space, thus we get $\mathscr{F}(\partial_k)=\mathcal{H}_k$ for $k=1,2,3$ and since $\mathscr{F}$ is a symmetric monoidal functor, we get $\mathscr{F}(\partial_1\sqcup\partial_2\sqcup\partial_3)=\mathcal{H}^*_1\otimes\mathcal{H}^*_2\otimes\mathcal{H}_3=\Hom(\mathcal{H}_1\otimes \mathcal{H}_2,\mathcal{H}_3)$. Note that each cobordism comes with a certain orientation. Thus we set $\partial_1$ and $\partial_2$ to be incoming boundaries and $\partial_3$ to be an outgoing boundary. Hence we have an incoming Hilbert space $\mathcal{H}_{in}:=\mathcal{H}_1^*\otimes \mathcal{H}^*_2\cong \calH_1\otimes\calH_2$ (associated to $\partial_1\sqcup\partial_2$) and an outgoing $\mathcal{H}_{out}:=\mathcal{H}_3$ (associated to $\partial_3$).
The state $\psi$ corresponding to this cobordism and the given TQFT is then given as the value of the morphism represented by the genus $1$ pair of pants above (i.e. the bounding manifold) under $\mathscr{F}$.}
\label{pair_of_pants}
\end{figure}
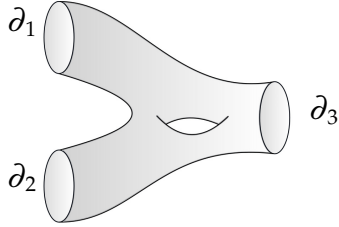
\end{center}
In particular, the state (evolution map) in Figure \ref{pair_of_pants}, is given by 
$$\mathscr{F}\left(\begin{gathered}{\includegraphics[scale=0.3]{genus1pants.eps}}\end{gathered}\right)=\psi\colon \mathcal{H}_{in}\to \mathcal{H}_{out}.$$

\begin{rem}
The above construction shows that the understanding of a perturbative quantization on manifolds with boundary is an important concept. We want to show that the QME, which is a gauge independence condition for partition functions on closed manifolds, can be extended to a more general condition on manifolds with boundary. 
\end{rem}

\begin{center}
\begin{figure}[h!]

\begingroup%
  \makeatletter%
  \providecommand\color[2][]{%
    \errmessage{(Inkscape) Color is used for the text in Inkscape, but the package 'color.sty' is not loaded}%
    \renewcommand\color[2][]{}%
  }%
  \providecommand\transparent[1]{%
    \errmessage{(Inkscape) Transparency is used (non-zero) for the text in Inkscape, but the package 'transparent.sty' is not loaded}%
    \renewcommand\transparent[1]{}%
  }%
  \providecommand\rotatebox[2]{#2}%
  \ifx\svgwidth\undefined%
    \setlength{\unitlength}{150.67789908bp}%
    \ifx\svgscale\undefined%
      \relax%
    \else%
      \setlength{\unitlength}{\unitlength * \real{\svgscale}}%
    \fi%
  \else%
    \setlength{\unitlength}{\svgwidth}%
  \fi%
  \global\let\svgwidth\undefined%
  \global\let\svgscale\undefined%
  \makeatother%
  \begin{picture}(1,0.63056987)%
    \put(0,0){\includegraphics[width=\unitlength]{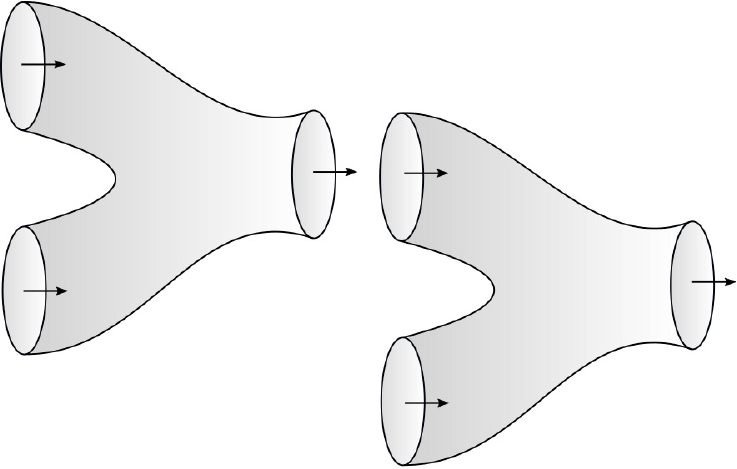}}%
    \put(0.47307113,0.47502189){\color[rgb]{0,0,0}\makebox(0,0)[lb]{\smash{$\Sigma$}}}%
    \put(0.78000127,0.2396118){\color[rgb]{0,0,0}\makebox(0,0)[lb]{\smash{$M_2$}}}%
    \put(0.22479298,0.40857472){\color[rgb]{0,0,0}\makebox(0,0)[lb]{\smash{$M_1$}}}%
  \end{picture}%
\endgroup%

\caption{Illustration for the gluing of two manifolds $M_1$ and $M_2$ along the common boundary $\Sigma$. The arrows at each boundary component represents the orientation $in$ or $out$. Note that the gluing has to be done by gluing $\Sigma$ on $M_1$ with $\Sigma$ on $M_2$ endowed with the opposite orientation.}
\label{gluing}
\end{figure}
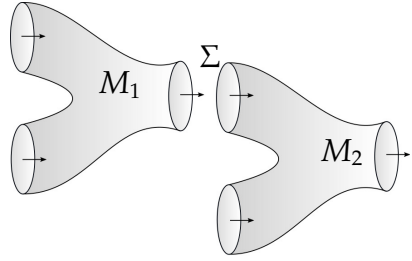
\end{center}

\subsection{The modified Quantum Master Equation}
Let us consider a compact $n$-manifold $M$ with boundary $\partial M$.  Let $\partial_{k_1}^{in}M$ and $\partial_{k_2}^{out}M$ denote the incoming and outgoing boundary components of $\partial M$ repsectively, where $k_1=1,...,n$ and $k_2=1,...,m$. In particular, we have $n$ incoming and $m$ outgoing boundary components such that 
\begin{align*}
\partial^{in}M&=\bigsqcup_{1\leq k_1\leq n}\partial_{k_1}^{in}M\\
\partial^{out}M&=\bigsqcup_{1\leq k_2\leq m}\partial_{k_2}^{out}M, 
\end{align*}
and thus $\partial M=\partial^{in}M\sqcup\partial^{out}M$. Considering a TQFT $\mathscr{F}$, we get $$\mathscr{F}(\partial M=\partial^{in}M\sqcup\partial^{out}M)=\mathcal{H}_{in}^*\otimes\mathcal{H}_{out}=\Hom(\mathcal{H}_{in},\mathcal{H}_{out})$$ and in particular, 
\begin{equation*}
\mathscr{F}\left(\bigsqcup_{k}\partial_k^{\bullet}M\right)=\bigotimes_{k}\mathscr{F}(\partial_k^\bullet M)=\bigotimes_k\mathcal{H}_{k,\bullet},
\end{equation*}
where $\bullet\in\{in,out\}$. An example of such a manifold $M$ is illustrated in Figure \ref{higher_surface}.

\begin{center}
\begin{figure}[h!]

\begingroup%
  \makeatletter%
  \providecommand\color[2][]{%
    \errmessage{(Inkscape) Color is used for the text in Inkscape, but the package 'color.sty' is not loaded}%
    \renewcommand\color[2][]{}%
  }%
  \providecommand\transparent[1]{%
    \errmessage{(Inkscape) Transparency is used (non-zero) for the text in Inkscape, but the package 'transparent.sty' is not loaded}%
    \renewcommand\transparent[1]{}%
  }%
  \providecommand\rotatebox[2]{#2}%
  \ifx\svgwidth\undefined%
    \setlength{\unitlength}{80.67332355bp}%
    \ifx\svgscale\undefined%
      \relax%
    \else%
      \setlength{\unitlength}{\unitlength * \real{\svgscale}}%
    \fi%
  \else%
    \setlength{\unitlength}{\svgwidth}%
  \fi%
  \global\let\svgwidth\undefined%
  \global\let\svgscale\undefined%
  \makeatother%
  \begin{picture}(1,2.11418763)%
    \put(0,0){\includegraphics[width=\unitlength]{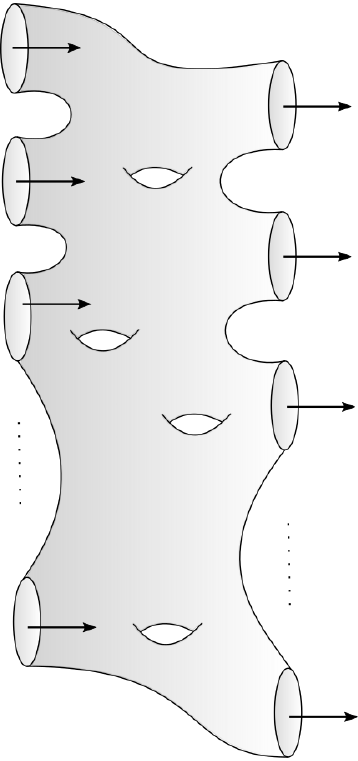}}%
    \put(0.40644048,0.66614779){\color[rgb]{0,0,0}\makebox(0,0)[lb]{\smash{$M$}}}%
  \end{picture}%
\endgroup%

\caption{Illustration of a $2$-manifold $M$ with $n$ incoming boundaries on the left and $m$ outgoing boundaries on the right. Its boundaries $\partial_k^\bullet M$ are the $1$-manifold objects in $\Cob_2$ and its bulk represents a particular morphism $\partial^{in}M\to\partial^{out}M$, namely the one with genus being $4$. Note again that different bulk structure (i.e. a different amount of genera) correspond to different morphisms in $\Cob_2$ and thus will lead to a different state $\psi$.}
\label{higher_surface}
\end{figure}
\end{center}

As already mentioned, the assignment of a Hilbert space to each boundary component can be done by techniques of geometric quantization. 
Thus, we need to fix a polarization $\mathcal{P}$ on $\calF^\partial_{\partial M}$ (boundary polarization). Moreover, we assume that the leaf space $\calB^\calP_{\partial M}$ (which depends on the polarization) is smooth and consider the Hilbert space $\calH^\calP_{\partial M}$ to be the functions on $\calB^\calP_{\partial M}$. 
\begin{rem}
Here we will assume that the space of fields has a linear structure, which is indeed the case for most theories (e.g. $BF$-like theories). However, in general this does not have to be the case.
\end{rem}
We assume a \emph{splitting} of the fibration $\calF_M\to \calB^\calP_{\partial M}$ as 
\begin{equation}
\label{split1}
\calF_M=\calB^\calP_{\partial M}\oplus \calY,
\end{equation}
such that $\omega_M$ is constant on the base $\calB_{\partial M}^\calP$. Basically, we have splitted the fields into a boundary part $\calB_{\partial M}^\calP$ and a bulk part $\calY$.

\vspace{0.2cm}

In order to define the quantum state, we need to introduce the concept of configuration spaces and their compactification and the formulation for manifolds with boundary.


\subsection{FMAS compactification}
\label{FMAS}
We start with the definition of the configuration space. 
\begin{defn}
Let $M$ be a manifold and $S$ a finite set. The \emph{open configuration space} of $S$ in $M$ is defined as
\begin{equation}
\mathsf{Conf}_S(M) := \{\iota\colon S \hookrightarrow M |\iota \hspace{0.2cm}\textnormal{injection}\}
\end{equation}
\end{defn}
Elements of $\mathsf{Conf}_{S}(M)$ are called $S$-configurations. To give an explicit definition of the compactification that can be extended to manifolds with boundaries and corners, we introduce the concept of \emph{collapsed configurations}. Intuitively, a collapsed $S$-configuration is the result of a collapse of a subset of the points in the $S$-configuration. However, we remember the relative configuration of the points before the collapse by directions in the tangent space. This is a configuration in the tangent space that is well-defined only up to translations and scaling. The difficulty is that one can imagine a limiting configuration where two points collapse first together and then with a third.
This explains the recursive nature of the following definition. Recall that if $X$ is a vector space, then $X\times \R_{>0}$ acts on $X$ by translations and scaling. 
\begin{defn}[Collapsed configuration in $M$]
Let $M$ be a manifold, $S$ a finite set and $\mathfrak{P} = \{S_1,\ldots,S_k\}$ be a partition of $S$. A \emph{$\mathfrak{P}$-collapsed configuration in $M$} is a $k$-tuple $(p_{\sigma},c_{\sigma})$ such that 
$((p_{\sigma},c_{\sigma}))_{\sigma = 1}^k$ satisfies 
\begin{enumerate}
\item $p_{\sigma} \in M$ and $p_{\sigma} \neq p_{\sigma'}$, for $\sigma \neq \sigma'$, 
\item $c_{\sigma} \in \Tilde{\mathsf{C}}_{S_{\sigma}}(T_{p_{\sigma}}M)$, where for $|S| = 1$, 
$\Tilde{\mathsf{C}}_S(X) := \{pt\}$  and for $|S| \geq 2$
\begin{equation}
\Tilde{\mathsf{C}}_S(X) := \coprod_{\substack{\mathfrak{P}=\{S_1,\ldots,S_k\} \\ S = \sqcup_\sigma S_\sigma, k\geq 2}}\left\lbrace \left(x_\sigma, c_\sigma \right)_{1\leq \sigma \leq k}\ \bigg|\ (x_\sigma, c_\sigma)\ \mathfrak{P}\text{-collapsed $S$-configuration in $X$}\right\rbrace \bigg/(X \times \R_{>0})
\end{equation}
\end{enumerate}

Here, $\varphi \in X \times \R_{>0}$ acts on $(x_\sigma,c_\sigma)$ by $(x_\sigma,c_\sigma) \mapsto (\varphi(x_\sigma), \dr\varphi_{x_\sigma}c_\sigma)$. 
\end{defn}
Intuitively, given a partition $\mathfrak{P}= \{S_1,\ldots,S_k\}$, a $k$-tuple $(p_\sigma,c_\sigma)$ describes the collapse of the points in $S_\sigma$ to $p_\sigma$. $c_\sigma$ remembers the relative configuration of the collapsing points. This relative configuration can itself be the result of a collapse of some points. 

\begin{defn}[FMAS compactification]
The \emph{compactified configuration space} $\mathsf{C}_S(M)$ of $S$ in $M$ is given by 
\begin{equation}
\mathsf{C}_S(M) := \coprod_{\substack{S_1,\ldots,S_k \\ S = \sqcup_\sigma S_\sigma}}\left\lbrace (p_\sigma, c_\sigma )_{1\leq \sigma \leq k}\ \bigg|\ (p_\sigma, c_\sigma)\ \mathfrak{P}\text{-collapsed $S$-configuration in $M$}\right\rbrace.
\end{equation}
%
\end{defn}

\subsection{Boundary strata}
A precise description of the combinatorics of the stratification can be found in \cite{FulMacPh}, where it is also shown that $\mathsf{C}_S(M)$ is a manifold with corners and is compact if $M$ is compact. For us, only strata in low codimensions are interesting. Let $S=\{s_1,\ldots,s_k\}$. 
The stratum of codimension 0 corresponds to the partition $\mathfrak{P}= \{\{s_1\},\ldots,\{s_k\}\}$. For $\ell>1$, strata of codimension 1 correspond to the collapse of exactly one subset $S'=\{s_1,\ldots,s_\ell\} \subset S$ with no further collapses, i.e a partition $\mathfrak{P}=\{\{s_1,\ldots,s_\ell\},\{s_{\ell+1}\},\ldots,\{s_k\}\}$ and configuration $(p_\sigma,c_\sigma)$ with $c_\sigma$ in the component of $\Tilde{\mathsf{C}}_{S'}(X)$ given by the partition $\mathfrak{P} =  \{\{s_1\},\ldots,\{s_\ell\}\}$. This boundary stratum will be denoted by $\de_{S'}\mathsf{C}_S(M)$, in particular, we have
\begin{equation}
\de\mathsf{C}_S(M) = \coprod_{S' \subset S}\de_{S'}\mathsf{C}_S(M).
\end{equation} 
There is a natural fibration $\de_{S'}\mathsf{C}_S(M) \to \mathsf{C}_{S \setminus S'\cup \{pt\}}(M)$ whose fiber is $\Tilde{\mathsf{C}}_S(\R^{\dim M})$. Finally, we note that if $|S| = 2$, then $\mathsf{C}_{S}(M) \cong Bl_{\overline{\Delta}}(M \times M)$, the \emph{differential-geometric blow-up} of the diagonal $\overline{\Delta} \subset M \times M$, and $\Tilde{\mathsf{C}}_S(X) \cong S^{\dim X-1}$.

\subsection{Configuration spaces for manifolds with boundary}
We proceed to recall the definition of a compactified configuration space for manifolds with boundary. Let $M$ be a compact manifold with boundary $\de M$. Recall that for a manifold $M$ with boundary $\de M$, at points $p \in \de M$ there is a well-defined notion of inward and outward half-space in $T_pM$. If $H \subset X$ is a half-space, then $\de H \subset X$ is a hyperplane. $\de H \times \R_{>0}$ acts on $H$ by translations and scaling. 
%
%
\begin{defn}[Configuration spaces for manifolds with boundary]
Let $M$ be a manifold with boundary $\de M$. For $S,T$ finite sets, we define the \emph{open configuration space} by
\begin{equation}
\mathsf{Conf}_{S,T}(M,\de M) := \{(\iota,\iota')\colon S \times T \hookrightarrow M \times \de M\}
\end{equation}
\end{defn}
\begin{defn}[Collapsed configuration on manifolds with boundary]
Let $(M, \de M)$ be a manifold with boundary. Let $S, T$ be finite sets and $\mathfrak{P}=\{S_1, \ldots, S_k\}$ a partition of $S \sqcup T$.  Then, a $\mathfrak{P}$-collapsed $(S,T)$-configuration in $M$ is a $k$-tuple of pairs $(p_\sigma,c_\sigma)$ such that
\begin{enumerate}
\item $p_\sigma \in M$ and $p_\sigma\neq p_{\sigma'}$, for all $\sigma\neq \sigma'$, 
\item $S_\sigma \cap T \neq \varnothing \Rightarrow p_\sigma \in \de M$,
\item $$c_\sigma \in \begin{cases} 
\Tilde{\mathsf{C}}_{S_\sigma}(T_{p_\sigma}M)  & p_\sigma \in M \setminus \de M \\  \Tilde{\mathsf{C}}_{S \cap S_\sigma, T \cap S_\sigma}(\mathbb{H}(T_{p_\sigma}M)) & p_\sigma \in \de M\end{cases}$$
\end{enumerate}
where $\mathbb{H}(T_{p_\sigma}M) \subset T_{p_\sigma}M$ denotes the inward half-space in $T_{p_\sigma}M$. Here, for a vector space $X$ and a half-space $H \subset X$, $\Tilde{\mathsf{C}}_{\varnothing,\{pt\}}(H) := \Tilde{\mathsf{C}}_{\{pt\},\varnothing}(H) := \{pt\}$, and for $\vert S \sqcup T\vert \geq 2$, 
$$\Tilde{\mathsf{C}}_{S,T}(H) := \coprod_{\substack{\mathfrak{P}=\{S_1,\ldots,S_k\}\\ S\sqcup T = \sqcup_\sigma S_\sigma, k \geq 2 }}\left\lbrace (v_\sigma,c_\sigma)\ \bigg|\ (v_\sigma, c_\sigma)\ \mathfrak{P}\text{-collapsed $(S,T)$-configuration in $H$} \right\rbrace \bigg/ (\de H \times \R_{>0}) $$

\end{defn}
\begin{defn}[FMAS compactification for manifolds with boundary]
We define the  \emph{compactification} $\mathsf{C}_{S,T}(M,\de M)$ of $\mathsf{Conf}_{S,T}(M,\de M)$ by 
\begin{equation}
\mathsf{C}_{S,T}(M,\de M) = 
\coprod_{\substack{ \mathfrak{P} =\{S_1,\ldots,S_k\} \\ S\sqcup T= \sqcup_\sigma S_\sigma}}\left\lbrace \left(p_\sigma , c_\sigma\right)_{1\leq \sigma \leq k}\ \bigg|\ (p_\sigma,c_\sigma)\  \mathfrak{P}\text{-collapsed $(S,T)$-configuration}\right\rbrace
\end{equation}

\end{defn}
 Again, this is a manifold with corners and is compact if $M$ is compact. We proceed to describe the strata of low codimension. Let $U= \{u_1,\ldots,u_k\}, V = \{v_1,\ldots,v_k\}.$ The codimension 0 stratum again is given by the partition $\mathfrak{P} = \{\{u_1\},\ldots,\{u_k\},\{v_1\},\ldots, \{v_\ell\}\}.$ Let us describe the strata of codimension 1. We denote by $\de^\textnormal{I}_S\mathsf{C}_{U,V}(M,\de M)$ a boundary stratum where a subset $S \subset U$ collapses in the bulk, described in the same way as above.  On manifolds with boundary, there are new boundary strata in the compactified configuration space given by the collapse of a subset of points to a point in the boundary. Concretely, given a subset $S=\{u_1,\ldots,u_{k'},v_1,\ldots,v_{\ell'}\} \subset U \sqcup V$, there is a boundary stratum $\de^{\textnormal{II}}_S\mathsf{C}_{U,V}(M,\de M)$ corresponding to the partition $\mathfrak{P}=\{S,\{u_{k'+1}\},\ldots, \{u_k\},\{v_{\ell'+1}\},\ldots,\{v_\ell\}\}$ and collapsed configurations $(p_\sigma,c_\sigma)$ with $p_{\sigma} \in \de M$ and $c_\sigma$ corresponding to the partition $\mathfrak{P'} = \{\{u_1\},\ldots,\{u_k\},\{v_1\},\ldots,\{v_\ell\}\}$. The boundary decomposes as 
\begin{equation}
\de \mathsf{C}_{U,V}(M,\de M) = \coprod_{S \subseteq U} \de^\textnormal{I}_S\mathsf{C}_{U,V}(M,\de M) \amalg\coprod_{S \subseteq U \sqcup V} \de^{\textnormal{II}}_S\mathsf{C}_{U,V}(M,\de M).
\end{equation}

All the computations are done on the level of Feynman graphs, so we want to give a definition of what a graph means for us:

\begin{defn}[Graph]
An (oriented) graph $\Gamma$ is an (ordered) pair $(V(\Gamma),E(\Gamma))$, where $V(\Gamma)$ is a finite set and $E(\Gamma)$ is a multiset over the direct product $V(\Gamma)\times V(\Gamma)$, i.e. a map $E(\Gamma)\colon V(\Gamma)\times V(\Gamma)\to \N$ (respectively a multiset over any subset $W$ of $V(\Gamma)$ with two elements, i.e. a map $E(\Gamma)\colon W\to \N$). The elements of $E(\Gamma)$ are called the \emph{edges}- and the elements of $V(\Gamma)$ are called the \emph{vertices} of $\Gamma$. We denote by $\mathcal{G}_{n}$ the set of all graphs with $n$ vertices.
\end{defn}

\begin{rem}
Note that we also allow \emph{multiple edges}, i.e. we can have the same edge $e=(v_1,v_2)$ twice in $E(\Gamma)$, with $v_1,v_2\in V(\Gamma)$. This definition would also allow \emph{tadpoles (short loops)}, i.e. edges $e=(v,v)\in E(\Gamma)$ connecting the vertex $v\in V(\Gamma)$. What this definition excludes, are graphs with \emph{legs}. Examples of these graphs are given in Figure \ref{ex_graphs}.
\end{rem}

\begin{figure}[h!]
\begingroup%
  \makeatletter%
  \providecommand\color[2][]{%
    \errmessage{(Inkscape) Color is used for the text in Inkscape, but the package 'color.sty' is not loaded}%
    \renewcommand\color[2][]{}%
  }%
  \providecommand\transparent[1]{%
    \errmessage{(Inkscape) Transparency is used (non-zero) for the text in Inkscape, but the package 'transparent.sty' is not loaded}%
    \renewcommand\transparent[1]{}%
  }%
  \providecommand\rotatebox[2]{#2}%
  \ifx\svgwidth\undefined%
    \setlength{\unitlength}{250.58172156bp}%
    \ifx\svgscale\undefined%
      \relax%
    \else%
      \setlength{\unitlength}{\unitlength * \real{\svgscale}}%
    \fi%
  \else%
    \setlength{\unitlength}{\svgwidth}%
  \fi%
  \global\let\svgwidth\undefined%
  \global\let\svgscale\undefined%
  \makeatother%
  \begin{picture}(1,0.39540534)%
    \put(0,0){\includegraphics[width=\unitlength]{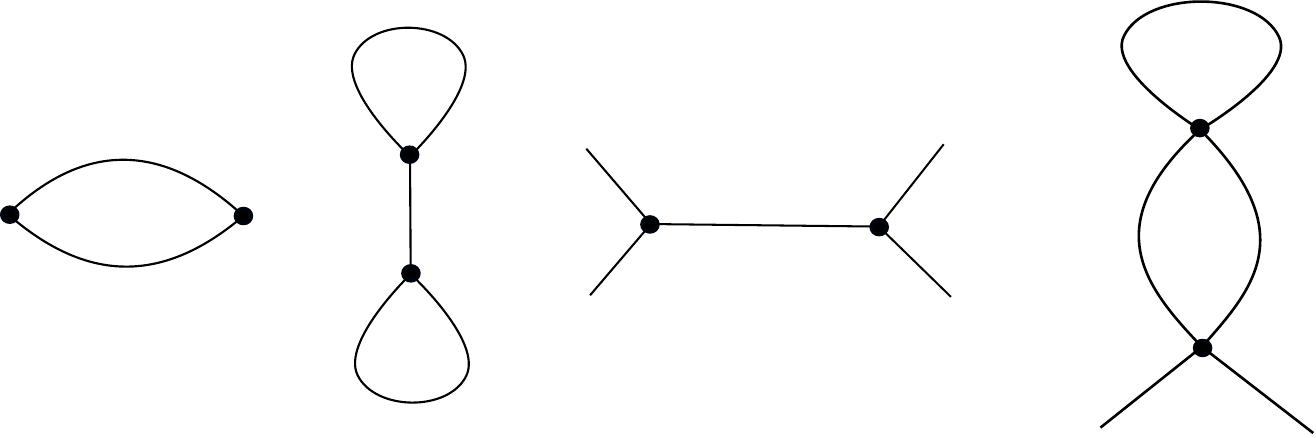}}%
  \end{picture}%
\endgroup%
\caption{Examples of different graphs. The left one is an example of a graph with multiple edges, but with no tadpoles or legs, the second one from left is an example of a graph with tadpoles, but no multiple edges and legs, the second one from right is a graph with legs, but no multiple edges and tadpoles, and the right one is a graph with multiple edges, tadpoles and legs. Note that here we do not consider any orientation.}
\label{ex_graphs}
\end{figure}


\begin{rem}
Note that in application to perturbative QFT, we are interested in \emph{Feynman graphs}, which are graphs where each vertex $v$ represents the interaction term of the given action and each edge between two such interactions represents a \emph{propagator} (a.k.a Green's function, a.k.a integral kernel) which captures the free-theory evolution of one state into another. 
We want the vertex set of a graph to be a discrete subset of a manifold (possibly with boundary). Since the propagator is singular on the diagonal we need to work with configuration spaces for Feynman graphs. 
We denote the set of graphs with $n$ vertices in the bulk and $m$ vertices on the boundary by $\calG_{n,m}$.
\end{rem}

\begin{notation}
For a manifold $M$ without boundary, we will denote the compactified configuration space of $n$ points  $\mathsf{C}_{[n]}(M)$ on $M$ by $\mathsf{C}_n(M)$ (here $[n] = \{1,\ldots, n\}$). Moreover, for a manifold $M$ with boundary, we denote the compactified configuration space  $\mathsf{C}_{n,m}(M)$ of $n$ points on the bulk of $M$ and $m$ points on the boundary $\de M$ of $M$ by $\mathsf{C}_{n,m}(M,\de M)$. We will also write $\mathsf{C}_\Gamma(M)$ for $\mathsf{C}_{n,m}(M,\de M)$, if $\Gamma$ is a graph with $n+m$ vertices, $n$ vertices in the bulk of $M$ and $m$ vertices on $\de M$. 
\end{notation}

Perturbatively, we can define the \emph{quantum state} to be given as 
\begin{equation}
\label{quantum_state}
\psi_M=\int_{\calL\subset \calY}\ee^{\frac{\I}{\hbar}\calS_M}\in\calH^\calP_{\partial M},
\end{equation}
where $\calL\subset\calY$ is a Lagrangian submanifold of $\calY$. This means that we define \eqref{quantum_state} as the formal power series expansion in $\hbar$ with coefficients given in terms of Feynman diagrams, i.e. we can write $\psi_M$ asymptotically as 
\begin{equation}
\label{expansion}
\psi_M\approx \sum_{k\geq 0}\sum_{\Gamma\in\mathcal{G}\atop \text{Feynman graph}}W(\Gamma)\hbar^k=\dotsm +W\left(\begin{gathered}\includegraphics[scale=0.9]{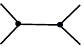}
\end{gathered}\right)\hbar^{k_1}+\dotsm +W\left(\begin{gathered}\includegraphics[scale=0.5]{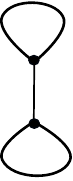} \end{gathered}\right)\hbar^{k_2}+\dotsm 
\end{equation}

where $W(\Gamma)\in \R$ denotes the \emph{weight} of the Feynman diagram $\Gamma$, which depends on the fields in $\calF_M$ (of both factors in the split), and $k_1,k_2\geq 0$ with $k_1<k_2$. Here we write $\approx$ for equality up to some constants.
In the \emph{effective low energy theory}, we can express \eqref{quantum_state} to be 
\begin{equation}
\label{effective_state}
\psi_M\approx \exp\left(\frac{\I}{\hbar}\sum_{\Gamma\in\mathcal{G}_{n,m}\atop \Gamma \text{ connected}}\hbar^{\ell(\Gamma)}\int_{\mathsf{C}_\Gamma(M)}\omega_\Gamma\right),
\end{equation}
where $\ell(\Gamma)$ denotes the number of loops of $\Gamma$, $\vert E(\Gamma)\vert$ denotes the amount of edges of $\Gamma$.
Moreover, $\omega_\Gamma$ is a differential form\footnote{This is for the case of a TQFT à la Schwarz. In general, the $\omega_\Gamma$s are distributions and their products and integration must be regularized (and renormalized).} on $\mathsf{C}_\Gamma(M)$, given as a polynomial in both field parts of the splitting \eqref{split1} (actually, only of a certain part of $\calY$ as we will see). The integration then corresponds to the weight as in \eqref{expansion}, i.e. 
\[
W(\Gamma)=\int_{\mathsf{C}_\Gamma(M)}\omega_\Gamma.
\]
We left away the usual constant coefficients in the formulae for simplicity.

We denote by $\calS^{\text{eff}}_M$ the \emph{low energy effective action}, which is given by the logarithm
\begin{equation}
\label{effective_action}
\calS^{\text{eff}}_M:\approx\sum_{\Gamma\in\mathcal{G}_{n,m}\atop \Gamma \text{ connected}}\hbar^{\ell(\Gamma)}\int_{\mathsf{C}_\Gamma(M)}\omega_\Gamma.
\end{equation}

\begin{rem}
Recall that the effective action $\calS^{\text{eff}}_M$, is given by \eqref{effective_action}, since it is defined by considering only connected Feynman diagrams in the perturbative expansion. 
\end{rem}

\subsubsection{Heuristic proof of the modified Quantum Master Equation}
We want to give a heuristic proof sketch of the Quantum Master Equation on manifolds with boundary.
Using again the splitting $\calF=\calB\oplus \calY$ of the space of fields, we can rewrite the mCME as  
\begin{align}
\delta_\calY\calS&=\iota_{Q_\calY}\omega,\\
\delta_{\calB}\calS&=-\alpha^\de.
\end{align}
where we have dropped the $\pi^*$. The two equations imply 
\begin{equation}
\label{eq01}
\frac{1}{2}(\calS,\calS)_\calY=\frac{1}{2}\iota_{Q_\calY}\iota_{Q_\calY}\omega=\calS^\de.
\end{equation}
Now assume that we have adapted Darboux coordinates $(b,p)$ on $\calF^\de$ with $b$ on $\calB$, $p$ on the leaves and $\alpha^\de=-\sum_ip_i\delta b^{i}$. Then the second equation implies 
\begin{equation}
\label{eq03}
\frac{\delta\calS}{\delta b}=p.
\end{equation}
This means that, in this splitting, $\calS$ is linear in the $b$'s. We now assume that $\calS$ also solves the equation $\Delta_\calY\calS=0$. 
\begin{rem}
Without boundary this means that we assume that $\calS$ solves both the Classical and the Quantum Master Equation. With boundary, $\Delta$ makes sense only on the $\calY$-factor. 
\end{rem}
Then we get 
$$\Delta_\calY\ee^{\frac{\I}{\hbar}\calS}=\left(\frac{\I}{\hbar}\right)^2\frac{1}{2}(\calS,\calS)_\calY\ee^{\frac{\I}{\hbar}\calS},$$
and equation \eqref{eq01} implies 
\begin{equation}
\label{eq02}
-\hbar^2\Delta_\calY\ee^{\frac{\I}{\hbar}\calS}=\calS^\de \ee^{\frac{\I}{\hbar}\calS}.
\end{equation}
Now we move to quantization. Take $\calH$ to be an appropriate space of functions on $\calB$. Equation \eqref{eq03} essentially says that 
$$\Hat{p}\calS=-\I\hbar p,\quad \text{with}\quad \Hat{p}=-\I\hbar\frac{\delta}{\delta b}.$$
\begin{rem}
Here $\calS$ is an element of $\calH$ parametrized by $\calY$. The $p$ appearing in the equation is now an element of $\calY$.
\end{rem}

If we quantize $\calS^\de$ by the Schr\"odinger prescription, we get
$$\Omega:=\calS^\de\left(b,-\I\hbar\frac{\delta}{\delta b}\right)$$
with all derivatives placed to the right, so 
\begin{equation}
\label{eq04}
\Omega\ee^{\frac{\I}{\hbar}\calS}=\calS^\de \ee^{\frac{\I}{\hbar}\calS}.
\end{equation}
Putting \eqref{eq02} and \eqref{eq04} we finally get the \emph{modifed Quantum Master Equation (mQME)}
\begin{equation}
(\hbar^2\Delta_\calY+\Omega)\ee^{\frac{\I}{\hbar}\calS}=0.
\end{equation}

\begin{rem}
The assumption $\Delta_\calY\calS=0$ is not really necessary (and is often not justified). More generally, we have 
$$\Delta_\calY\ee^{\frac{\I}{\hbar}\calS}=\left(\left(\frac{\I}{\hbar}\right)\Delta_\calY\calS+\left(\frac{\I}{\hbar}\right)^2\frac{1}{2}(\calS,\calS)_\calY\right)\ee^{\frac{\I}{\hbar}\calS}.$$
If we define 
$$\calS^\de_\hbar:=\frac{1}{2}(\calS,\calS)_\calY-\I\hbar\Delta_\calY\calS=\calS^\de+O(\hbar),$$
and $\Omega$ to be the Schr\"odinger quantization (canonical quantization) of $\calS^\de_\hbar$, we recover the mQME.
By construction we have $\Delta_\calY^2=0$ and $[\Delta_\calY,\Omega]=0$. The operator $\Omega_\calY:=\hbar^2\Delta_\calY+\Omega$ appearing in the mQME then squares to zero if and only if $\Omega^2=0$. The existence of a splitting such that this holds is a fundamental condition (absence of anomalies) which allows passing to the $\Omega_\calY$-cohomology. Cohomology in degree zero describes $\calY$-parametrized physical states.
\end{rem}

\subsubsection{$BF$-like Theories}
Let us from now on consider the case of $BF$-like theories as in the following definition.

\begin{defn}[$BF$-like theories]
\label{BF_like}
We say that a BV-BFV theory is \emph{$BF$-like} if
\begin{align*}
\calF_M &= (\Omega(M)\otimes V[1])\oplus (\Omega(M)\otimes V^*[(\dim M)-2])\ni (\mathsf{A},\mathsf{B})\\
\calS_M &= \int_M\left(\langle\mathsf{B},\dd\mathsf{A}\rangle+\calV(\mathsf{A},\mathsf{B})\right),
\end{align*}
where $V$ is a graded vector space, $\langle\enspace,\enspace\rangle$ denotes the pairing between $V^*$ and $V$, and $\calV$ denotes some density-valued function of the fields $\mathsf{A}$ and $\mathsf{B}$,
such that $\calS_M$ satisfies the Classical Master Equation for $M$ without boundary.
\end{defn}

\begin{rem}
Examples of $BF$-like theories are abelian $BF$ theory, the Poisson Sigma Model, Chern--Simons theory, or the B-model. 
\end{rem}

\begin{rem}
For $BF$-like theories we split the boundary $\de M$ into the disjoint union of two boundary components, i.e. $\de M=\de_1M\sqcup \de_2M$, where we set the $\frac{\delta}{\delta \mathsf{B}}$-polarization on $\de_1M$ and the $\frac{\delta}{\delta \mathsf{A}}$-polarization on $\de_2M$. Moreover denote the split as in \eqref{split1} of a field by 
$$(\mathsf{A},\mathsf{B})=(\mathbb{A},\mathbb{B})\oplus(\Hat{\mathsf{A}},\Hat{\mathsf{B}}).$$
\end{rem}

\subsubsection{Boundary terms}
\label{boundary_terms}
A splitting as in \eqref{split1} leads to a fiberwise version of the mCME (equation \eqref{mCME}). Thus, the exponential of the action is only $\Delta$-closed up to some boundary terms that can be summarized as the action of a differential operator $\Omega_{\partial M}^\calP$ on $\calB_{\partial M}^\calP$ that \emph{quantizes} the boundary action $\calS^\partial_{\partial M}$. We call $\Omega_{\de M}^\calP$ the \emph{BFV boundary operator}.
In an ideal situation (e.g. for abelian $BF$ theory (\cite{CMR2})) one can consider it as the \emph{standard quantization} of the boundary action. Namely, for a boundary field $b\in\calB_{\partial M}^\calP$, we can consider its canonical quantization $-\I\hbar\frac{\delta}{\delta b}$, and hence we get 
\begin{equation}
\label{BFV_op}
\Omega^\calP_{\partial M}:=\calS^\partial_{\partial M}\left(b,-\I\hbar\frac{\delta}{\delta b}\right).
\end{equation}
More concrete, 
$$\Omega^\calP_{\de M}=(-1)^{\dim M}\I\hbar\int_{\de M}\left(\dd\mathbb{A}\frac{\delta}{\delta \mathbb{A}}+\dd\mathbb{B}\frac{\delta}{\delta\mathbb{B}}\right).$$

\begin{rem}
\label{boundary_integrals}
In general this is not true. In particular, $\Omega^\calP_{\de M}$ is given by Feynman graphs collapsing on the boundary of the cmpactified configuration space as defined in \ref{FMAS} (see Figure \ref{collapsing} for an illustration). 
Let $\Gamma$ be a Feynman graph and $\omega_\Gamma$ the corresponding differential form over the compactified configuration space $\mathsf{C}_\Gamma(M)$. Consider Stokes' theorem $\int_{\mathsf{C}_\Gamma(M)}\dd \omega_\Gamma=\int_{\de \mathsf{C}_\Gamma(M)}\omega_\Gamma$. The left hand side contains terms where $\dd$ acts on $\mathbb{A}$ and $\mathbb{B}$ and terms where $\dd$ acts on the propagator. The former corresponds to the action of $\frac{1}{\I\hbar}\Omega^\calP_0$, where $\Omega^\calP_0$ denotes the standard quantization, the latter when summed over graphs $\Gamma$ assemble to the action of $-\I\hbar\Delta_{\calV_M^\calP}$ on the state. The right hand side contains three classes of terms:
\begin{itemize}
\item{Integrals over boundary components where two vertices collapse in the bulk. The combination of the Feynman diagrams in the expansion ensures that these terms cancel out when we sum over all the diagrams.\footnote{This cancellation relies on the assumption that the perturbed action satisfies the CME, which is equivalent to $\sum_i \pm \frac{\delta}{\delta \mathsf{A}_i}\calV(\mathsf{A},\mathsf{B})\cdot \frac{\delta}{\delta \mathsf{B}_i}\calV(\mathsf{A},\mathsf{B})=0$, which in turn implies a relation on contractions of pairs of vertex tensors.}
}
\item{Integrals over boundary components where more than two vertices collapse in the bulk (``hidden faces''). The usual arguments—vanishing theorems—ensure the vanishing of all these terms apart, possibly, for faces where all the vertices of a connected component of a graph collapse. In all the above mentioned theories, with the exception of Chern–Simons theory, also these terms vanish. In Chern–Simons theory, they may possibly survive, but can be compensated by a framing dependent term.
}
\item{Terms where two or more (bulk and/or boundary) vertices collapse together at the boundary or a single bulk vertex hits the boundary. The integral on such a boundary face splits into an integral over a subgraph $\Gamma'$ of $\Gamma$ corresponding to the collapsed vertices and an integral over $\Gamma/\Gamma'$, the graph obtained by identifying all the vertices in $\Gamma'$ and forgetting the edges inside $\Gamma'$. We define the action of $\frac{\I}{\hbar}\Omega_{\textnormal{pert}}^\calP$, which is the additional term appearing besides $\Omega_0^\calP$, by the sum of the boundary contributions of the $\Gamma'$s. If we now sum ovr all graphs $\Gamma$, all these terms will give $\frac{\I}{\hbar}\Omega_{\textnormal{pert}}^\calP$ applied to the state.
}
\end{itemize}
Hence the BFV boundary operator is given by $\Omega_{\de M}^\calP=\Omega_0^\calP+\Omega_{\textnormal{pert}}^\calP$.
\end{rem}

As a consequence of Remark \ref{boundary_integrals} we get 
\begin{equation}
\label{mQME1}
(\hbar^2\Delta_{\calV_M^\calP}+\Omega_{\partial M}^\calP)\ee^{\frac{\I}{\hbar}\calS_M}=0.
\end{equation}
Moreover, we assume that $\Omega_{\partial M}^\calP$ squares to zero, which corresponds to the fact that the theory does not have any anomalies. 

\begin{center}
\begin{figure}[h!]

\begingroup%
  \makeatletter%
  \providecommand\color[2][]{%
    \errmessage{(Inkscape) Color is used for the text in Inkscape, but the package 'color.sty' is not loaded}%
    \renewcommand\color[2][]{}%
  }%
  \providecommand\transparent[1]{%
    \errmessage{(Inkscape) Transparency is used (non-zero) for the text in Inkscape, but the package 'transparent.sty' is not loaded}%
    \renewcommand\transparent[1]{}%
  }%
  \providecommand\rotatebox[2]{#2}%
  \ifx\svgwidth\undefined%
    \setlength{\unitlength}{200.07926837bp}%
    \ifx\svgscale\undefined%
      \relax%
    \else%
      \setlength{\unitlength}{\unitlength * \real{\svgscale}}%
    \fi%
  \else%
    \setlength{\unitlength}{\svgwidth}%
  \fi%
  \global\let\svgwidth\undefined%
  \global\let\svgscale\undefined%
  \makeatother%
  \begin{picture}(1,0.69254713)%
    \put(0,0){\includegraphics[width=\unitlength]{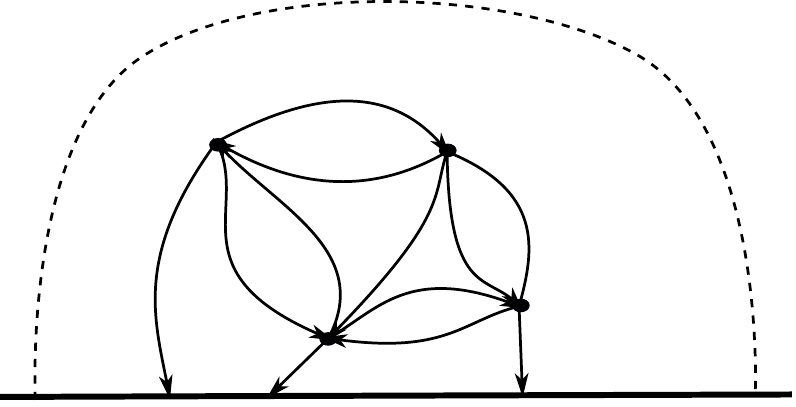}}%
    \put(0.79566662,0.0582663){\color[rgb]{0,0,0}\makebox(0,0)[lb]{\smash{$\partial M$}}}%
  \end{picture}%
\endgroup%

\caption{An example of a graph in $\mathcal{G}_{4,3}$ in the bulk collapsing to the boundary of $M$, which contributes to a term in $\Omega^\calP_{\partial M}$. The dashed semi-circle represents the collapsing of the graph.}
\label{collapsing}
\end{figure}
\end{center}

\begin{prop} 
\label{closed}
we have 
\begin{equation}
\Omega^\calP_{\partial M}\psi_M=0, 
\end{equation}
\end{prop}
\begin{proof}[Formal sketch of the proof]
One can look at this formally using \eqref{expansion}:
$$\Omega^\calP_{\partial M}\psi_M=\Omega_{\partial M}^\calP\int_{\calL}\ee^{\frac{\I}{\hbar}\calS_M}=\int_{\calL}\Omega_{\partial M}^\calP\ee^{\frac{\I}{\hbar}\calS_M}=-\hbar^2\int_\calL \Delta\ee^{\frac{\I}{\hbar}\calS_M}=0,$$
where we have used \eqref{mQME1} and \eqref{QME} (i.e. $\Delta \ee^{\frac{\I}{\hbar}\calS_M}=0$).
\end{proof}

\begin{rem}
One can also check that gauge fixing (i.e. deforming the Lagrangian submanifold) changes the state $\psi_M$ by an $\Omega^\calP_{\partial M}$-exact term. Moreover, Proposition \ref{closed} tells us that $\psi_M$ defines a class in the space $H^0_{\Omega^\calP_{\partial M}}(\calH^\calP_{\partial M})$.
\end{rem}

\subsubsection{Residual fields}
The perturbative definition of the path integral is defined by perturbing around a quadratic theory, i.e. by classical solutions of the quadratic part of the action $\calS_M$. We denote by $\calV_M^\calP$ the space of critical points of the quadratic part relative to the boundary polarization $\calP$ modulo symmetries. We call $\calV_M^\calP$ the space of \emph{residual fields} (also called \emph{low energy fields}, \emph{zero modes}, etc.).

\begin{rem}
The space $\calV_M^\calP$ is a finite-dimensional space, which, for $BF$-like theories, can be computed by relative cohomology, i.e.
$$\calV_M^\calP=\bigoplus_{k\geq 0}(H^k(M,\de_1M)\otimes V[1])\oplus(H^k(M,\de_2M)\otimes V^*[(\dim M)-2]).$$
Moreover, note that $\calV_M^\calP$ depends only on the bulk fields.
\end{rem}

By the structure of a cobordism, we expect the state $\psi_M$ not only to depend on the boundary fields but also on the residual fields on the bulk, which is encoded in a certain cohomology structure, since the bulk is different for different amount of genera. This was also described in the caption of Figure \ref{higher_surface}. We assume another symplectic splitting of the bulk part (see equation \eqref{split1}) by
\begin{equation}
\label{split2}
\calY=\calV_M^\calP\oplus \calY'.
\end{equation}
Thus, the state can be defined as a BV pushforward
\begin{equation}
\label{state_final}
\psi_M=\int_{\calL\subset \calY'}\ee^{\frac{\I}{\hbar}\calS_M}\in\calH^\calP_{\partial M}\otimes \mathcal{Z}_M^\calP,
\end{equation}
where $\mathcal{Z}_M^\calP$ is the space of functions on $\calV_M^\calP$ and $\calL\subset \calY'$ is a Lagrangian submanifold of $\calY'$.
\begin{rem}
In fact, $\mathcal{Z}^\calP_M$ is given by \emph{formal half-densities} on $\calV^\calP_M$, which we denote by $\Dens^\frac{1}{2}(\calV^\calP_M)$. Thus the \emph{space of states} is given by $\calH^\calP_M:=\calH_{\partial M}^\calP\otimes \Dens^\frac{1}{2}(\calV^\calP_M)$.
\end{rem}
By Proposition \ref{fiber_int} and Remark \ref{cohom_class}, we can consider the part of the BV Laplacian, denoted by $\Delta_{\calV^\calP_M}$, which acts on the space $\mathcal{Z}^\calP_M$. Then if $\psi_M$ statisfies the QME, we see that there is a well-defined $\Delta_{\calV^\calP_M}$-cohomology of $\psi_M$, and hence gives us elements in $\calV^\calP_M$. Finally, we can summarize the condition for gauge independence by the \emph{modified Quantum Master Equation (mQME)}
\begin{equation}
\label{mQME}
(\hbar^2\Delta_{\calV^\calP_M}+\Omega_{\partial M}^\calP)\psi_M=0.
\end{equation}

\subsection{Gluing}
Assume we have two $d$-manifolds $M_1$ and $M_2$, which we want to glue together by a common boundary component $\Sigma$ to a manifold $M$, i.e. 
$$M=M_1\cup_\Sigma M_2.$$
The glued state $\psi_M$ can then be obtained from the states $\psi_{M_1}$ and $\psi_{M_2}$ by a BV pushforward 
\begin{equation}
\label{Gluing_BV}
\mathscr{P}_{BV}\colon (\psi_{M_1},\psi_{M_2})\mapsto \psi_M=\mathscr{P}_{BV}(\langle\psi_{M_1},\psi_{M_2}\rangle_{\calH_{\Sigma}^\calP}),
\end{equation}
where $\mathscr{P}_{BV}$ is the BV-pushforward with respect to the odd-symplectic fibration of residual fields $\mathscr{P}\colon \calV^\calP_{M_1}\oplus \calV_{M_2}^\calP\to \calV^\calP_M$ and where $\langle\enspace,\enspace\rangle_{\calH_{\Sigma}^\calP}$ denotes the pairing on $\calH_{\Sigma}^\calP$. This is illustrated in Figure \ref{gluing}. Hence we can describe the axiomatics of perturbtaive quantum gauge theories on manifolds with boundary by the following definiton:
\begin{defn}[Axiomatics]
The \emph{axiomatics} of a perturbative quantum gauge theory on a manifold $M$ with boundary, is given by the following data:
\begin{itemize}
\item{To each $(d-1)$-manifold $\Sigma$ (e.g. $\partial M$) we associate a \emph{chain complex} $(\calH_\Sigma^\calP,\Omega_\Sigma^\calP)$.}
\item{To each $d$-manifold $M$ we associate a finite-dimensional BV manifold $\calV_M^\calP$, the space of residual fields.
}
\item{To each $d$-manifold $M$ we associate a state $\psi_M$ satisfying the \emph{mQME}.}
\item{Gluing is given by pairing states and performing a \emph{BV pushforward} as in \eqref{Gluing_BV}.
}

\end{itemize}
\end{defn}


\section{Globalization in the BV-BFV Formalism}
In this section we focus on a particular class of topological BV theories arising from the AKSZ construction \cite{AKSZ} (see also \cite{BCM}). One can generalize the construction of the previous section for nonlinear split AKSZ theories by adding a background field (\cite{CMW4,CMW3}).

One can generalize the above construction for nonlinear split AKSZ theories by adding a background field (\cite{CMW4,CMW3}). This gives a formal globalization theory and the gauge fixing is taken care of by the globalized version of the mQME, which is called the modified ``differential'' QME (\cite{CMW4}). This is done by considering techniques of formal geometry (\cite{B,GK}). 
The globalization of Kontsevich's star product and the Poisson Sigma Model in the BV formalism can be found in \cite{BCM,CFT}.
In this subsection we want to briefly recall the main results on globalization in the BV-BFV formalism.

\subsection{AKSZ Sigma Models} 
Let us recall the definition of Hamiltonian manifolds and AKSZ Sigma Models.

\begin{defn}[Differential graded symplectic manifold]
A \emph{dg symplectic manifold} of degree $d$ is a graded manifold $\calM$ endowed with a symplectic form $\omega=\dd \alpha$ of degree $d$ and a Hamiltonian function $\Theta$ of degree $d+1$ satisfying $\{\Theta,\Theta\}=0$, where $\{\enspace,\enspace\}$ is the Poisson bracket induced by $\omega$. 
\end{defn}

\begin{rem}
This is sometimes also called a \emph{Hamiltonian manifold}.
\end{rem}

\begin{defn}[AKSZ Sigma Model]
\label{def:AKSZ}
The \emph{AKSZ Sigma Model} with target a Hamiltonian manifold $(\calM,\omega=\dd\alpha,\Theta)$ of degree $d-1$ is the BV theory, which associates to a $d$-manifold $\Sigma$ the BV manifold $(\calF_\Sigma,\omega_\Sigma,\calS_\Sigma)$, where\footnote{This is the infinite-dimensional graded manifold adjoint to the Cartesian product (internal morphisms).} $\calF_\Sigma=\Map(T[1]\Sigma,\calM)$, $\omega_\Sigma$ is of the form $\omega_\Sigma=\int_\Sigma\omega_{\mu\nu}\delta \mathsf{A}^\mu\land \delta\mathsf{A}^\nu$,
and $\calS_\Sigma[\mathsf{A}]=\int_\Sigma\alpha_\mu(\mathsf{A})\dd \mathsf{A}^\mu+\Theta(\mathsf{A})$, where $\mathsf{A}\in\calF_\Sigma$,
$\omega_{\mu\nu}$ are the components of the symplectic form $\omega$, $\alpha_\mu$ are the components of $\alpha$ and $\mathsf{A}^\mu$ are the components of $\mathsf{A}$ in local coordinates. 
\end{defn}

\begin{rem}
Note that an AKSZ Model is a $BF$-like theory as in \ref{BF_like}.
\end{rem}

\begin{rem}
In the case of AKSZ theories, the procedure of Section \ref{local_field_theory} works particularly well and very naturally \cite{CMR2}. Namely, the BFV structure associated to a boundary $\Sigma_1$ is precisely obtained as in Definition \ref{def:AKSZ} by putting $\Sigma_1$ instead of $\Sigma$. Since the dimension of $\Sigma_1$ is one less than that of $\Sigma$, the resulting symplectic form $\omega_{\Sigma_1}$ will have degree $0$ instead of $-1$. More generally, if $\Sigma_k$ is a corner of codimension $k$, its associated BF$^k$V structure is obtained by putting $\Sigma_k$ instead of $\Sigma$ in Definition \ref{def:AKSZ}. This nice structure corresponds in the BV language to the fully extended topological field theories as in \cite{BD,Lurie2009}.
\end{rem}

\begin{rem}
The AKSZ construction plays a prominent role also in derived algebraic geometry \cite{PTVV}, the generalization of the BV and BFV formalisms to the context of algebraic geometry (typically, however, only the nonpositive part of the complex is retained there, whereas the positive part, which describes the symmetry, is replaced by a stacky description). The extended structure found in \cite{CMR1}, see the above remark, can also be incorporated in this setting \cite{Cal}.
\end{rem}

We denote our fields by $(\mathsf{X},\boldsymbol{\eta})\in\calF_\Sigma$ and, assuming a split of the space of fields $$\calF_\Sigma=\calB_{\partial \Sigma}^\calP\oplus \calV_\Sigma^\calP\oplus \calY'$$ as we have seen before, we denote the according split of the fields by
\begin{equation}
(\mathsf{X},\boldsymbol{\eta})=\underbrace{(\mathbb{X},\E)}_{\in \calB^\calP_{\partial \Sigma}}\oplus \underbrace{(\mathsf{x},\mathsf{e})}_{\in\calV^\calP_\Sigma}\oplus \underbrace{(\mathscr{X},\mathscr{E})}_{\in \calY'}
\end{equation}

In \cite{CMW4}, we study the following type of AKSZ Sigma Models.

\begin{defn}[Split AKSZ Sigma Model]
We call an AKSZ Sigma Model \emph{split}, if the target is of the form 
\begin{equation}
\calM=T^*[d-1]M
\end{equation}
with canonical symplectic structure, where $M$ is a graded manifold.
\end{defn}

\subsection{Formal geometry}
We briefly recall the aspects of formal geometry which are most relevant for the constructions of \cite{CMW4}. Let $M$ be a smooth manifold.

\begin{defn}[Generalized exponential map]
A \emph{generalized exponential map} is a map $\varphi\colon U\to M$, where $U\subset TM$ is an open neighborhood of the zero section, such that $\varphi(x,0)=x$ and $\dd\varphi(x,0)=\id _{T_xM}$ for all $x\in M$.
\end{defn}

\begin{rem}
For $x\in M$ and $y\in T_xM\cap U$ we write $\varphi(x,y)=\varphi_x(y)$.
\end{rem}

\begin{defn}[Formal exponential map]
A \emph{formal exponential map} is an equivalence class of generalized exponential maps, where two generalized exponential maps are said to be equivalent if their $y$-jets agree to all orders.
\end{defn}

For a function $f\in C^\infty(M)$, we can produce a section $\sigma\in\Gamma(\Hat{Sym}T^*M)$ by defining 
\begin{equation}
\sigma_x:=\mathsf{T}\varphi_x^*f,
\end{equation}
where $\mathsf{T}$ denotes the Taylor expansion in the fiber coordinates around $y=0$ and $\Hat{\Sym}$ denotes the completed symmetric algebra. Note that we use any representative of $\varphi$ to define the pullback. We denote this section by $\mathsf{T}\varphi^* f$. Moreover, since it only depends on the jets of the representative, it is independent of the choice of representative.

\begin{defn}[Grothendieck connection]
Given a formal exponential map $\varphi$, we can define the associated \emph{Grothendieck connection} $D_\mathsf{G}$ on $\Hat{\Sym}T^*M$, given by $D_\mathsf{G}=\dd+R$, where $\dd$ is the de Rham differential and $R\in \Gamma(T^*M\otimes TM\otimes \Hat{\Sym}T^*M)$ is a $1$-form with values in derivations of $\Hat{\Sym}T^*M$, defined in local coordinates by $R_i\dd x^{i}$ with 
\begin{equation}
\label{vector_field}
R_i(x;y) = \left(\left(\frac{\partial \varphi_x}{\partial y}\right)^{-1}\right)^k_j\frac{\partial\varphi_x^{j}}{\partial x^{i}}\frac{\de }{\de y^k}=: Y^k_i(x;y)\frac{\de}{\de y^k}.
\end{equation}
\end{defn}

\begin{rem}
One can obtain that \eqref{vector_field} does not depend on the choice of coordinates. Moreover, $D_\mathsf{G}$ is a flat connection and a section $\sigma\in\Gamma(\Hat{\Sym}T^*M)$ is flat if and only if it is of the form $\sigma=\mathsf{T}\varphi^* f$ for some $f\in C^\infty(M)$.
\end{rem}

\subsection{Globalized BV-BFV Quantization}
\label{subsec:globalized_BV-BFV_quantization}
Now one can use the constructions above to formulate a globalized quantum state, which we call the full covariant state as in \cite{CMW4}. For this we need to extend the action by a formal globalization term, where we also lift the fields as the pullback of the formal exponential map at a constant field $x\colon \Sigma \to M$. This corresponds to linearizing the space of fields $\calF_\Sigma$ around these constant maps. 

\begin{defn}[Formal globalized action]
\label{formal_globalized_action}
For $(\mathsf{X},\boldsymbol{\eta})\in\calF_\Sigma$, we define the \emph{formal globalized action} by 
\begin{equation}
\label{formal_glob_action}
\Tilde{\calS}_{\Sigma,x}[(\hatX,\hateta)]=\int_\Sigma\left(\Hat{\boldsymbol{\eta}}_i\dd\Hat{\mathsf{X}}^{i}+\mathsf{T}\Tilde{\varphi}_x^*\Theta(\mathsf{X},\boldsymbol{\eta})+Y_i^j(x;\hatX)\hateta_j\dd x^{i}\right),
\end{equation}
where $\Tilde{\varphi}_x\colon \Map(T[1]\Sigma,T^*[d-1]T_xM)\to \Map(T[1]\Sigma,\calM)$ denotes the lift of the formal exponential map $\varphi_x$ for $x\in \Sigma$ and $(\hatX,\hateta)$ is the preimage of $(\mathsf{X},\boldsymbol{\eta})$ under this lift.
\end{defn}

The Feynman rules corresponding to the formal globalized action as in \eqref{formal_glob_action} are given in Figure \ref{fig:FeynmanRules}. 

\begin{figure}
\centering
\subfigure[Interaction vertex]{
\centering
\begin{tikzpicture}
\node[vertex] (o) at (0,0) {};
\node[coordinate, label=below:{$i_1$}] at (30:1) {$i_1$}
edge[fermion] (o);
\node[coordinate, label=above:{$i_2$}] at (60:1) {$i_2$}
edge[fermion] (o);
\node[coordinate, label=below:{$i_k$}] at (145:1) {$i_k$}
edge[fermion] (o);
\draw[dotted] (90:0.5) arc (90:130:0.5);
\node[coordinate, label=below:{$j_1$}] (j1) at (-30:1) {};
\node[coordinate, label=below:{$j_2$}] (j2) at (-60:1) {};
\node[coordinate, label=below:{$j_l$}] (j3) at (-145:1) {};
\draw[dotted] (-90:0.5) arc (-90:-145:0.5); 
\draw[fermion] (o) -- (j1);
\draw[fermion] (o) -- (j2);
\draw[fermion] (o) -- (j3);
\node[coordinate,label=right:{$\leadsto\quad\Theta^{j_1\ldots j_l}_{i_1 \ldots i_k}(x)$}] at (2,0) {};
\end{tikzpicture}
}
\hspace{2cm}
\subfigure[$R$ vertex]{
\centering
\begin{tikzpicture}
\node[vertex,label=left:{$R$}] (o) at (0,0) {};
\node[coordinate, label=below:{$i_1$}] at (30:1) {$i_1$}
edge[fermion] (o);
\node[coordinate, label=above:{$i_2$}] at (60:1) {$i_2$}
edge[fermion] (o);
\node[coordinate, label=below:{$i_k$}] at (145:1) {$i_k$}
edge[fermion] (o);
\draw[dotted] (90:0.5) arc (90:130:0.5);
\node[coordinate, label=below:{$j$}] (j1) at (-90:1) {};
\draw[fermion] (o) -- (j1);
\node[coordinate,label=right:{$\leadsto\quad Y^{j}_{i;i_1\ldots i_k}(x)\dd x^i$}] at (2,0) {};
\end{tikzpicture}
}

\hspace{2cm}
\subfigure[Residual fields]{
\centering
\begin{tikzpicture}
\node[circle,draw,inner sep=1pt] (x) at (0,1) {$\mathsf{x}^i$};
\node[coordinate, label=below:{$i$}] (x2) at (30:2) {$i_1$}
edge[fermion] (x);
\node[coordinate, label=below:{$j$}] (e2)at (-30:2) {$i_1$};
\node[circle,draw,inner sep=1pt] (e) at (0,-1) {$\mathsf{e}_j$}
edge[fermion] (e2);
\node[coordinate, label={${}$}] at (4,0) {};
\end{tikzpicture}
}
\hspace{3cm}
\subfigure[Boundary vertices]{
\begin{tikzpicture}
\draw (-1,1) -- (1,1);
\node[vertex, label=above:{$\mathbb{X}$}] (x) at (0,1) {};
\node[coordinate] (b1) at (0,0) {}; 
\draw[fermion] (x) -- (b1);
\draw (1,-1) -- (3,-1);
\node[vertex, label=below:{$\mathbb{E}$}] (e) at (2,-1) {};
\node[coordinate] (b2) at (2,0) {}; 
\draw[fermion] (b2) -- (e);

\end{tikzpicture}
\hspace{2cm}
}
\caption{Summary of Feynman graphs and rules}\label{fig:FeynmanRules}
\end{figure}
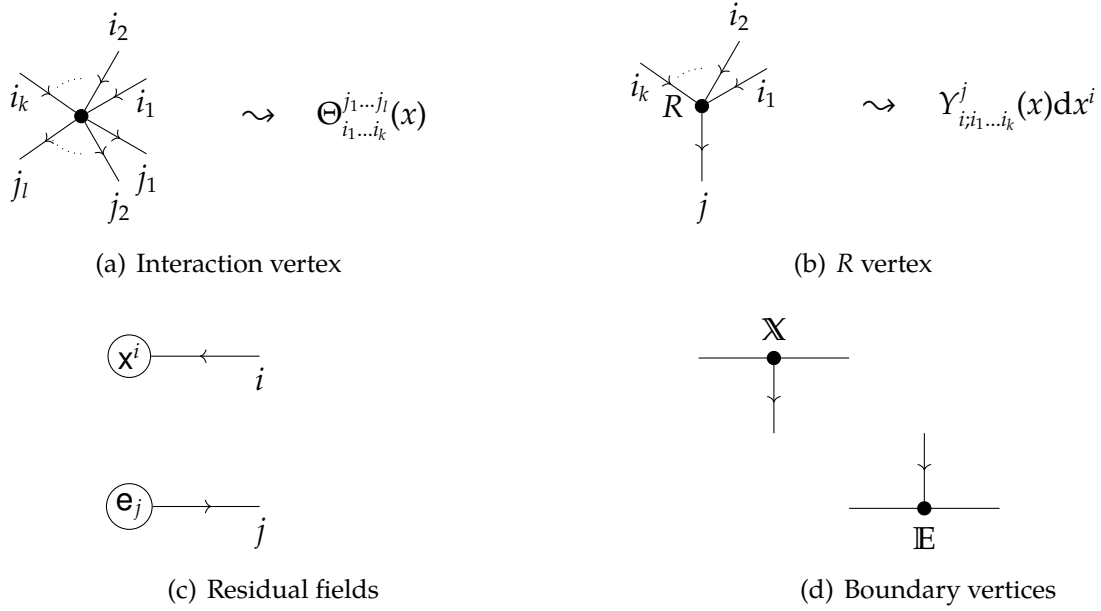

We denote the space of states at any constant field $x$ by $\Hat{\calH}^\calP_{\Sigma,x}:=\calH^\calP_{\partial \Sigma,x}\otimes \Dens^{\frac{1}{2}}(\calV_{\Sigma,x}^\calP)$, where $\Dens^{\frac{1}{2}}(\calV^\calP_{\Sigma,x})$ denotes the space of \emph{half-densities} on $\calV^\calP_{\Sigma,x}$, and view it as the fibers of a vector bundle over $\Sigma$, where the \emph{total space} is defined by $\Hat{\calH}^\calP_{\Sigma,tot}:=\bigsqcup_{x\in\Sigma}\Hat{\calH}^\calP_{\Sigma,x}$

\begin{defn}[Principal covariant state]
The \emph{principal covariant state} is defined by the integral
\begin{equation}
\Tilde{\psi}_{\Sigma,x}:\approx \exp\left(\frac{\I}{\hbar}\sum_\Gamma\hbar^{\ell(\Gamma)}\int_{\mathsf{C}_\Gamma(\Sigma)}\omega_\Gamma(\mathbb{X},\E,\sfx,\sfe,x)\right),
\end{equation}
using the Feynman rules given in Figure \ref{fig:FeynmanRules}
\end{defn}

\begin{rem}
The \emph{principal covariant state} $\Tilde{\psi}_{\Sigma,x}$ can formally be defined as the formal perturbative expansion of the BV integral
\begin{equation}
\Tilde{\psi}_{\Sigma,x}:=\int_{(\mathscr{X},\mathscr{E})\in\calL\subset \calY'}\ee^{\frac{\I}{\hbar}\Tilde{S}_{\Sigma,x}[(\hatX,\hateta)]}\in\Omega^\bullet(M,\Hat{\calH}^\calP_{\Sigma,tot}).
\end{equation}
\end{rem}

The principal state is sufficient for gluing and cutting purposes, but as in the linear case we need to introduce the full state to prove that the \emph{globalized version} of the modified Quantum Master Equation holds. Therefore, we need the concept of \emph{composite fields}, which we denote by square brackets $[\enspace]$, e.g. for a boundary field $\mathbb{A}$ we will write $[\mathbb{A}^{i_1}\dotsm \mathbb{A}^{i_k}]$. They can be understood as a \emph{regularization} of higher functional derivatives: the higher functional derivative $\frac{\delta^k}{\delta\mathbb{A}^{i_1}\dotsm \delta\mathbb{A}^{i_k}}$ gets replaced by a first order functional derivative $\frac{\delta}{\delta[\mathbb{A}^{i_1}\dotsm \mathbb{A}^{i_k}]}$. Concretely, this corresponds to introducing additional boundary vertices as in Figure \ref{fig:composite_field_vertices}.

\begin{figure}[h!]
\subfigure[Boundary vertex on the $\frac{\delta}{\delta\E}$-polarized boundary component of $\de \Sigma$] {
\begin{tikzpicture}
\draw (-3,0) -- (3,0); 
\node[vertex] (o) at (0,0) {};
\node[coordinate, label=below:{$[\mathbb{X}^{i_1}\cdots\mathbb{X}^{i_k}]$}] at (o.south) {};
\node[coordinate, label=below:{$i_1$}] (b1) at (30:1) {$i_1$};
\node[coordinate, label=above:{$i_2$}] (b2) at (60:1) {$i_2$};
\node[coordinate, label=below:{$i_k$}] (b3) at (145:1) {$i_k$};
\draw[dotted] (90:0.5) arc (90:130:0.5); 
\draw[fermion] (o) -- (b1);
\draw[fermion] (o) -- (b2);
\draw[fermion] (o) -- (b3);
\end{tikzpicture}
}
\hspace{2cm}
\subfigure[Boundary vertex on the $\frac{\delta}{\delta\mathbb{X}}$-polarized boundary component of $\de \Sigma$] {
\begin{tikzpicture}
\draw (-3,0) -- (3,0); 
\node[vertex] (o) at (0,0) {};
\node[coordinate, label=below:{$[\mathbb{E}_{i_1}\cdots\mathbb{E}_{i_k}]$}] at (o.south) {};
\node[coordinate, label=below:{$i_1$}] at (30:1) {$i_1$}
edge[fermion] (o);
\node[coordinate, label=above:{$i_2$}] at (60:1) {$i_2$}
edge[fermion] (o);
\node[coordinate, label=below:{$i_k$}] at (145:1) {$i_k$}
edge[fermion] (o);
\draw[dotted] (90:0.5) arc (90:130:0.5);
\end{tikzpicture}
}
\caption{Composite field vertices.}\label{fig:composite_field_vertices}
\end{figure}
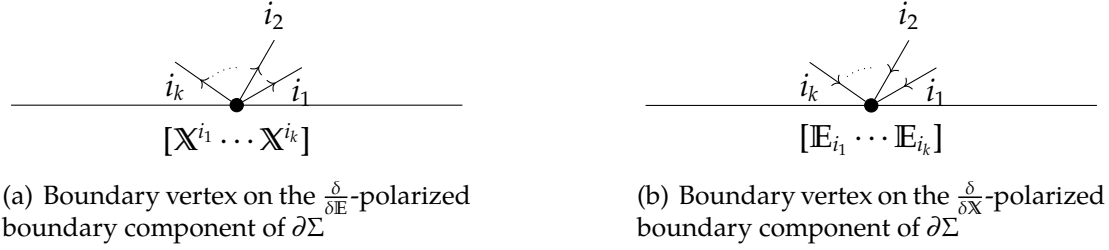 

\begin{defn}[Full covariant state]
We define the \emph{full covariant state} $\btpsi_{\Sigma,x}$ by 
\begin{equation}
\btpsi_{\Sigma,x}:\approx \exp\left(\frac{\I}{\hbar}\sum_\Gamma\hbar^{\ell(\Gamma)}\int_{\mathsf{C}_\Gamma(\Sigma)}\omega_\Gamma(\mathbb{X},\E,\sfx,\sfe,x)\right)
\end{equation}
using the Feynman rules in Figure \ref{fig:FeynmanRules} and additionally with the rules for the boundary vertices as in Figure \ref{fig:composite_field_vertices}. 
\end{defn}

\begin{rem}
Similarly we have to define the BFV boundary operator with the additional Feynman rules. We call it the \emph{full} BFV boundary operator and denote it by $\boldsymbol{\Omega}_{\de \Sigma}^\calP$. 
\end{rem}

One of the main result of \cite{CMW4} is that this state statisfies the globalized version of the modified Quantum Master Equation, which we call the modified differential Quantum Master Equation (mdQME). It is stated as the following theorem.

\begin{thm}[mdQME for split AKSZ theories \cite{CMW4}]
Consider the full covariant perturbative state $\btpsi_{\Sigma,x}$ as a quantization of an anomaly free split AKSZ theory with target $T^*[d-1]M$, where $M$ is a graded manifold. Then
\begin{equation}
\label{AKSZ_mdQME}
\left(\dr_x -\I\hbar \Delta_{\calV^\calP_{\Sigma,x}} + \frac{\I}{\hbar} \boldsymbol{\Omega}^\calP_{\partial\Sigma}\right) \btpsi_{\Sigma,x}=0,
\end{equation}
where we denote by $\dr_x$ the de Rham differential on $\Bar{M}$, the body of the graded manifold $M$.
\end{thm}

Note that the operator $\nabla_\mathsf{G}:=\left(\dr_x -\I\hbar \Delta_{\calV^\calP_{\Sigma,x}} + \frac{\I}{\hbar} \boldsymbol{\Omega}_{\partial\Sigma}^\calP\right)$ can be seen as a connection on the total state space $\Hat{\calH}^\calP_{\Sigma,tot}$. Influenced from the classical case, we call $\nabla_\mathsf{G}$ the \emph{quantum Grothendieck BFV (GBFV) operator}. Another main result of \cite{CMW4} is the following theorem.

\begin{thm}[\cite{CMW4}]
The connection $\nabla_\mathsf{G}$ is flat, i.e.
\begin{equation}
(\nabla_\mathsf{G})^2\equiv 0.
\end{equation}
\end{thm}

\begin{rem}
This construction depends on different choices, such as the propagator, the choice of residual fields and the formal exponential map. In \cite{CMW4} it was shown how the full covariant state behaves under change of data. In particular, it transforms similar to a gauge transformation.
\end{rem}

\subsubsection{Obstruction theory for the Poisson Sigma Model}
Note that the construction above is only true for manifolds with a single boundary polarization. Since the Poisson Sigma Model is a topological field theory of AKSZ type, we are able to use the mdQME for a condition to gauge fix the globalized version of the Poisson Sigma Model, i.e. the model which is related to a globalized version of Kontsevich's star product (\cite{K}). This construction (see \cite{CMW2} for a detailed description for the Moyal product) uses the notion of a \emph{relational symplectic groupoid} (\cite{CC2}) and the gluing procedure in the BV-BFV formalism.

Unfortunately, we have to use a disk with alternating boundary conditions and boundary components without any polarization, so the construction above does not simply apply. We call the points on the boundary where the polarization changes ``corners'' though they are not corners in the usual sense of manifolds.

\begin{figure}[h!]
\centering
\tikzset{
particle/.style={thick,draw=black},
particle2/.style={thick,draw=blue},
avector/.style={thick,draw=black, postaction={decorate},
    decoration={markings,mark=at position 1 with {\arrow[black]{triangle 45}}}},
gluon/.style={decorate, draw=black,
    decoration={coil,aspect=0}}
 }
\begin{tikzpicture}[x=0.04\textwidth, y=0.04\textwidth]
\draw[particle] (0,0)--(2,5) node[above]{};
\draw[particle2] (2,5.5)--(4,5.5)node[above]{};
\draw[particle2] (2,5)--(4,5) node[right]{};
\draw[particle2] (0,0)--(2,0) node[right]{};
\draw[particle2] (0,-0.5)--(2,-0.5)node[below]{};
\draw[particle2] (4,-0.5)--(6,-0.5)node[below]{};
\node[](u1) at (2,0){};
\node[](u2) at (4,0){};
\draw[particle2] (4,0)--(6,0) node[right]{};
\draw[particle] (6,0)--(4,5) node[right]{};
\node[](2) at (0.7,0){};
\node[](3) at (1.4,0){};
\node[](11) at (2,5.5){};
\node[](12) at (4,5.5){};
\node[](13) at (0,-0.5){};
\node[](14) at (2,-0.5){};
\node[](15) at (4,-0.5){};
\node[](16) at (6,-0.5){};
\node[](p1) at (3,6.5){};
\node[](p2) at (1,-1.5){};
\node[](p3) at (5,-1.5){};
\node[](p4) at (3,7){$\delta_{\tilde x}$};
\draw[fill=black] (p1) circle (0.05cm);
\draw[fill=black] (p2) circle (0.05cm);
\draw[fill=black] (p3) circle (0.05cm);
\node[](f) at (1,-2){$f$};
\node[](g) at (5,-2){$g$};
\node[](x13) at (5,-0.75){$\mathbb{X}$};
\node[](x12) at (1,-0.75){$\mathbb{X}$};
\node[](eta11) at (3,5.77){$\E$};
\node[](eta1) at (0,2.5){$\Hat{\boldsymbol{\eta}}=0$};
\node[](eta2) at (6,2.5){$\Hat{\boldsymbol{\eta}}=0$};
\node[](eta3) at (3,0){$\Hat{\boldsymbol{\eta}}=0$};
\node[](e1) at (1,0.3){$\E$};
\node[](e2) at (5,0.3){$\E$};
\node[](e3) at (3,4.7){$\mathbb{X}$};
\semiloop[particle]{u1}{u2}{0};
\semiloop[particle]{11}{12}{0};
\semiloop[particle]{14}{13}{180};
\semiloop[particle]{16}{15}{180};
\end{tikzpicture}
\caption{The gluing for the star product for two smooth maps $f$ and $g$. This will produce the star product $f\star g$. Here $\delta_{\tilde x}$ is a delta function $x\mapsto \delta(\tilde x-x)$.}
\label{gluing_PSM}
\end{figure}
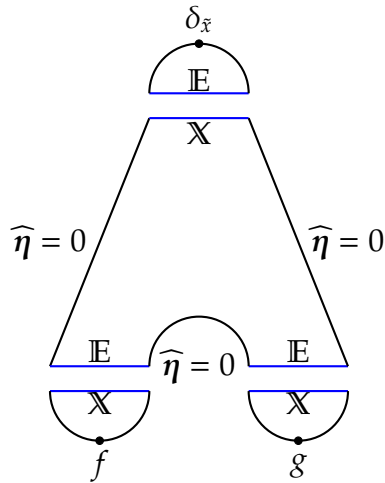

To produce Kontsevich's star product, the boundary of the disk is going to be glued together as in Figure \ref{gluing_PSM}, where the black boundary components are not polarized, but rather carry another boundary condition, namely that $\Hat{\boldsymbol{\eta}}$ vanishes. Thus there will be graphs collapsing to the boundary of the manifold in the boundary of the compactified configuration space, which cannot be taken care of by $\Omega^\calP_{\partial\Sigma}$.  

However, in \cite{CMW3} it was shown that there is way out of this by introducing another term for the action, and using Kontsevich's formality map and the concept of Fedosov quantization. This will lead to new terms appearing in the mdQME, which kill the extra terms such that the mdQME, and thus gauge independence, is indeed satisfied. This is called the \emph{twisted theory} (\cite{CMW3}).

\printbibliography

\end{document}